\documentclass[a4paper, UKenglish, cleveref]{lipics-v2021}
\nolinenumbers

\EventEditors{Jonathan Aldrich and Alexandra Silva}
\EventNoEds{2}
\EventLongTitle{39th European Conference on Object-Oriented Programming (ECOOP 2025)}
\EventShortTitle{ECOOP 2025}
\EventAcronym{ECOOP}
\EventYear{2025}
\EventDate{June 30--July 2, 2025}
\EventLocation{Bergen, Norway}
\EventLogo{}
\SeriesVolume{333}
\ArticleNo{19}

\title{The Algebra of Patterns (Extended Version)}

\author{David Binder}{University of Kent, United Kingdom}{D.Binder@kent.ac.uk}{0000-0003-1272-0972}{supported by ARIA programme for Safeguarded AI (grant MSAI-PR01-P11).}
\author{Lean Ermantraut}{Radboud University Nijmegen, Netherlands}{lean.ermantraut@ru.nl}{0000-0003-2609-9259}{}
\authorrunning{David Binder and Lean Ermantraut}
\Copyright{David Binder and Lean Ermantraut}

\usepackage{booktabs}
\usepackage[all,cmtip]{xy}
\usepackage{csquotes}
\usepackage{booktabs}
\usepackage{stmaryrd}
\usepackage{mathtools}
\usepackage{courier}
\usepackage{bussproofs}
\usepackage{tikz-cd}
\usepackage{diagbox}
\usepackage{afterpage}
\usepackage[normalem]{ulem}
\usepackage{dashbox}
\usepackage{mdframed}
\usepackage{colortbl}
\usepackage{anyfontsize}
\usepackage{listings}
\usepackage{alltt} 
\usepackage{multicol} 
\usepackage{scalerel} 
\usepackage{syntax}

\newcommand{\patWildcard}{\_}
\newcommand{\patOr}[2]{#1 \mathbin{\parallel} #2}
\newcommand{\patAnd}[2]{#1 \mathbin{\&} #2}
\newcommand{\patNeg}[1]{\neg#1}
\newcommand{\patTop}{\patWildcard}
\newcommand{\patBot}{\#}

\newcommand{\patLeftOr}[2]{#1 \parallel^{\leftarrow} #2}

\newcommand{\caseof}[2]{\mathbf{case}\ #1\ \mathbf{of}\ \{\ #2\ \}}
\newcommand{\defaultClause}{\mathbf{default}}

\newcommand{\tmRed}{\mathtt{Red}}
\newcommand{\tmGreen}{\mathtt{Green}}
\newcommand{\tmBlue}{\mathtt{Blue}}
\newcommand{\tmTrue}{\mathtt{True}}
\newcommand{\tmFalse}{\mathtt{False}}
\newcommand{\tmNil}{\mathtt{Nil}}
\newcommand{\tmCons}[2]{\mathtt{Cons}(#1,#2)}
\newcommand{\tmPair}[2]{\mathtt{Pair}(#1,#2)}
\newcommand{\tmInl}[1]{\iota_1(#1)}
\newcommand{\tmInr}[1]{\iota_2(#1)}
\newcommand{\tmMo}{\mathtt{Mo}}
\newcommand{\tmTu}{\mathtt{Tu}}
\newcommand{\tmWe}{\mathtt{We}}
\newcommand{\tmTh}{\mathtt{Th}}
\newcommand{\tmFr}{\mathtt{Fr}}
\newcommand{\tmSa}{\mathtt{Sa}}
\newcommand{\tmSu}{\mathtt{Su}}
\newcommand{\tmAdmin}{\mathtt{Admin}}
\newcommand{\tmRegisteredUser}{\mathtt{RegisteredUser}}
\newcommand{\tmGuest}{\mathtt{Guest}}

\newcommand{\tyBool}{\mathbb{B}}
\newcommand{\tyPair}[2]{#1 \otimes #2}
\newcommand{\tySum}[2]{#1 \oplus #2}

\newcommand{\matches}[3]{\ensuremath{#1 \triangleright #2 \leadsto #3}}
\newcommand{\matchesNot}[3]{\ensuremath{#1 \not\triangleright #2 \leadsto #3}}

\newcommand{\singlestep}{\rightarrow}
\newcommand{\multistep}{\rightarrow^{\ast}}

\newcommand{\semantic}[1]{\ensuremath{\llbracket\, #1\, \rrbracket}}

\newcommand{\overlaps}[2]{\ensuremath{#1 \between #2}}
\newcommand{\disjointSymbol}{\between\hspace{-0.3cm}\diagup}
\newcommand{\disjoint}[2]{\ensuremath{#1 \mathbin{\disjointSymbol} #2}}

\newcommand{\linearpos}[1]{#1\ \mathbf{lin}^+}
\newcommand{\linearneg}[1]{#1\ \mathbf{lin}^-}
\newcommand{\linearposneg}[1]{#1\ \mathbf{lin}^\pm}
\newcommand{\deterministic}[1]{#1\ \mathbf{det}}
\newcommand{\freevareven}[1]{\ensuremath{\mathrm{FV}^{e}(#1)}}
\newcommand{\freevarodd}[1]{\ensuremath{\mathrm{FV}^{o}(#1)}}
\newcommand{\wf}[1]{#1\ \checkmark}

\newcommand{\patTyping}[4]{\ensuremath{#1 ; #2 \Rightarrow #3 : #4 }}

\newcommand{\nnf}[1]{\mathtt{nnf}(#1)}
\newcommand{\nnfpos}[1]{\mathtt{nnf}^{+}(#1)}
\newcommand{\nnfneg}[1]{\mathtt{nnf}^{-}(#1)}

\newcommand{\dnf}[1]{\mathtt{dnf}(#1)}

\newcommand{\normalize}[1]{\mathtt{norm}(#1)}
\newcommand{\combine}[2]{\mathtt{combine}(#1,#2)}

\newcommand{\headconstructors}[1]{\mathtt{head}(#1)}
\newcommand{\To}{\Rightarrow}

\newcommand{\mdoubleplus}{\mathbin{+\mkern-5mu+}}

\newcommand{\rocq}{\includegraphics[scale=0.024]{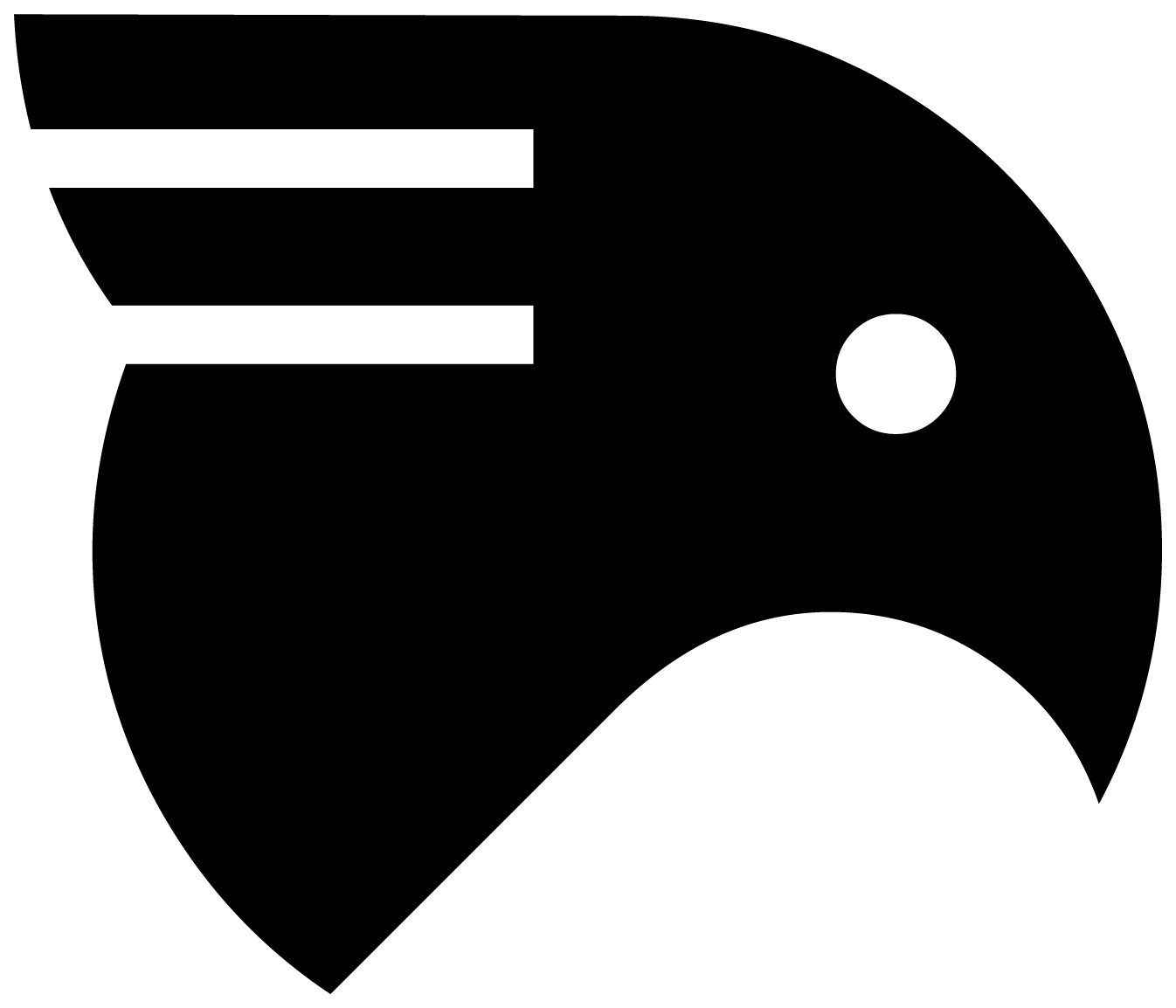}}

\usepackage{xcolor}



\begin{document}

\maketitle

\begin{abstract}
    Pattern matching is a popular feature in functional, imperative and object-oriented programming languages.
    Language designers should therefore invest effort in a good design for pattern matching.
    Most languages choose a \emph{first-match semantics} for pattern matching; that is, clauses are tried in the order in which they appear in the program until the first one matches.
    As a consequence, the order in which the clauses appear cannot be arbitrarily changed, which results in a less declarative programming model.
    The declarative alternative to this is an \emph{order-independent semantics} for pattern matching, which is not implemented in most programming languages since it requires more verbose patterns.
    The reason for this verbosity is that the syntax of patterns is usually not expressive enough to express the \emph{complement} of a pattern.
    In this paper, we show a principled way to make order-independent pattern matching practical.
    Our solution consists of two parts:
    First, we introduce a \emph{boolean algebra of patterns} which can express the complement of a pattern.
    Second, we introduce \emph{default clauses} to pattern matches.
    These default clauses capture the essential idea of a fallthrough case without sacrificing the property of order-independence.
\end{abstract}

\section{Introduction}
\label{sec:intro}
Pattern matching is a programming language feature that combines two operations:
Testing whether a value matches a given pattern and, if the test is successful, extracting subexpressions of the value and binding them to variables.
A pattern matching expression consists of the \emph{scrutinee} which is analyzed by the pattern match and a list of \emph{clauses}.
Each clause consists of a \emph{pattern}, hence the name, and an expression called the \emph{right-hand side}.
These elements can all be seen in the following example of the length function, which computes the length of a list:
\begin{equation*}
    \text{length}(\mathit{xs}) \coloneqq \caseof{\underbracket{\quad \mathit{xs}\quad}_{\text{scrutinee}}}{ \underbracket{\,\tmNil \Rightarrow 0\,}_{\text{clause}}, \underbracket{\,\tmCons{\patWildcard}{\mathit{zs}} \,}_{\text{pattern}} \Rightarrow \underbracket{1 + \text{length}(\mathit{zs})}_{\text{right-hand side}}\ }
\end{equation*}

Programming with pattern matching is popular.
For a long time, only programmers who use functional languages could enjoy this feature, but today pattern matching has arrived in the mainstream of object-oriented and imperative programming languages.
For example, Python, Java, Dart, and C\# added pattern matching as an extension to the language, and some newer languages such as Rust and Swift, which are not primarily functional languages, supported pattern matching from the beginning.

Pattern matching has become one of the central constructs that is used to structure the control and data flow of a program.
How pattern matches are checked and evaluated therefore fundamentally determines how we reason about the meaning of our programs;
both informally when we read and write code in our daily lives, and formally in proof assistants, static analyzers and other verification tools.
It is therefore important to choose a semantics for pattern matching which makes reasoning as simple as possible.
One choice that language designers have to make is whether the clauses are tried in the order in which they appear in the program, or whether the order of the clauses does not matter.
The following definition captures this difference:
\begin{definition}[First-Match and Order-Independent Semantics]
    A language uses \emph{first-match semantics} if the clauses of a pattern match are tried in the order in which they occur until the first pattern matches the scrutinee.
    A language uses \emph{order-independent semantics} if any of the matching clauses in a pattern match can be used to non-deterministically reduce a pattern matching expression.
\end{definition}

Overlapping clauses, i.e. two clauses with patterns that can match the same scrutinee, are a problem for order-independent semantics since a pattern-matching expression with overlapping clauses can non-deterministically reduce to two different results.
The obvious solution is to enforce that two clauses must never overlap.
But if we adopt this solution, then programmers usually have to write much more verbose patterns.
For example, instead of the following function which checks whether the argument is the color $\mathtt{Red}$
\begin{equation*}
    \text{isRed}(c) \coloneqq \caseof{c}{
        \tmRed \Rightarrow \tmTrue,\
        \patWildcard \Rightarrow \tmFalse
        },
\end{equation*}%
the programmer has to write the following more verbose expression
\begin{equation*}
    \text{isRed}(c) \coloneqq \caseof{c}{
        \tmRed \Rightarrow \tmTrue,\
        \tmGreen \Rightarrow \tmFalse,\
        \tmBlue \Rightarrow \tmFalse
        }.
\end{equation*}%
The source of this verbosity is the poor expressiveness of patterns, which usually do not allow to express the \emph{complement} of a pattern without exhaustively enumerating all the alternatives.
In this example, we had to explicitly mention both $\tmBlue$ and $\tmGreen$ to express the complement of the pattern $\tmRed$.
In languages that support or-patterns we can combine the last two clauses in the single clause $(\patOr{\tmGreen}{\tmBlue}) \Rightarrow \tmFalse$, but we still have to mention all constructors explicitly.
It is unreasonable to ask programmers to always list all constructors of a data type, and they would rightfully reject a programming language which enforces this.
Every proposal to use order-independent pattern matching in a realistic programming language therefore has to address this verbosity problem.

We propose two complementary solutions to the verbosity problem.
In our system, the function can be written in the following ways:
\begin{align}
    \tag{Negation Pattern}
    \text{isRed}(c) &\coloneqq \caseof{c}{
        \tmRed \Rightarrow \tmTrue,\
        \patNeg{\tmRed} \Rightarrow \tmFalse } \\
    \tag{Default Clause}
    \text{isRed}(c) &\coloneqq \caseof{c}{
        \tmRed \Rightarrow \tmTrue,\
        \defaultClause \Rightarrow \tmFalse }
\end{align}

To write the first variant we enrich the syntax of patterns with \emph{negation patterns} which match whenever the negated pattern does not match.
The two patterns are no longer overlapping, and at the same time, the pattern matching expression is not more verbose than our first version.
Such negation patterns were already proposed under the name of \enquote{anti-patterns} \cite{Kirchner2007,Kirchner2008,Kirchner2010} and are also mentioned by some other authors \cite{Krishnaswami2009, GellerHirschfeldBracha2010}; the core contribution of this paper, however, is to consider them in the context of a complete \emph{boolean algebra of patterns} which allows us to apply algebraic reasoning to patterns.

To write the second variant we enrich the syntax of pattern-matching expressions with \emph{default clauses}\footnote{%
    Some programming languages support default clauses with similar semantics to what we propose.
    Ada, for example, enforces that clauses in a case statement are not overlapping, but allows an optional \enquote{other} clause with the same semantics as our default clause.
    The clauses that Ada allows are however much more restricted, since they only work for enums instead of algebraic data types.
}.
A pattern matching expression can contain at most one such default clause (where $\mathbf{default}$ is a keyword instead of a pattern) which matches whenever none of the normal clauses matches the scrutinee.
At first glance, the meaning of a default clause looks identical to the meaning of a clause with a wildcard pattern, but this is only the case for first-match semantics.
If we use order-independent semantics, then we see that the meaning of a default clause does not correspond to a clause with a wildcard pattern, but to a clause that is headed by the negation of the disjunction of all other patterns in the pattern matching expression.
In the compilation algorithm, however, default clauses are treated natively, so that we don't have to expand a default clause to a large pattern in the case of pattern matches with many clauses.

We are particularly interested in the kind of reasoning which is preserved when we add or remove clauses from a pattern match, or when we add or remove constructors from a data type.
We will argue that this is a significant software engineering challenge that most pattern-based languages do not adequately address yet.

\subsection{Algebraic Reasoning and Variable Binding}

The main contribution of this work is to support a rich catalog of equivalences for patterns.
In fact, we strive to treat patterns as a boolean algebra which consists of negation patterns, and-patterns, or-patterns, wildcard-patterns and absurd-patterns, featuring all the expected equivalences.
This is challenging from the perspective of variable binding.
For instance, we want the law of double negation, i.e.~$\semantic{\patNeg{\patNeg{p}}} \equiv \semantic{p}$, to hold, but this entails that we cannot ignore variables that occur inside patterns in a negated context.
We will demonstrate that we can achieve the desired equivalences by a variable binding definition that tracks whether variables occur under an even or odd number of negations.
A related problem is to have well-defined variable binding in the context of or-patterns $\patOr{p}{p'}$, which requires a careful definition of linearity to make sure that the same variables are bound in all cases of an or-pattern.
We also describe the conditions under which we can replace a pattern
$\patOr{p_1}{p_2}$ by the pattern $\patOr{p_2}{p_1}$; this transformation is not semantic preserving in most systems which implement or-patterns.

\subsection{Overview and Contributions}

The rest of this article is structured as follows:
\begin{itemize}
    \item In \cref{sec:motivation} we motivate why the first-match semantics that most programming languages use for pattern matching fails to address several important reasoning and software engineering problems.
    We show how our contributions can improve upon this status quo.
    \item In \cref{sec:pattern-algebra} we present our first main contribution: the algebra of patterns.
    We describe the syntax of patterns, the rules for non-deterministically matching a pattern against a value, and the restrictions that characterize the subset of linear and deterministic patterns. We also prove the algebraic properties that characterize the semantic equivalence relation on patterns.
    \item \cref{sec:terms} introduces our second contribution, a small term language that contains pattern-matching expressions with default clauses.
    This section also introduces exemplary typing rules to the untyped patterns of \cref{sec:pattern-algebra} and the constructs of our term language; the compilation of patterns described in later sections, however, does not depend on patterns being typed.
    \item Not all extensions to pattern matching can be compiled to efficient code, but ours can. In \cref{sec:normalization} we show how to compile patterns by translating them to a variant of a disjunctive normal form. We use these normalized patterns in the compilation of pattern-matching expressions to decision trees, which we describe in  \cref{sec:compilation}.
    \item We discuss future work in \cref{sec:future-work}, related work in \cref{sec:related-work} and conclude in \cref{sec:conclusion}.
\end{itemize}

The results in sections \cref{sec:pattern-algebra} have been formally verified in the proof assistant Rocq, and are made available as complementary material.
Theorems and definitions which have been formalized are marked with the symbol \rocq.

\section{Motivation}
\label{sec:motivation}
Most programming languages that implement pattern matching use a first-match semantics for which the order of clauses matters.
Changing this state of affairs is our main motivation for introducing a more expressive set of algebraic patterns.
Let us therefore explain why we think that first-match semantics is not a good default choice for declarative languages.

\subsection{First-Match Semantics Weakens Equational Reasoning}
\label{subsec:problems:equational-reasoning}

Equational reasoning is one of the most important methods functional programmers use to reason about code.
Equational reasoning allows us to show that two expressions are equivalent if we can obtain one from the other by a series of rewriting steps.
Each rewriting step is justified by some equation between terms, and functions defined by pattern matching are one of the main sources of such equations.
Take, for example, the definitions of the functions \enquote{id} and \enquote{map}:
\begin{align*}
    \tag{1}
    &\text{id}(x) \coloneqq x \\
    \tag{2}
    &\text{map}(f, []) \coloneqq []\\
    \tag{3}
    &\text{map}(f, x :: xs) \coloneqq f(x) :: \text{map}(f, xs)
\end{align*}
We can show that mapping the identity function over a list returns the original list, i.e. that $\text{map}(\text{id}, xs) = xs$ holds, by using the individual clauses of the definition of ``map'':
\begin{align*}
    \text{map}(\text{id},[]) &=_{(2)} [] \\
    \text{map}(\text{id},x :: xs) &=_{(3)} \text{id}(x) :: \text{map}(\text{id}, xs)
                                   =_{(1)} x :: \text{map}(\text{id}, xs)
                                   =_{(IH)} x :: xs
\end{align*}
We have annotated each rewriting step with the equation that justifies it, and in the last step, we have used the induction hypothesis (IH).
In this example, we have used the fact that every clause of the pattern match is a valid equation that we can use for rewriting.
But this was only valid because the equations in the definition of \enquote{map} do not overlap.
To see why this is essential, consider the example from the introduction again:
\begin{align*}
    \tag{4}
    &\text{isRed}(\tmRed) \coloneqq \tmTrue\\
    \tag{5}
    &\text{isRed}(\patWildcard) \coloneqq \tmFalse
\end{align*}
If we are using these clauses as equations we can now show that True is equal to False, using the proof $\tmTrue =_{(4)} \text{isRed}(\tmRed) =_{(5)} \tmFalse$.
A human is of course very unlikely to make such a silly mistake, but if we cannot guarantee that every clause holds as a valid equality between expressions, then we also limit what automatic tools can do for us.

\subsection{First-Match Semantics Complicates Reasoning About Change}
\label{subsec:problems:open-closed-worlds}

Software engineering has been described as the problem of integrating programming over time \cite{Winters2020google}.
This means that we not only have to understand the meaning of a program at a fixed point in time, but we also have to understand how the meaning of a program changes as it is developed and maintained.
It should therefore be as simple as possible to reason about how the meaning of a program changes when a programmer adds a clause to a pattern match, removes a clause from a pattern match, adds a constructor to a data type, or removes a constructor from a data type.
Let us see why first-match semantics complicates reasoning about these kinds of changes.

If we want to understand the consequences of adding or removing a clause from a pattern match, then we have to consider all clauses occurring both above and below the clause we are changing.
Having to reason about the entire context of the clause makes it much harder to spot bugs in code reviews, since it is not uncommon that pattern matching expressions can span multiple pages of code.
Order-independent semantics guarantees that it doesn't matter if we add, delete or modify a clause at the beginning, middle or end of a list of clauses.

The other kind of change that frequently occurs is that we add or remove a constructor from an existing data type.
Whether we have to adjust existing pattern matches depends on how those have been written.
Even if we enforce that clauses in a pattern matching expression must not overlap, we have multiple possibilities to write exhaustive pattern matches.
Consider a data type \enquote{Group} which consists of administrators, registered users and guests.
We want to guarantee that only administrators have write access, but we have two possibilities to write the function:
\begin{equation*}
  \begin{array}{ll}
    \text{hasWriteAccess} : \text{Group} \to \text{Bool} &\text{hasWriteAccess} : \text{Group} \to \text{Bool} \\
    \text{hasWriteAccess}(\tmAdmin) \coloneqq \tmTrue &\text{hasWriteAccess}(\tmAdmin) \coloneqq \tmTrue \\
    \text{hasWriteAccess}(\patOr{\tmRegisteredUser}{\tmGuest}) \coloneqq \tmFalse &\text{hasWriteAccess}(\patNeg{\tmAdmin}) \coloneqq \tmFalse
  \end{array}
\end{equation*}
On the left side we have pattern matched exhaustively on the complement of \texttt{Admin}, whereas on the right side we have used a negation pattern.
These two programs behave the same, but we can observe a difference when we add moderators as a fourth type of user.
We have to revisit the function on the left, whereas the function on the right continues to compile.

Some programming languages have already started to add support for enforcing these kind of considerations.
For example, Rust supports the \texttt{\#[non_exhaustive]} attribute on type declarations\footnote{Cp. \href{https://doc.rust-lang.org/stable/reference/attributes/type_system.html}{doc.rust-lang.org/stable/reference/attributes/type_system.html}.}.
This attribute restricts pattern matches and ensures that they continue to compile after a new constructor has been added to the type; annotating the Group type with such an attribute would disallow the definition on the left.
Dually, OCaml warns against writing \enquote{fragile pattern matches}.
A fragile pattern can hide newly added constructors, which might have unintended consequences.
The algebraic patterns presented in this paper provide the necessary tools to program in situations where the programmer wants to enforce either one of these dual restrictions.
There is also some empirical evidence that programmers feel uneasy about using wildcard patterns in evolving codebases.
For example, the authors of a recent study about the programming practices of functional programmers \cite{Lubin2021} write:
\enquote{Participants felt similarly about the use of wildcards in pattern matching, which silently assign pre-existing behavior to new variants of an enumeration. P7 mentioned that, in their main codebase, wildcards are completely disallowed for this reason, even though they make programming more convenient.}

\subsection{Compositional Pattern Matching and Extensible Data Types}
\label{subsec:problems:compositional-pm}

Our main motivation in this paper is to obtain a more declarative programming model for standard nominal data types and ordinary pattern matching.
The compilation algorithm that we are presenting, however, also works for polymorphic variants \cite{Garrigue1998polymorphic}.
Polymorphic variants are used in the programming language OCaml, for example, to handle open sets of possible error conditions without having to declare a data type which lists all possible errors.
Our algorithm does not presuppose a data type declaration with a fixed set of constructors which is necessary to check for exhaustiveness; we do not discuss exhaustiveness of patterns in the main part of this paper, but outline the necessary modifications to check for it in \cref{sec:appendix-exhaustiveness}.
We think that the ideas and the algebraic patterns that we present in this paper are also useful for some proposed solutions of the expression problem.
In the compositional programming approach \cite{Rioux2023,Zhang2021,Zhang2020}, for example, we can create pattern matches whose clauses are distributed among multiple different modules.
In that scenario it is especially important to guarantee that the contributing clauses which come from different modules do not overlap in their left-hand side.

\section{The Algebra of Patterns}
\label{sec:pattern-algebra}
As our first contribution, we introduce an \emph{algebra of patterns}.
The syntax of patterns and values is specified in \cref{fig:pattern-algebra:syntax}.
We assume that we have an infinite set \textsc{CtorNames} of constructor names such as $\tmTrue$, $\tmFalse$ or $\tmNil$, and a set \textsc{Variables} which contains variables such as $x,y$ or $z$.
Each constructor $\mathcal{C}$ comes with an implicit \emph{arity} $n$, and we sometimes write $\mathcal{C}^n$ to make this explicit.
For example, the constructor $\mathtt{Cons}$ has arity 2 since it needs to be applied to two arguments.
Patterns $p$ consist of variable patterns $x$, constructor patterns $\mathcal{C}^n(p_1,\ldots,p_n)$, and-patterns $\patAnd{p}{p}$, or-patterns $\patOr{p}{p}$, wildcard-patterns $\patWildcard$, absurd-patterns $\patBot$, and negation patterns $\patNeg{p}$.
Values consist of constructors applied to other values; the two-element list $\tmCons{\tmTrue}{\tmCons{\tmFalse}{\tmNil}}$ is an example of a value.
If a pattern $p$ matches a value $v$ we obtain a substitution $\sigma$.
Such substitutions are formalized as lists of mappings $x \mapsto v$, and we write $\sigma_1 \mdoubleplus \sigma_2$ for the concatenation of two such lists.
Not every list of mappings is a proper map since some lists contain multiple mappings for the same variable, but we will later prove that patterns from the subset of linear patterns produce substitutions that map a variable to at most one value.
There are two functions that compute the free variables contained in a pattern.
The function $\freevareven{p}$ computes the free variables that occur under an \emph{even} number of negations, whereas $\freevarodd{p}$ computes the variables under an \emph{odd} number of negations.

\begin{figure}
  \[
  \begin{array}{lclr}
      \multicolumn{4}{c}{\mathcal{C},\mathcal{C}^n \in \textsc{CtorNames}\quad x,y,z \in \textsc{Variables}} \\
       p & \Coloneqq & x \mid \mathcal{C}^n(p_1, \ldots, p_n) \mid \patAnd{p}{p} \mid \patOr{p}{p} \mid \patTop \mid \patBot \mid \patNeg{p} & \emph{Patterns} \\
       v & \Coloneqq & \mathcal{C}^n(v_1,\ldots,v_n) & \emph{Values} \\
       m & \Coloneqq & x \mapsto v & \emph{Mapping} \\
       \sigma & \Coloneqq & [] \mid m :: \sigma & \emph{Substitution} \\
  \end{array}
  \]
  \[
  \begin{array}{rlrl}
    \freevareven{x} & \coloneqq  \{ x \} & \freevarodd{x} & \coloneqq \emptyset \\
    \freevareven{\patWildcard} & \coloneqq  \emptyset & \freevarodd{\patWildcard} & \coloneqq \emptyset \\
    \freevareven{\patBot} & \coloneqq  \emptyset & \freevarodd{\patBot} & \coloneqq \emptyset \\
    \freevareven{\patNeg{p}} & \coloneqq  \freevarodd{p} & \freevarodd{\patNeg{p}} & \coloneqq \freevareven{p} \\
    \freevareven{\patOr{p_1}{p_2}} & \coloneqq  \freevareven{p_1} \cup \freevareven{p_2} & \freevarodd{\patOr{p_1}{p_2}} & \coloneqq \freevarodd{p_1} \cup \freevarodd{p_2} \\
    \freevareven{\patAnd{p_1}{p_2}} & \coloneqq  \freevareven{p_1} \cup \freevareven{p_2} & \freevarodd{\patAnd{p_1}{p_2}} & \coloneqq \freevarodd{p_1} \cup \freevarodd{p_2} \\
    \freevareven{\mathcal{C}(p_1,\ldots,p_n)} & \coloneqq  \cup_{i=1}^{n}\freevareven{p_i} & \freevarodd{\mathcal{C}(p_1,\ldots,p_n)} & \coloneqq  \cup_{i=1}^{n} \freevarodd{p_i}
  \end{array}
  \]
  \caption{The syntax of values, patterns and substitutions. We track whether variables occur under an even or odd number of negations.}
  \label{fig:pattern-algebra:syntax}
\end{figure}

Variable, constructor and wildcard patterns are standard, so we do not introduce them explicitly and turn directly to the more interesting absurd-patterns, or-patterns, negation-patterns, and and-patterns.

The absurd-pattern $\patBot$ is the dual of the wildcard-pattern $\patTop$.
Whereas the wildcard-pattern matches any value, the absurd-pattern never matches a value.
We include the absurd-pattern in the syntax of patterns since we need a canonical bottom element of the boolean algebra.
But there are also more practical reasons to include an absurd pattern.
In the lazy programming language Haskell, for example, we sometimes have to write a pattern that never matches but forces the evaluation of an argument.
This is often achieved by adding a guard clause that evaluates to false, but we can also use an absurd pattern for the same purpose.

An or-pattern $\patOr{p_1}{p_2}$ matches a value if either $p_1$ or $p_2$ (or both) matches the value.
Or-patterns introduce a complication for the compiler:
we have to ensure that if values are bound to variables in $p_1$ during matching, then $p_2$ has to bind values of the same type to the same variables.
The following example on the left shows how or-patterns without variables can be used to combine multiple clauses into one.
\[
\begin{array}{ll}
  \mathbf{data}\ \text{Day}\ = \tmMo \mid \tmTu \mid \tmWe \mid \tmTh \mid \tmFr \mid \tmSa \mid \tmSu  & \\
  \text{isWeekend} : \text{Day} \to \text{Bool} & \text{isWeekend} : \text{Day} \to \text{Bool} \\
  \text{isWeekend}(\patOr{\tmSa}{\tmSu}) \coloneqq \tmTrue & \text{isWeekend}(\patOr{\tmSa}{\tmSu}) \coloneqq \tmTrue \\
  \text{isWeekend}(\patOr{\tmMo}{\patOr{\tmTu}{\patOr{\tmWe}{\patOr{\tmTh}{\tmFr}}}}) \coloneqq \tmFalse & \text{isWeekend}(\patNeg{(\patOr{\tmSa}{\tmSu})}) \coloneqq \tmFalse
\end{array}
\]

We could also have written the second clause of the previous example using a wildcard pattern, but then the two patterns would overlap.
We can get rid of the verbose second clause by using a negation pattern, which is illustrated on the right.
A negation pattern matches if, and only if, the negated pattern does not match.

The last type of patterns we introduce are and-patterns $\patAnd{p_1}{p_2}$ which succeed if both $p_1$ and $p_2$ match the scrutinee.
While or-patterns $\patOr{p_1}{p_2}$ require that $p_1$ and $p_2$ bind the same set of variables, and-patterns $\patAnd{p_1}{p_2}$ require $p_1$ and $p_2$ to bind disjoint sets of variables.
Many languages already have a special case of and-patterns: An $@$-pattern $x@p$ (some languages also use the keyword \enquote{as}) is an and-pattern where the left side must be a variable.
If we want to rewrite the previous example to return a string instead of a boolean value, then we could use the following and-patterns:
\begin{align*}
  &\text{isWeekend} : \text{Day} \to \text{String} \\
  &\text{isWeekend}(\patAnd{x}{(\patOr{\tmSa}{\tmSu})}) \coloneqq \text{show}(x) \mdoubleplus \text{`` is on the weekend''} \\
  &\text{isWeekend}(\patAnd{x}{\patNeg{(\patOr{\tmSa}{\tmSu})}}) \coloneqq \text{show}(x) \mdoubleplus \text{`` is not on the weekend''}
\end{align*}

Here we use the show function to print a weekday and $\mdoubleplus$ for string concatenation.
We will use this snippet as a running example in \cref{sec:normalization} and \cref{sec:compilation} to show how patterns are normalized and compiled to efficient code.

\subsection{Matching Semantics of Patterns}
\label{subsec:algebra-of-patterns:formal-semantics}

We will now present the formal rules which specify what happens when we match a pattern against a value; these rules can be found in the upper half of \cref{fig:pattern-algebra:rules}.
We use two mutually recursive judgments to formalize the semantics of matching:
The judgment $\matches{p}{v}{\sigma}$ expresses that the pattern $p$ successfully matches the value $v$ and also binds values to variables of the pattern; those mappings are recorded in the substitution $\sigma$.
The judgment $\matchesNot{p}{v}{\sigma}$ expresses that the pattern $p$ does \emph{not} match the value $v$, and also produces a substitution.
That the judgment $\matchesNot{p}{v}{\sigma}$ also produces a substitution might be surprising: How can there be a substitution if the value does not match the pattern?
We can think of these as “temporarily deactivated” substitutions which will become available again if we enclose the pattern in an additional negation.
In fact, such substitutions are necessary if we want to guarantee that the pattern $p$ is semantically equivalent to $\patNeg{\patNeg{p}}$ in the case where $p$ contains variables, for example in \cref{eq:example:good:5} below.

\begin{figure}[p]
    \begin{flushright}
      \fbox{Pattern matches: \matches{p}{v}{\sigma}}
    \end{flushright}
    \vspace{-0.1cm}

    \begin{minipage}{0.3\textwidth}
      \begin{prooftree}
        \AxiomC{\phantom{\matchesNot{p}{v}{\sigma}}}
        \RightLabel{\textsc{Var}}
        \UnaryInfC{$\matches{x}{v}{[x \mapsto v]}$}
      \end{prooftree}
    \end{minipage}
    \begin{minipage}{0.3\textwidth}
      \begin{prooftree}
        \AxiomC{\phantom{\matchesNot{p}{v}{\sigma}}}
        \RightLabel{\textsc{Wild}}
        \UnaryInfC{\matches{\patTop}{v}{[]}}
      \end{prooftree}
    \end{minipage}
    \begin{minipage}{0.3\textwidth}
      \begin{prooftree}
        \AxiomC{\matchesNot{p}{v}{\sigma}}
        \RightLabel{\textsc{Neg}$_1$}
        \UnaryInfC{\matches{\patNeg{p}}{v}{\sigma}}
      \end{prooftree}
    \end{minipage}
    \vspace{0.1cm}

    \begin{minipage}{0.25\textwidth}
      \begin{prooftree}
        \AxiomC{\matches{p_1}{v}{\sigma}}
        \RightLabel{\textsc{Or}$_1$}
        \UnaryInfC{\matches{\patOr{p_1}{p_2}}{v}{\sigma}}
      \end{prooftree}
    \end{minipage}
    \hfill
    \begin{minipage}{0.25\textwidth}
      \begin{prooftree}
        \AxiomC{\matches{p_2}{v}{\sigma}}
        \RightLabel{\textsc{Or}$_2$}
        \UnaryInfC{\matches{\patOr{p_1}{p_2}}{v}{\sigma}}
      \end{prooftree}
    \end{minipage}
    \hfill
    \begin{minipage}{0.4\textwidth}
      \begin{prooftree}
        \AxiomC{\matches{p_1}{v}{\sigma_1}}
        \AxiomC{\matches{p_2}{v}{\sigma_2}}
        \RightLabel{\textsc{And}}
        \BinaryInfC{\matches{\patAnd{p_1}{p_2}}{v}{\sigma_1 \mdoubleplus \sigma_2}}
      \end{prooftree}
    \end{minipage}

    \begin{prooftree}
      \AxiomC{\matches{p_1}{v_1}{\sigma_1}}
      \AxiomC{\ldots}
      \AxiomC{\matches{p_n}{v_n}{\sigma_n}}
      \RightLabel{\textsc{Ctor}}
      \TrinaryInfC{\matches{\mathcal{C}^n(p_1,\ldots,p_n)}{\mathcal{C}^n(v_1,\ldots,v_n)}{\sigma_1 \mdoubleplus \ldots \mdoubleplus \sigma_n}}
    \end{prooftree}

    \begin{flushright}
      \fbox{Pattern doesn't match: \matchesNot{p}{v}{\sigma}}
    \end{flushright}
    \vspace{-0.1cm}

    \begin{minipage}{0.3\textwidth}
      \begin{prooftree}
        \AxiomC{\phantom{\matchesNot{p}{v}{X}}}
        \RightLabel{\textsc{Absurd}}
        \UnaryInfC{\matchesNot{\patBot}{v}{[]}}
      \end{prooftree}
    \end{minipage}
    \begin{minipage}{0.6\textwidth}
      \begin{prooftree}
        \AxiomC{$\exists i: \matchesNot{p_i}{v_i}{\sigma}$}
        \RightLabel{\textsc{Ctor}$_1$}
        \UnaryInfC{\matchesNot{\mathcal{C}^n(p_1,\ldots,p_n)}{\mathcal{C}^n(v_1,\ldots,v_n)}{\sigma}}
      \end{prooftree}
    \end{minipage}
    \vspace{0.1cm}

    \begin{minipage}{0.3\textwidth}
      \begin{prooftree}
        \AxiomC{\matches{p}{v}{\sigma}}
        \RightLabel{\textsc{Neg}$_2$}
        \UnaryInfC{\matchesNot{\patNeg{p}}{v}{\sigma}}
      \end{prooftree}
    \end{minipage}
    \begin{minipage}{0.6\textwidth}
      \begin{prooftree}
        \AxiomC{$\mathcal{C}^n \neq \mathcal{C}'^m$}
        \RightLabel{\textsc{Ctor}$_2$}
        \UnaryInfC{\matchesNot{\mathcal{C}^n(p_1,\ldots,p_n)}{\mathcal{C}'^m(v_1,\ldots,v_m)}{[]}}
      \end{prooftree}
    \end{minipage}
    \vspace{0.1cm}

    \begin{minipage}{0.27\textwidth}
      \begin{prooftree}
        \AxiomC{\matchesNot{p_1}{v}{\sigma}}
        \RightLabel{\textsc{And}$_1$}
        \UnaryInfC{\matchesNot{\patAnd{p_1}{p_2}}{v}{\sigma}}
      \end{prooftree}
    \end{minipage}
    \hfill
    \begin{minipage}{0.27\textwidth}
      \begin{prooftree}
        \AxiomC{\matchesNot{p_2}{v}{\sigma}}
        \RightLabel{\textsc{And}$_2$}
        \UnaryInfC{\matchesNot{\patAnd{p_1}{p_2}}{v}{\sigma}}
      \end{prooftree}
    \end{minipage}
    \hfill
    \begin{minipage}{0.4\textwidth}
      \begin{prooftree}
        \AxiomC{\matchesNot{p_1}{v}{\sigma_1}}
        \AxiomC{\matchesNot{p_2}{v}{\sigma_2}}
        \RightLabel{\textsc{Or}}
        \BinaryInfC{\matchesNot{\patOr{p_1}{p_2}}{v}{\sigma_1 \mdoubleplus \sigma_2}}
      \end{prooftree}
    \end{minipage}


  \begin{flushright}
      \fbox{Pattern is linear: $\linearpos{p}$ and $\linearneg{p}$}
  \end{flushright}
  \vspace{-0.2cm}

  \begin{minipage}{0.3\textwidth}
    \begin{prooftree}
        \AxiomC{\phantom{\overlaps{p}{p}}}
        \RightLabel{\textsc{L-Var}$^\pm$}
        \UnaryInfC{$\linearposneg{x}$}
    \end{prooftree}
  \end{minipage}
  \begin{minipage}{0.3\textwidth}
    \begin{prooftree}
      \AxiomC{\phantom{\overlaps{p}{p}}}
      \RightLabel{\textsc{L-Absurd}$^\pm$}
      \UnaryInfC{$\linearposneg{\patBot}$}
    \end{prooftree}
  \end{minipage}
  \begin{minipage}{0.3\textwidth}
    \begin{prooftree}
        \AxiomC{\phantom{\overlaps{p}{p}}}
        \RightLabel{\textsc{L-Wild}$^\pm$}
        \UnaryInfC{$\linearposneg{\patWildcard}$}
    \end{prooftree}
  \end{minipage}
  \vspace{0.1cm}

  \begin{minipage}{0.45\textwidth}
    \begin{prooftree}
      \AxiomC{$\linearpos{p_{1,2}} \quad \freevareven{p_1} = \freevareven{p_2}$}
      \RightLabel{\textsc{L-Or}$^+$}
      \UnaryInfC{$\linearpos{\patOr{p_1}{p_2}}$}
    \end{prooftree}
  \end{minipage}
  \hfill
  \begin{minipage}{0.495\textwidth}
    \begin{prooftree}
      \AxiomC{$\linearneg{p_{1,2}} \quad \freevarodd{p_1} \cap \freevarodd{p_2} = \emptyset$}
      \RightLabel{\textsc{L-Or}$^-$}
      \UnaryInfC{$\linearneg{\patOr{p_1}{p_2}}$}
    \end{prooftree}
  \end{minipage}
  \vspace{0.1cm}

  \begin{minipage}{0.5\textwidth}
    \begin{prooftree}
      \AxiomC{$\linearpos{p_{1,2}} \quad \freevareven{p_1} \cap \freevareven{p_2} = \emptyset$}
      \RightLabel{\textsc{L-And}$^+$}
      \UnaryInfC{$\linearpos{\patAnd{p_1}{p_2}}$}
    \end{prooftree}
  \end{minipage}
  \hfill
  \begin{minipage}{0.475\textwidth}
    \begin{prooftree}
      \AxiomC{$\linearneg{p_{1,2}} \quad \freevarodd{p_1} = \freevarodd{p_2}$}
      \RightLabel{\textsc{L-And}$^-$}
      \UnaryInfC{$\linearneg{\patAnd{p_1}{p_2}}$}
    \end{prooftree}
  \end{minipage}
  \vspace{0.1cm}

  \begin{minipage}{0.45\textwidth}
    \begin{prooftree}
      \AxiomC{$\linearneg{p}$}
      \RightLabel{\textsc{L-Neg}$^+$}
      \UnaryInfC{$\linearpos{\patNeg{p}}$}
    \end{prooftree}
  \end{minipage}
  \begin{minipage}{0.45\textwidth}
    \begin{prooftree}
      \AxiomC{$\linearpos{p}$}
      \RightLabel{\textsc{L-Neg}$^-$}
      \UnaryInfC{$\linearneg{\patNeg{p}}$}
    \end{prooftree}
  \end{minipage}
  \vspace{0.1cm}

  \begin{minipage}{0.475\textwidth}
    \begin{prooftree}
      \AxiomC{$\linearpos{p_i} \quad \bigcap_{i=1}^{n}\freevareven{p_i} = \emptyset$}
      \RightLabel{\textsc{L-Ctor}$^+$}
      \UnaryInfC{$\linearpos{C^n(p_1, \ldots, p_n)}$}
    \end{prooftree}
  \end{minipage}
  \hfill
  \begin{minipage}{0.475\textwidth}
    \begin{prooftree}
      \AxiomC{$\linearneg{p_i} \quad \bigcup_{i=1}^{n}\freevarodd{p_i} = \emptyset$}
      \RightLabel{\textsc{L-Ctor}$^-$}
      \UnaryInfC{$\linearneg{C^n(p_1, \ldots, p_n)}$}
    \end{prooftree}
  \end{minipage}

  \begin{flushright}
    \fbox{Pattern is deterministic: $\deterministic{p}$}
  \end{flushright}
  \vspace{-0.2cm}

  \begin{minipage}{0.22\textwidth}
    \begin{prooftree}
      \AxiomC{\phantom{\overlaps{p}{p}}}
      \RightLabel{\textsc{D-Var}}
      \UnaryInfC{$\deterministic{x}$}
    \end{prooftree}
  \end{minipage}
  \begin{minipage}{0.22\textwidth}
    \begin{prooftree}
      \AxiomC{\phantom{\overlaps{p}{p}}}
      \RightLabel{\textsc{D-Wild}}
      \UnaryInfC{$\deterministic{\patWildcard}$}
    \end{prooftree}
  \end{minipage}
  \begin{minipage}{0.22\textwidth}
    \begin{prooftree}
      \AxiomC{$\deterministic{p_1}$}
      \RightLabel{\textsc{D-Neg}}
      \UnaryInfC{$\deterministic{\patNeg{p_1}}$}
    \end{prooftree}
  \end{minipage}
  \begin{minipage}{0.3\textwidth}
    \begin{prooftree}
      \AxiomC{\phantom{\overlaps{p}{p}}}
      \RightLabel{\textsc{D-Absurd}}
      \UnaryInfC{$\deterministic{\patBot}$}
    \end{prooftree}
  \end{minipage}
  \vspace{0.1cm}

  \begin{minipage}{0.4\textwidth}
    \begin{prooftree}
      \AxiomC{$\deterministic{p_{1,2}} \quad \disjoint{p_1}{p_2}$}
      \RightLabel{\textsc{D-Or}$_1$}
      \UnaryInfC{$\deterministic{\patOr{p_1}{p_2}}$}
    \end{prooftree}
  \end{minipage}
  \begin{minipage}{0.5\textwidth}
    \begin{prooftree}
      \AxiomC{$\deterministic{p_{1,2}} \quad \freevareven{p_1} = \freevareven{p_2} = \emptyset$}
      \RightLabel{\textsc{D-Or}$_2$}
      \UnaryInfC{$\deterministic{\patOr{p_1}{p_2}}$}
    \end{prooftree}
  \end{minipage}
  \vspace{0.1cm}

  \begin{minipage}{0.4\textwidth}
    \begin{prooftree}
      \AxiomC{$\deterministic{p_{1,2}} \quad \disjoint{\patNeg{p_1}}{\patNeg{p_2}}$}
      \RightLabel{\textsc{D-And}$_1$}
      \UnaryInfC{$\deterministic{\patAnd{p_1}{p_2}}$}
    \end{prooftree}
  \end{minipage}
  \begin{minipage}{0.55\textwidth}
    \begin{prooftree}
      \AxiomC{$\deterministic{p_{1,2}} \quad \freevarodd{p_1} = \freevarodd{p_2} = \emptyset$}
      \RightLabel{\textsc{D-And}$_2$}
      \UnaryInfC{$\deterministic{\patAnd{p_1}{p_2}}$}
    \end{prooftree}
  \end{minipage}

  \begin{prooftree}
    \AxiomC{$\deterministic{p_1} \cdots \ \deterministic{p_n}$}
    \RightLabel{\textsc{D-Ctor}}
    \UnaryInfC{$\deterministic{C^n(p_1, \ldots, p_n)}$}
  \end{prooftree}

  \caption{Rules for matching patterns against values, and checking for linearity and determinism.}
  \label{fig:pattern-algebra:rules}
\end{figure}

\begin{example}
  These examples show some well-behaved instances of the matching judgments:
  \begin{align}
    \matches{\text{Cons}(x,xs)}{\text{Cons}(2,\text{Cons}(3,\text{Nil}))}{[x \mapsto 2, xs \mapsto \text{Cons}(3,\text{Nil})]}
    \label{eq:example:good:1}\\
    \matches{\patOr{\text{True}}{\text{False}}}{\text{True}}{[]}
    \label{eq:example:good:2}\\
    \matchesNot{\text{True}}{\text{False}}{[]}
    \label{eq:example:good:3}\\
    \matchesNot{\patNeg{x}}{\text{True}}{[x \mapsto \text{True}]}
    \label{eq:example:good:4}\\
    \matches{\patNeg{\patNeg{x}}}{\text{True}}{[x \mapsto \text{True}]}
    \label{eq:example:good:5}
  \end{align}
  \cref{eq:example:good:1} illustrates how the rules \textsc{Var} and \textsc{Ctor} can be used to match a constructor pattern against a list while binding components of the list to variables $x$ and $xs$.
  \cref{eq:example:good:2} shows how the rule \textsc{Or}$_1$ (together with the rule \textsc{Ctor}) allows to successfully match the left side of an or-pattern against a value.
  \cref{eq:example:good:3} illustrates the rule \textsc{Ctor}$_2$ and \cref{eq:example:good:4} the combination of rules \textsc{Var} and \textsc{Neg}$_2$.
  \cref{eq:example:good:5} shows why we need two mutually recursive judgments: in order for the pattern $\patNeg{\patNeg{x}}$ to match against the value $\text{True}$ we need to use the rules \textsc{Neg}$_1$, \textsc{Neg}$_2$ and \textsc{Var} to ensure that the variable $x$ is bound to the value $\text{True}$.
\end{example}

Using the rules for $\matches{p}{v}{\sigma}$ and $\matchesNot{p}{v}{\sigma}$ we can already prove that they are, in a certain sense, sound and complete:
\begin{theorem}[Matching rules are sound, \rocq]
  \label{thm:matching-sound}
  The rules from \cref{fig:pattern-algebra:rules} are sound.
  There is no pattern $p$, value $v$ and substitutions $\sigma_1,\sigma_2$ such that both $\matches{p}{v}{\sigma_1}$ and $\matchesNot{p}{v}{\sigma_2}$ hold.
\end{theorem}
\begin{theorem}[Matching rules are complete, \rocq]
  \label{thm:matching-complete}
  The rules from \cref{fig:pattern-algebra:rules} are complete, i.e. for any pattern $p$ and value $v$ there is some $\sigma$ such that $\matches{p}{v}{\sigma}$ or $\matchesNot{p}{v}{\sigma}$ holds.
\end{theorem}

Since we are interested in algebraic reasoning, we have to say when we consider two patterns to be semantically equivalent.
We will define semantic equivalence in \cref{def:semantic-equivalence-patterns} using the two judgments introduced above.
But first, we have to define when two substitutions (which we formalized as lists of mappings) are semantically equivalent:

\begin{definition}[Semantic equivalence of substitutions, \rocq]
  Two substitutions $\sigma$ and $\sigma'$ are semantically equivalent if they contain the same mappings:
  \begin{equation*}
    \semantic{\sigma} \equiv \semantic{\sigma'}  \coloneqq \forall m, m \in \sigma \Leftrightarrow m \in \sigma'
  \end{equation*}
\end{definition}

We can now state the definition for semantic equivalence of patterns.

\begin{definition}[Semantic equivalence of patterns, \rocq]
  \label{def:semantic-equivalence-patterns}
  Two patterns $p$ and $q$ are semantically equivalent if they match the same values with equivalent substitutions, and if they also do not match the same values with equivalent substitutions:
  \begin{align*}
    \semantic{p} \equiv \semantic{q} \coloneqq \forall v, \forall \sigma, \ & \matches{p}{v}{\sigma} \Rightarrow \exists \sigma', \matches{q}{v}{\sigma'} \wedge \semantic{\sigma} \equiv \semantic{\sigma'} \\
    \land \ & \matches{q}{v}{\sigma} \Rightarrow \exists \sigma', \matches{p}{v}{\sigma'} \wedge \semantic{\sigma} \equiv \semantic{\sigma'} \\
    \land \ & \matchesNot{p}{v}{\sigma} \Rightarrow \exists \sigma', \matchesNot{q}{v}{\sigma'} \wedge \semantic{\sigma} \equiv \semantic{\sigma'} \\
    \land \ & \matchesNot{q}{v}{\sigma} \Rightarrow \exists \sigma', \matchesNot{p}{v}{\sigma'} \wedge \semantic{\sigma} \equiv \semantic{\sigma'}
  \end{align*}
\end{definition}

We have to include both judgment forms if we want semantic equivalence to be a congruence relation.
To see why this is the case, suppose that we omit the second half of \cref{def:semantic-equivalence-patterns}.
It would then follow that the two patterns $\patBot$ and $\patNeg{x}$ are semantically equivalent since they match against the same set of values (i.e. no values at all).
But if we apply a negation-pattern to both patterns, then we can easily show that the resulting patterns are no longer semantically equivalent:
$\semantic{\patNeg{\patBot}} \equiv \semantic{\patWildcard} \not\equiv \semantic{x} \equiv \semantic{\patNeg{\patNeg{x}}}$.
But since \cref{def:semantic-equivalence-patterns} requires both patterns to agree on the values they match \emph{and} don't match against, it actually defines a congruence relation on patterns:

\begin{theorem}[Congruence, \rocq]
   If $\semantic{p_1} \equiv \semantic{p_1'}$ to $\semantic{p_n} \equiv \semantic{p_n'}$ hold, then we also have $\semantic{\patNeg{p_1}} \equiv \semantic{\patNeg{p_1'}}$, $\semantic{\patAnd{p_1}{p_2}} \equiv \semantic{\patAnd{p_1'}{p_2'}}$, $\semantic{\patOr{p_1}{p_2}} \equiv \semantic{\patOr{p_1'}{p_2'}}$ and $\semantic{\mathcal{C}(p_1,\ldots,p_n)} \equiv \semantic{\mathcal{C}(p_1',\ldots,p_n')}$.
\end{theorem}

We can use the definition of semantic equivalence to prove that the following algebraic laws hold for patterns:

\begin{theorem}[Algebraic Equivalences of Patterns, I, \rocq]
  \label{thm:algebraic-equivalences:one}
  For all patterns $p, q, r$, the following equivalences hold:
  \[
  \begin{array}{rclrclc}
    \semantic{\patAnd{p}{q}} &\equiv& \semantic {\patAnd{q}{p}} &
    \semantic{\patOr{p}{q}} &\equiv& \semantic {\patOr{q}{p}} &
    \emph{Commutativity} \\
    \semantic{\patAnd{p}{(\patAnd{q}{r})}} &\equiv& \semantic{\patAnd{(\patAnd{p}{q})}{r}} &
    \semantic{\patOr{p}{(\patOr{q}{r})}} &\equiv& \semantic{\patOr{(\patOr{p}{q})}{r}} &
    \emph{Associativity} \\
    \semantic{\patAnd{p}{\patWildcard}} &\equiv& \semantic{p} &
    \semantic{\patOr{p}{\patBot}} &\equiv& \semantic{p} &
    \emph{Neutral Elements} \\
    \semantic{\patNeg{\patWildcard}} &\equiv& \semantic{\patBot} &
    \semantic{\patNeg{\patBot}} &\equiv& \semantic{\patWildcard} &
    \emph{Duality} \\
    \semantic{\patNeg{(\patOr{p}{q})}} &\equiv& \semantic{\patAnd{(\patNeg{p})}{(\patNeg{q})}} &
    \semantic{\patNeg{(\patAnd{p}{q})}} &\equiv& \semantic{\patOr{(\patNeg{p})}{(\patNeg{q})}} &
    \emph{De Morgan} \\
    \multicolumn{6}{c}{\semantic{\patNeg{\patNeg{p}}} \equiv \semantic{p} } & \emph{Double Negation}
  \end{array}
  \]
\end{theorem}

In addition to these boolean laws we can also prove the following equivalences which involve constructor patterns:

\begin{theorem}[Equivalences of Constructor Patterns, I, \rocq]
  \label{thm:algebraic-equivalences:two}
  For all patterns $p_1$ to $p_n$ and $p_1'$ to $p_n'$, the following equivalences hold:
  \begin{align*}
    \semantic{\patAnd{\mathcal{C}(p_1,\ldots,p_n)}{\mathcal{C}(p'_1, \ldots, p'_n)}} &\equiv \semantic{\mathcal{C}(\patAnd{p_1}{p'_1}, \ldots, \patAnd{p_n}{p_n'})}\\
    \semantic{\mathcal{C}(p_1,\ldots, \patOr{p_i}{p'_i},\ldots, p_n)} &\equiv \semantic{\patOr{\mathcal{C}(p_1,\ldots,p_i,\ldots,p_n)}{\mathcal{C}(p_1,\ldots,p'_i,\ldots,p_n)}}
  \end{align*}
\end{theorem}

There are still some equivalences that we expect to hold but which are missing from \cref{thm:algebraic-equivalences:one} and \cref{thm:algebraic-equivalences:two}.
These additional laws, like the distributive law for and- and or-patterns, are not universally valid and require additional restrictions on patterns.
We will motivate and introduce these additional constraints in the next subsection.

\subsection{Linear Patterns}
\label{subsec:pattern-algebra:linear-patterns}

The first restriction we impose on patterns is \emph{linearity}.
The linearity restriction is easy to state if the syntax of patterns does not include and-patterns, or-patterns and negation-patterns: A pattern is linear if every variable occurs at most once.
To see why this restriction is necessary consider the nonlinear pattern $\tmCons{x}{x}$, for which we can derive the following judgment:
\begin{align}
  \label{eq:match-counterexample-1}
  \matches{\tmCons{x}{x}}{\tmCons{2}{\tmNil}}{[x\mapsto 2, x\mapsto \tmNil]}
\end{align}
The problem here is that multiple values are bound to the same variable, i.e. the result of the pattern match is not a proper substitution.

We have a more expressive language of algebraic patterns, and the simple definition of linearity described above is no longer sufficient.
To see why this is the case, consider the pattern $\patOr{(\patAnd{x}{\tmTrue})}{\tmFalse}$ that is linear according to the simple definition, but which shows the following pathological behavior:
\begin{align}
  \label{eq:match-counterexample-2}
  \matches{\patOr{(\patAnd{x}{\tmTrue})}{\tmFalse}}{\tmTrue}{[x \mapsto \tmTrue]} \\
  \label{eq:match-counterexample-3}
  \matches{\patOr{(\patAnd{x}{\tmTrue})}{\tmFalse}}{\tmFalse}{[]}
\end{align}
\cref{eq:match-counterexample-2,eq:match-counterexample-3} show that matching against this pattern can produce substitutions that do not have the same domain.

The two examples above tell us what the correct notion of linearity should guarantee:
First, if a linear pattern matches against a value, then the resulting substitution should be a proper substitution which contains at most one mapping for a variable.
Second, if we use a pattern in a clause we should be able to read off the variables that will be in the domain of the substitution of a successful match.

In order to ensure that these properties hold we have to use two mutually recursive judgments $\linearpos{p}$ and $\linearneg{p}$, whose rules are given in \cref{fig:pattern-algebra:rules}.
We use $\linearposneg{p}$ for patterns which satisfy both $\linearpos{p}$ and $\linearneg{p}$.
\cref{thm:pattern-algebra:linear-patterns-covering,thm:pattern-algebra:linear-patterns-substitutions} show that linear patterns do have the properties described in the previous paragraph.

\begin{theorem}[Linear patterns produce covering substitutions, \rocq]
  \label{thm:pattern-algebra:linear-patterns-covering}
  For any pattern $p$, value $v$ and substitution $\sigma$:
  \begin{enumerate}
    \item If $\linearpos{p}$ and $\matches{p}{v}{\sigma}$, then the domain of $\sigma$ is equal to $\freevareven{p}$.
    \item If $\linearneg{p}$ and $\matchesNot{p}{v}{\sigma}$, then the domain of $\sigma$ is equal to $\freevarodd{p}$.
  \end{enumerate}
\end{theorem}

\begin{theorem}[Linear patterns produce proper substitutions, \rocq]
  \label{thm:pattern-algebra:linear-patterns-substitutions}
  For any pattern $p$, value $v$ and substitution $\sigma$:
  \begin{enumerate}
    \item If $\linearpos{p}$ and $\matches{p}{v}{\sigma}$, then $\sigma$ is a proper substitution.
    \item If $\linearneg{p}$ and $\matchesNot{p}{v}{\sigma}$, then $\sigma$ is a proper substitution.
  \end{enumerate}
  A substitution $\sigma$ is proper if any variable occurs at most once in the domain of $\sigma$.
\end{theorem}

We will now motivate why the rules in \cref{fig:pattern-algebra:rules} have the specific form that they do.
\cref{thm:pattern-algebra:linear-patterns-covering} tells us that we should think about the variables in $\freevareven{p}$ as the actual pattern variables which we can use in the right-hand side of a clause, and that we should think of the variables in $\freevarodd{p}$ as the \enquote{temporarily deactivated} pattern variables
\enquote{reactivated} by negation.

Let us first look at the restrictions on the variables in $\freevareven{p}$, since those are more intuitive.
In the or-pattern $\patOr{p_1}{p_2}$ we require that $\freevareven{p_1}$ is equal to $\freevareven{p_2}$ since we do not know which of the two patterns will match, and we have to ensure that the same variables appear in the substitution.
In the and-pattern $\patAnd{p_1}{p_2}$ we require that $\freevareven{p_1}$ is disjoint from $\freevareven{p_2}$ since we do not want the same variable to be mapped to two values in the substitution.
The reasoning is very similar for constructor patterns $\mathcal{C}(p_1,\ldots,p_i$) where we require all the $\freevareven{p_i}$ to be disjoint.

Next, let us look at the restrictions on $\freevarodd{e}$, i.e. the variables which occur under an odd number of negations.
The rules for and-patterns and or-patterns are motivated by duality.
We have seen that the De Morgan rules are valid for patterns.
We therefore want, for example, that if the pattern $\patNeg{(\patOr{p_1}{p_2})}$ is linear, the pattern $\patAnd{\patNeg{p_1}}{\patNeg{p_2}}$ should be linear as well.
For this to be true, the restrictions on $\freevarodd{p}$ for or-patterns have to mirror the restrictions on $\freevareven{p}$ for and-patterns, and vice-versa.
We still have to explain why we require for a constructor pattern $\mathcal{C}(p_1,\ldots,p_n)$ that all $\freevarodd{p_i}$ have to be the empty set.
The reason for this restriction lies in the following semantic equivalence,
which is motivated by its importance for rewriting patterns into a normal form (see \cref{subsec:normalization:nnf}) and which we will prove below.
\begin{equation*}
  \semantic{\patNeg{\mathcal{C}(p_1,\ldots,p_n)}}
  \equiv
  \semantic{\patOr{\patNeg{\mathcal{C}(\patWildcard,\ldots,\patWildcard)}}{\patOr{\mathcal{C}(\patNeg{p_1},\ldots,\patWildcard)}{\patOr{\ldots}{\mathcal{C}(\patWildcard,\ldots,\patNeg{p_n})}}}}
\end{equation*}
The pattern on the right consists of several patterns that are joined by or-patterns.
The rules for or-patterns require that each disjunct contains the same sets $\freevareven{-}$.
But since the first pattern $\patNeg{\mathcal{C}(\patWildcard,\ldots,\patWildcard)}$ doesn't contain any variables, the other disjuncts must not contain any variables under an even number of negations.
And since the subpatterns $p_i$ occur under a negation, we can deduce that the patterns $p_i$ must not contain any variables under an odd number of negations.

Using linearity we can now prove the missing equivalences which were not included in \cref{thm:algebraic-equivalences:one,thm:algebraic-equivalences:two}.

\begin{theorem}[Algebraic equivalences of Patterns, II, \rocq]
  \label{thm:algebraic-equivalences:three}
  For all patterns $p,q$ and $r$, the following equivalences hold if the patterns on both sides are linear (i.e.~$\linearposneg{}$).
  \[
  \begin{array}{rlrlc}
    \semantic{\patAnd{p}{(\patOr{q}{r})}}\mkern-10mu &\equiv \semantic{\patOr{(\patAnd{p}{q})}{(\patAnd{p}{r})}} &
    \semantic{\patOr{p}{(\patAnd{q}{r})}}\mkern-10mu &\equiv \semantic{\patAnd{(\patOr{p}{q})}{(\patOr{p}{r})}} &
    \emph{Distributivity} \\
    \semantic{\patAnd{p}{p}}\mkern-10mu &\equiv \semantic{p} &
    \semantic{\patOr{p}{p}}\mkern-10mu &\equiv \semantic{p} &
    \emph{Idempotence} \\
    \semantic{\patAnd{p}{\patBot}}\mkern-10mu &\equiv \semantic{\patBot} &
    \semantic{\patOr{p}{\patWildcard}}\mkern-10mu &\equiv \semantic{\patWildcard} &
    \emph{Zeros}
  \end{array}
  \]
\end{theorem}

\cref{thm:algebraic-equivalences:one} and \cref{thm:algebraic-equivalences:three} taken together show that patterns form a boolean algebra, with the caveat that the conjunction or disjunction of two linear patterns does not necessarily result in a linear pattern.
We can also prove some additional equivalences involving constructor patterns.

\begin{theorem}[Equivalences of Constructor Patterns, II, \rocq]
  \label{thm:algebraic-equivalences:four}
  For all patterns $p_1$ to $p_n$ and $p_1'$ to $p_n'$ and constructors $\mathcal{C} \neq \mathcal{C}'$, the following equivalences hold if the patterns on both sides are linear (i.e.~$\linearposneg{}$).
  \begin{align*}
    &&\semantic{\mathcal{C}(p_1,\ldots,\patBot,\ldots,p_n)} &\equiv \semantic{\patBot} \\
    &&\semantic{\patAnd{\mathcal{C}(p_1,\ldots,p_n)}{\mathcal{C}'(p'_1,\ldots,p'_m)}} &\equiv \semantic{\patBot} \\
    &&\semantic{\patAnd{\mathcal{C}(p_1,\ldots, p_n)}{\patNeg{\mathcal{C}'(p'_1,\ldots, p'_m)}}} &\equiv \semantic{\mathcal{C}(p_1,\ldots,p_n)}\\
    &&\multispan2{$\semantic{\patNeg{\mathcal{C}(p_1,\ldots,p_n)}} \equiv \semantic{\patOr{\patNeg{\mathcal{C}(\patWildcard,\ldots,\patWildcard)}}{\patOr{\mathcal{C}(\patNeg{p_1},\ldots,\patWildcard)}{\patOr{\ldots}{\mathcal{C}(\patWildcard,\ldots,\patNeg{p_n})}}}}$\hfil}
  \end{align*}
\end{theorem}

The equivalences of \cref{thm:algebraic-equivalences:four} are essential for proving the correctness of the normalization algorithm specified in \cref{sec:normalization}.

\subsection{Deterministic Patterns}
\label{subsec:pattern-algebra:deterministic-patterns}

We also have to introduce a notion of deterministic patterns.
To see why this is necessary, consider the pathological pattern $\patOr{\tmPair{x}{\patWildcard}}{\tmPair{\patWildcard}{x}}$.
This pattern is linear, but when we match against the value $\tmPair{2}{4}$ we can either bind the value $2$ or the value $4$ to the variable $x$.
Deterministic patterns guarantee that when there are multiple derivations which show that a pattern matches against a value, then these alternatives must yield equivalent substitutions\footnote{\
  Most languages which implement or-patterns do not check whether patterns are deterministic and instead evaluate or-patterns left-to-right, but this invalidates algebraic laws such as the commutativity of or-patterns.
  If we need such left-biased or-patterns $\patLeftOr{p}{q}$ then we can easily define them as $\patOr{p}{(\patAnd{\patNeg{p}}{q})}$.
}.

To define deterministic patterns we first have to introduce overlapping and disjoint patterns. In \cref{sec:appendix-overlap} we give a procedure for deciding whether two patterns in a certain normalized form -- which we will introduce in \cref{sec:normalization} -- overlap.
\begin{definition}[Overlapping and Disjoint Patterns]
  Two patterns $p$ and $q$ are overlapping, written $\overlaps{p}{q}$, if there exists a value $v$ and substitutions $\sigma_1$ and $\sigma_2$ such that $\matches{p}{v}{\sigma_1}$ and $\matches{q}{v}{\sigma_2}$.
  Two patterns $p$ and $q$ are disjoint, written $\disjoint{p}{q}$ if such a value and such substitutions do not exist.
  \label{def:pattern-algebra:disjoint-patterns}
\end{definition}

We write $\deterministic{p}$ for patterns which satisfy the determinism restriction specified in \cref{fig:pattern-algebra:rules}.
Note that there are two rules for or\nobreakdash-patterns:
The rule \textsc{D-Or}$_1$ says that an or\nobreakdash-pattern is deterministic if the subpatterns are disjoint, whereas the rule \textsc{D-Or}$_2$ allows arbitrary deterministic subpatterns as long as they do not contain free variables.
As for the linearity restriction, these two rules for or-patterns give rise to corresponding rules for and-patterns via De Morgan duality: Given that the pattern $\patNeg{(\patOr{p1}{p2})}$ is deterministic, we want the semantically equivalent pattern $\patAnd{\patNeg{p_1}}{\patNeg{p_2}}$ to be deterministic as well.

\begin{theorem}[Matching for wellformed patterns is deterministic, \rocq]
  \label{thm:wf-matching-det}
  For any pattern $p$, value $v$ and substitutions $\sigma_1,\sigma_2$:
  \begin{enumerate}
    \item If $\linearpos{p}$, $\deterministic{p}$ and $\matches{p}{v}{\sigma_1}$ and $\matches{p}{v}{\sigma_2}$, then $\sigma_1 \equiv \sigma_2$.
    \item If $\linearneg{p}$, $\deterministic{p}$ and $\matchesNot{p}{v}{\sigma_1}$, and $\matchesNot{p}{v}{\sigma_2}$, then $\sigma_1 \equiv \sigma_2$.
  \end{enumerate}
\end{theorem}

This property is also necessary to show that the evaluation of terms is deterministic, since we use the relation $\matches{p}{v}{\sigma}$ when we define the evaluation of pattern matches in the next subsection.

\section{Terms and Types}
\label{sec:terms}
In this section we present the second part of our solution: default clauses.
We present a small language, specify its operational semantics and introduce the typing rules for patterns and expressions.
The language contains only data types and pattern matching and omits higher-order functions, but these could be added in the standard way.

\subsection{The Syntax of Expressions}
\label{subsec:terms:expressions}

Expressions $e$ consist of variables $x$, the application of a constructor to arguments $\mathcal{C}(e_1,\ldots,e_n)$ and case expressions.
Every case expression consists of a scrutinee, a list of normal clauses $c$ which all have the form $p \Rightarrow e$, and a default clause $\defaultClause \Rightarrow e$.

\[
  \begin{array}{lclr}
    e & \Coloneqq & x \mid \mathcal{C}(e_1,\ldots,e_n) \mid \caseof{e}{c_1,\ldots,c_n\, ;\, \defaultClause \Rightarrow e} & \emph{Expression} \\
    c & \Coloneqq & p \Rightarrow e & \emph{Clause} \\
  \end{array}
\]

Every case expression contains exactly one default clause, but a realistic language would also allow to omit the default clause if all other clauses are exhaustive.
As we already mentioned, $\defaultClause$ does not belong to the syntax of patterns, but to the syntax of case expressions.
In particular, the meaning of a default clause is not equal to the meaning of a clause $\patWildcard \Rightarrow e_d$ which uses a wildcard pattern, but equal to a clause $\patAnd{\patNeg{p_1}}{\patAnd{\cdots}{\patNeg{p_n}}} \Rightarrow e_d$, where the $p_i$ are all the patterns from the other clauses of the case expression.
For this reason it is not possible to define the meaning of a default clause in isolation, unlike for normal clauses, and a default clause also does not immediately give rise to a valid equation for term rewriting.

\subsection{Operational Semantics}

We formalize a call-by-value semantics in which the only values $v$ are constructors applied to other values.
We use evaluation contexts $E$ and a congruence rule \textsc{E-Cong} to specify evaluation within a subterm, following standard practice \cite{Felleisen1992}.
\[
  \begin{array}{lcl}
       E & \Coloneqq & \square \mid \mathcal{C}(v_1,\ldots,v_n,E,e_1,\ldots,e_m) \mid \caseof{E}{c_1,\ldots,c_n\, ;\, \defaultClause \Rightarrow e}
  \end{array}
\]

We write $e \singlestep^r e'$ for the single-step reduction of the expression $e$ to $e'$ which does not use the congruence rule, and $e \singlestep e'$ for the single step reduction within an evaluation context.
The evaluation of case expressions is governed by the two rules \textsc{E-Match} and \textsc{E-Default}.
The rule \textsc{E-Match} allows to reduce a case expression using any of the normal clauses that match the scrutinee; in particular, it is not necessary to try the clauses in the order in which they occur.
On the right-hand side of that rule we write $e\, \sigma$ for the result of applying the substitution $\sigma$ to the expression $e$.
The rule \textsc{E-Default} allows to reduce a case expression to the expression of the default clause if none of the other clauses match against the scrutinee.

\begin{flushright}
  \fbox{Single-step evaluation: $e \singlestep e$}
\end{flushright}

\begin{prooftree}
  \AxiomC{\matches{p}{v}{\sigma}}
  \RightLabel{\textsc{E-Match}}
  \UnaryInfC{$\mathbf{case}\ v\ \mathbf{of}\ \{ c_1,\ldots,c_n, p \Rightarrow e, c_{n+2},\ldots, c_m\, ; \, \defaultClause \Rightarrow e_d \} \singlestep^r e\sigma$}
\end{prooftree}

\begin{prooftree}
  \AxiomC{$\matchesNot{p_1}{v}{\sigma_1}\quad \cdots\quad \matchesNot{p_n}{v}{\sigma_n}$}
  \RightLabel{\textsc{E-Default}}
  \UnaryInfC{$\caseof{v}{p_1 \Rightarrow e_1,\ldots,p_n \Rightarrow e_n\, ;\, \defaultClause \Rightarrow e_d} \singlestep^r e_d$}
\end{prooftree}

\begin{prooftree}
  \AxiomC{$e \singlestep^r e'$}
  \RightLabel{\textsc{E-Cong}}
  \UnaryInfC{$E[e] \singlestep E[e']$}
\end{prooftree}

Nothing in these rules guarantees that evaluation is deterministic.
For example, we can write a case expression with overlapping patterns.
If we want evaluation to be deterministic, then we have to check that expressions are wellformed, written as $\wf{e}$.
For this we have to check that all patterns which occur in case expressions are deterministic, positively linear, and that no two clauses in a case expression have overlapping patterns.

\begin{flushright}
  \fbox{Wellformed expressions: $\wf{e}$}
\end{flushright}

\begin{minipage}{0.45\textwidth}
  \begin{prooftree}
    \AxiomC{\phantom{$\wf{e}$}}
    \RightLabel{\textsc{Wf-Var}}
    \UnaryInfC{$\wf{x}$}
  \end{prooftree}
\end{minipage}
\hfill
\begin{minipage}{0.45\textwidth}
  \begin{prooftree}
    \AxiomC{$\wf{e_1}\cdots \wf{e_n}$}
    \RightLabel{\textsc{Wf-Ctor}}
    \UnaryInfC{$\wf{\mathcal{C}(e_1,\ldots,e_n)}$}
  \end{prooftree}
\end{minipage}

\begin{prooftree}
  \AxiomC{$\wf{\{e,e_d,e_1,\ldots,e_n\}}$}
  \AxiomC{$\deterministic{p_i}\quad
  \linearpos{p_i}\quad
  \text{ if } i \neq j \text{ then }\disjoint{p_i}{p_j}$}
  \RightLabel{\textsc{Wf-Case}}
  \BinaryInfC{$\wf{\caseof{e}{p_1 \Rightarrow e_1,\ldots, p_n \Rightarrow e_n\, ;\, \defaultClause \Rightarrow e_d}}$}
\end{prooftree}

The evaluation of wellformed expressions is deterministic.
Furthermore, to check whether an expression is wellformed does not depend on any type information.

\begin{theorem}[Evaluation of wellformed expressions is deterministic]
  \label{thm:wf:deterministic}
  For all expressions $e$ which are wellformed (that is $\wf{e}$ holds), evaluation is deterministic:
  If $e \singlestep e_1$ and $e \singlestep e_2$, then $e_1 = e_2$.
\end{theorem}

Multi-step evaluation $\multistep$ is defined as the reflexive-transitive closure of single-step evaluation and allows us to define when two closed expressions are semantically equivalent:
\begin{definition}[Semantic Equivalence of Expressions]
  \label{def:semantic-equivalence-expressions}
  Two closed expressions $e$ and $e'$ are semantically equivalent, written $\semantic{e} \equiv \semantic{e'}$, if either both of them evaluate to the same value $v$, i.e. $e \multistep v$ and $e' \multistep v$, or if both expressions diverge.
\end{definition}

\begin{theorem}[Permutation of clauses does not change semantics]
  The meaning of a wellformed pattern match does not change when we permute clauses.
  If $\overline{c'}$ is a permutation of the list of clauses $\overline{c}$, then
  $\llbracket \mathbf{case}\ v\ \mathbf{of}\ \{ \overline{c}\ ; \mathbf{default} \Rightarrow e_d \} \rrbracket  \equiv \llbracket \mathbf{case}\ v\ \mathbf{of}\ \{ \overline{c'}\ ; \mathbf{default} \Rightarrow e_d \} \rrbracket$.
\end{theorem}

This theorem makes formal what we call pure no-overlap semantics for pattern matching.
The next theorem guarantees that we can use algebraic principles when reasoning about patterns in clauses.

\begin{theorem}[Substitution of equivalent patterns preserves semantics]
  If $\llbracket p_1 \rrbracket  \equiv \llbracket p_2 \rrbracket$, then
  \begin{equation*}
    \llbracket \mathbf{case}\ v\ \mathbf{of}\ \{p_1 \Rightarrow e, \overline{c}\ ; \mathbf{default} \Rightarrow e_d \} \rrbracket
    \equiv
    \llbracket \mathbf{case}\ v\ \mathbf{of}\ \{p_2 \Rightarrow e, \overline{c}\ ; \mathbf{default} \Rightarrow e_d \} \rrbracket.
  \end{equation*}
\end{theorem}

\subsection{Typing Rules}
\label{subsec:terms:typing}

In this section, we are going to introduce types, typing contexts and (pattern) typing rules.
In order to illustrate the typing rules we include three simple data types: the boolean type $\mathbb{B}$ with constructors $\tmTrue$ and $\tmFalse$, the type of pairs $\tau \otimes \tau$ with the constructor $\mathtt{Pair}$, and the type of binary sums $\tau \oplus \tau$ with constructors $\iota_1$ and $\iota_2$.

\[
  \begin{array}{lclr}
       \tau & \Coloneqq & \mathbb{B} \mid \tau \otimes \tau \mid \tau \oplus \tau & \emph{Types} \\
       \Gamma, \Delta & \Coloneqq & \cdot \mid \Gamma, x : \tau & \emph{Typing contexts}
  \end{array}
\]

We use the judgment \patTyping{\Gamma}{\Delta}{p}{\tau} to express that the pattern $p$ matches against a value of type $\tau$ and binds variables whose types are recorded in the context $\Gamma$ on a successful match, and in context $\Delta$ in the case of a non-successful match.
In other words, the typing context $\Gamma$ is used for typing the pattern variables that occur under an even number of negations, and the typing context $\Delta$ is used for typing the variables under an odd number of negations.
We write $\Gamma_1 \mdoubleplus \Gamma_2$ for the naive concatenation of two typing contexts.

  \begin{flushright}
    \fbox{Pattern typing: \patTyping{\Gamma}{\Delta}{p}{\tau}}
  \end{flushright}
  \vspace{0.2cm}

  \begin{minipage}{0.3\textwidth}
  \begin{prooftree}
    \AxiomC{}
    \RightLabel{\textsc{P-Var}}
    \UnaryInfC{$\patTyping{x : \tau}{\cdot}{x}{\tau}$}
  \end{prooftree}
  \end{minipage}
  \begin{minipage}{0.3\textwidth}
    \begin{prooftree}
      \AxiomC{}
      \RightLabel{\textsc{P-Wildcard}}
      \UnaryInfC{$\patTyping{\cdot}{\cdot}{\patWildcard}{\tau}$}
    \end{prooftree}
  \end{minipage}
  \begin{minipage}{0.3\textwidth}
    \begin{prooftree}
      \AxiomC{}
      \RightLabel{\textsc{P-Absurd}}
      \UnaryInfC{$\patTyping{\cdot}{\cdot}{\patBot}{\tau}$}
    \end{prooftree}
  \end{minipage}
  \vspace{0.2cm}

  \begin{minipage}{0.45\textwidth}
  \begin{prooftree}
    \AxiomC{$\patTyping{\Gamma_1}{\Delta}{p_1}{\tau} \quad \patTyping{\Gamma_2}{\Delta}{p_2}{\tau}$}
    \RightLabel{\textsc{P-And}}
    \UnaryInfC{$\patTyping{\Gamma_1 \mdoubleplus \Gamma_2}{\Delta}{\patAnd{p_1}{p_2}}{\tau}$}
  \end{prooftree}
  \end{minipage}
  \hfill
  \begin{minipage}{0.45\textwidth}
  \begin{prooftree}
    \AxiomC{$\patTyping{\Gamma}{\Delta_1}{p_1}{\tau} \quad \patTyping{\Gamma}{\Delta_2}{p_2}{\tau}$}
    \RightLabel{\textsc{P-Or}}
    \UnaryInfC{$\patTyping{\Gamma}{\Delta_1 \mdoubleplus \Delta_2}{\patOr{p_1}{p_2}}{\tau}$}
  \end{prooftree}
  \end{minipage}
  \vspace{0.2cm}

  \begin{minipage}{0.3\textwidth}
    \begin{prooftree}
      \AxiomC{$\patTyping{\Gamma}{\Delta}{p}{\tau}$}
      \RightLabel{\textsc{P-Neg}}
      \UnaryInfC{$\patTyping{\Delta}{\Gamma}{\patNeg{p}}{\tau}$}
    \end{prooftree}
  \end{minipage}
  \begin{minipage}{0.3\textwidth}
    \begin{prooftree}
      \AxiomC{}
      \RightLabel{\textsc{P-True}}
      \UnaryInfC{$\patTyping{\cdot}{\cdot}{\tmTrue}{\tyBool}$}
    \end{prooftree}
  \end{minipage}
  \begin{minipage}{0.3\textwidth}
    \begin{prooftree}
      \AxiomC{}
      \RightLabel{\textsc{P-False}}
      \UnaryInfC{$\patTyping{\cdot}{\cdot}{\tmFalse}{\tyBool}$}
    \end{prooftree}
  \end{minipage}
  \vspace{0.2cm}

  \begin{minipage}{0.4\textwidth}
    \begin{prooftree}
      \AxiomC{$\patTyping{\Gamma}{\Delta}{p}{\tau_1}$}
      \RightLabel{\textsc{P-Inl}}
      \UnaryInfC{$\patTyping{\Gamma}{\Delta}{\tmInl{p}}{\tySum{\tau_1}{\tau_2}}$}
    \end{prooftree}
  \end{minipage}
  \begin{minipage}{0.4\textwidth}
    \begin{prooftree}
      \AxiomC{$\patTyping{\Gamma}{\Delta}{p}{\tau_2}$}
      \RightLabel{\textsc{P-Inr}}
      \UnaryInfC{$\patTyping{\Gamma}{\Delta}{\tmInr{p}}{\tySum{\tau_1}{\tau_2}}$}
    \end{prooftree}
  \end{minipage}

  \begin{prooftree}
    \AxiomC{$\patTyping{\Gamma_1}{\Delta_1}{p_1}{\tau_1}$}
    \AxiomC{$\patTyping{\Gamma_2}{\Delta_2}{p_2}{\tau_2}$}
    \RightLabel{\textsc{P-Pair}}
    \BinaryInfC{$\patTyping{\Gamma_1 \mdoubleplus \Gamma_2}{\Delta_1 \mdoubleplus \Delta_2}{\tmPair{p_1}{p_2}}{\tyPair{\tau_1}{\tau_2}}$}
  \end{prooftree}

These pattern typing rules do not make sense on their own, but only in conjunction with the rules for linearity.
For example, the typing rules allow to derive the typing judgment $x : \tau_1, x: \tau_2 ; \cdot \Rightarrow \tmPair{x}{x} : \tau_1 \otimes \tau_2$, but the linearity rules disallow this pattern\footnote{\
  An implementation might choose to check linearity and typing of patterns simultaneously, but we prefer to present them separately to illustrate that we can also use algebraic patterns in the untyped setting.
}.
Lastly, we present the rules for typing expressions:

\begin{flushright}
  \fbox{Expression typing: $\Gamma \vdash e : \tau$}
\end{flushright}

\begin{prooftree}
  \AxiomC{$\Gamma \vdash e : \sigma$}
  \AxiomC{$\Gamma \vdash e_d : \tau$}
  \AxiomC{$\forall i :\ \patTyping{\Gamma_i}{\Delta}{p_i}{\sigma} \text{ and } \Gamma \mdoubleplus \Gamma_i \vdash e_i : \tau$}
  \RightLabel{\textsc{T-Match}}
  \TrinaryInfC{$\Gamma \vdash \caseof{e}{p_1 \Rightarrow e_1,\ldots, p_n \Rightarrow e_n\, ;\, \defaultClause \Rightarrow e_d} : \tau$}
\end{prooftree}

\begin{minipage}{0.3\textwidth}
  \begin{prooftree}
    \AxiomC{$\Gamma(x) = \tau$}
    \RightLabel{\textsc{T-Var}}
    \UnaryInfC{$\Gamma \vdash x : \tau$}
  \end{prooftree}
\end{minipage}
\begin{minipage}{0.3\textwidth}
  \begin{prooftree}
    \AxiomC{}
    \RightLabel{\textsc{T-True}}
    \UnaryInfC{$\Gamma \vdash \tmTrue : \tyBool$}
  \end{prooftree}
\end{minipage}
\hfill
\begin{minipage}{0.3\textwidth}
  \begin{prooftree}
    \AxiomC{}
    \RightLabel{\textsc{T-False}}
    \UnaryInfC{$\Gamma \vdash \tmFalse : \tyBool$}
  \end{prooftree}
\end{minipage}
\vspace{0.1cm}

\noindent
\begin{minipage}{0.41\textwidth}
  \begin{prooftree}
    \AxiomC{$\Gamma \vdash e_1 : \sigma \quad \Gamma \vdash e_2 : \tau$}
    \RightLabel{\textsc{T-Pair}}
    \UnaryInfC{$\Gamma \vdash \tmPair{e_1}{e_2} : \tyPair{\sigma}{\tau}$}
  \end{prooftree}
\end{minipage}
\begin{minipage}{0.272\textwidth}
  \begin{prooftree}
    \AxiomC{$\Gamma \vdash e : \sigma$}
    \RightLabel{\textsc{T-Inl}}
    \UnaryInfC{$\Gamma \vdash \tmInl{e} : \tySum{\sigma}{\tau}$}
  \end{prooftree}
\end{minipage}
\begin{minipage}{0.274\textwidth}
  \begin{prooftree}
    \AxiomC{$\Gamma \vdash e : \tau$}
    \RightLabel{\textsc{T-Inr}}
    \UnaryInfC{$\Gamma \vdash \tmInr{e} : \tySum{\sigma}{\tau}$}
  \end{prooftree}
\end{minipage}
\vspace{0.2cm}

In the rule \textsc{T-Match} we use the pattern typing rule for the pattern in each clause, and add the resulting typing context $\Gamma_i$ which contains the types of all pattern variables under an even number of negations to the outer typing context $\Gamma$ in the respective clause.

\section{Normalizing Algebraic Patterns}
\label{sec:normalization}
So far, we have considered algebraic patterns in their unrestricted form which is written by the programmer and checked by the compiler.
In this section, we consider normal forms for patterns which simplify the compilation algorithm in \cref{sec:compilation}.
We end up with a  normal form that is a variant of disjunctive normal forms, but which takes the use of constructors and variables into account.
We describe a normalization procedure which proceeds in three steps:
The first step, described in \cref{subsec:normalization:nnf}, pushes negation-patterns inward so that they only occur applied to constructors without subpatterns, or to variables.
The second step, described in \cref{subsec:normalization:dnf}, pushes or-patterns outwards, so that we end up with a disjunctive normal form, that is a disjunction of elementary conjuncts.
In the third step, we further simplify these elementary conjuncts; this is described in \cref{subsec:normalization:conjuncts}.
The default clause is not affected by the normalization procedure described in this section; the compilation of case expressions to decision trees described in the next section can handle default clauses natively.

\subsection{Computing the Negation Normal Form}
\label{subsec:normalization:nnf}

Consider the pattern $\patNeg{\tmPair{\tmTrue}{\tmFalse}}$ which matches anything except a tuple consisting of the boolean values $\tmTrue$ and $\tmFalse$.
Another way to express the behavior of this pattern is to say that it matches anything which isn't a tuple, or it matches a tuple if either the left element isn't the value $\tmTrue$ or the right element isn't the value $\tmFalse$.
We can therefore also write the pattern as $\patOr{\patNeg{\tmPair{\patWildcard}{\patWildcard}}}{\patOr{\tmPair{\patNeg{\tmTrue}}{\patWildcard}}{\tmPair{\patWildcard}{\patNeg{\tmFalse}}}}$.
We can generalize this observation and use it to push all negation-patterns inside, by repeatedly rewriting the pattern using the last equation of \cref{thm:algebraic-equivalences:four}.
After we have normalized patterns in this way, negation-patterns only occur applied to variables or to constructors where all subpatterns of the constructor are wildcard patterns.
We abbreviate these negated constructor patterns $\patNeg{\mathcal{C}^n(\patWildcard,\ldots,\patWildcard})$ as $\patNeg{\mathcal{C}^n}$.

\begin{definition}[Negation Normal Forms]
    The syntax of negation normal forms is:
    \[
    \begin{array}{lcl}
      N & \Coloneqq & x \mid \patNeg{x} \mid \patNeg{\mathcal{C}^n} \mid \mathcal{C}^n(N_1,\ldots,N_n) \mid \patAnd{N}{N} \mid \patOr{N}{N} \mid \patWildcard \mid \patBot
    \end{array}
    \]
    \begin{gather*}
      \begin{array}{rlrl}
      \nnfneg{\patAnd{p_1}{p_2}} &\coloneqq \patOr{\nnfneg{p_1}}{\nnfneg{p_2}} & \nnfpos{\patAnd{p_1}{p_2}} &\coloneqq \patAnd{\nnfpos{p_1}}{\nnfpos{p_2}} \\
      \nnfneg{\patOr{p_1}{p_2}} &\coloneqq \patAnd{\nnfneg{p_1}}{\nnfneg{p_2}} & \nnfpos{\patOr{p_1}{p_2}} &\coloneqq \patOr{\nnfpos{p_1}}{\nnfpos{p_2}} \\
      \nnfneg{\patNeg{p}} &\coloneqq \nnfpos{p} & \nnfpos{\patNeg{p}} &\coloneqq \nnfneg{p} \\
      \nnfneg{\patBot} &\coloneqq \patWildcard & \nnfpos{\patBot} &\coloneqq \patBot \\
      \nnfneg{\patWildcard} &\coloneqq \patBot & \nnfpos{\patWildcard} &\coloneqq \patWildcard \\
      \nnfneg{x} &\coloneqq \patNeg{x} & \nnfpos{x} &\coloneqq x \\
      \end{array} \\
      \begin{array}{rl}
      \nnfneg{\mathcal{C}^n(p_1,\ldots,p_n)} &\coloneqq \patOr{\patNeg{\mathcal{C}^n}}{\patOr{\mathcal{C}^n(\nnfneg{p_1},\ldots,\patWildcard)}{\patOr{\ldots}{\mathcal{C}^n(\patWildcard,\ldots,\nnfneg{p_n})}}} \\
      \nnfpos{\mathcal{C}^n(p_1,\ldots,p_n)} &\coloneqq \mathcal{C}^n(\nnfpos{p_1},\ldots,\nnfpos{p_n}) \\
      \end{array}
  \end{gather*}
\end{definition}

The normalization procedure $\nnf{-}$ is split into the two functions $\nnfneg{-}$ and $\nnfpos{-}$ which track if we are currently normalizing a negated or a non-negated pattern.
At the top-level, the function $\nnf{-}$ therefore just invokes $\nnfpos{-}$.

\begin{example}
    We compute the negation normal form of our running example.
    \begin{equation*}
        \nnfpos{\patAnd{x}{(\patOr{\tmSa}{\tmSu})}} = \patAnd{x}{(\patOr{\tmSa}{\tmSu})}
        \qquad
        \nnfpos{\patAnd{x}{\patNeg{(\patOr{\tmSa}{\tmSu})}}} = \patAnd{x}{\patAnd{(\patNeg{\tmSa})}{(\patNeg{\tmSu})}}
    \end{equation*}
\end{example}

We treat the negation normal forms $N$ as a subset of patterns $p$, and hence all functions and properties of patterns can also be applied to normal forms.
If we translate a pattern into negation normal form then the pattern is still wellformed, can be typed using the same context and type, and has the same semantics as before.
This is witnessed by the following lemma.

\begin{lemma}[Negation Normalization Preserves Linearity, Typing and Semantics]
    For all patterns $p$, contexts $\Gamma, \Delta$ and types $\tau$, we have:
    \begin{enumerate}
        \item Linearity is preserved: If $\linearpos{p}$, then $\linearpos{\nnf{p}}$.
        \item Typing is preserved: If $\patTyping{\Gamma}{\Delta}{p}{\tau}$, then $\patTyping{\Gamma}{\Delta}{\nnf{p}}{\tau}$.
        \item The patterns are semantically equivalent: If $\linearposneg{p}$, then $\semantic{p} \equiv \semantic{\nnf{p}}$.
    \end{enumerate}
    \label{lem:nnf:correctness}
\end{lemma}
\begin{proof}
  See \cref{sec:appendix-proofs}.
\end{proof}

\subsection{Computing the Disjunctive Normal Form}
\label{subsec:normalization:dnf}
The next step is to bring the negation normal form into a disjunctive normal form.
Patterns in disjunctive normal form consist of an outer disjunction of so-called ``elementary conjuncts'' which do not contain or-patterns.
The syntax of the disjunctive normal forms $D$ and elementary conjuncts $K$ is given below.
We can normalize a pattern in negation normal form to the disjunctive normal form by the function $\dnf{-}$ which takes a pattern in negation normal form and returns a set of elementary conjuncts. At the top-level we have to wrap the result in a $\parallel \{ \ldots \}$ node.

\begin{definition}[Disjunctive Normal Form]
    \[
    \begin{array}{lclr}
      D & \Coloneqq & \parallel \{ K_1,\ldots, K_n \} & \emph{Disjunctive Normal Form} \\
      K & \Coloneqq & x \mid \patNeg{x} \mid \patNeg{\mathcal{C}^n} \mid \mathcal{C}^n(K_1,\ldots,K_n) \mid \patAnd{K}{K} \mid \patWildcard \mid \patBot & \emph{Elementary Conjunct} \\
    \end{array}
  \]
  \begin{gather*}
    \dnf{x} \coloneqq \lbrace x \rbrace \mkern15mu
    \dnf{\patNeg{x}} \coloneqq \lbrace \patNeg{x} \rbrace \mkern15mu
    \dnf{\patNeg{\mathcal{C}^n}} \coloneqq \lbrace \patNeg{\mathcal{C}^n} \rbrace \mkern15mu
    \dnf{\patWildcard} \coloneqq \lbrace \patWildcard \rbrace \mkern15mu
    \dnf{\patBot} \coloneqq \lbrace \patBot \rbrace \\
    \begin{array}{rcl}
      \dnf{\mathcal{C}^n(N_1,\ldots,N_n)} &\coloneqq& \lbrace \mathcal{C}(k_1,\ldots,k_n) \mid k_1 \in \dnf{N_1},\ldots, k_n \in \dnf{N_n} \rbrace \\
    \dnf{\patAnd{N_1}{N_2}} &\coloneqq& \lbrace \patAnd{k_1}{k_2} \mid k_1 \in \dnf{N_1}, k_2 \in \dnf{N_2} \rbrace \\
    \dnf{\patOr{N_1}{N_2}} &\coloneqq& \dnf{N_1} \cup \dnf{N_2}
    \end{array}
  \end{gather*}
  \label{def:normalization:dnf}
\end{definition}

\begin{example}
    We compute the disjunctive normal form for our running example.
    \begin{equation*}
        \dnf{\patAnd{x}{(\patOr{\tmSa}{\tmSu})}} =\  \parallel \{\patAnd{x}{\tmSa},\patAnd{x}{\tmSu} \}
        \qquad
        \dnf{\patAnd{x}{\patAnd{(\patNeg{\tmSa})}{(\patNeg{\tmSu})}}} =\ \parallel \{ \patAnd{x}{\patAnd{(\patNeg{\tmSa})}{(\patNeg{\tmSu})}} \}
    \end{equation*}
\end{example}

We can embed the disjunctive normal form $\parallel \{ K_1,\ldots,K_n \}$ into patterns as $\patOr{K_1}{\patOr{\ldots}{K_n}}$. This set of conjuncts is by design never empty ($n > 0$), however, if we do admit this case ($n = 0$) it translates to the absurd-pattern $\patBot$ and is treated as such in any use case (see \cref{sec:compilation,sec:appendix-exhaustiveness}).

\subsection{Normalizing Elementary Conjuncts}
\label{subsec:normalization:conjuncts}

As a last step, we further simplify the elementary conjuncts of \cref{subsec:normalization:dnf}.
We write $\overline{D}$ and $\overline{K}$ for these normalized disjunctive normal forms and normalized elementary conjuncts:

\begin{definition}[Normalized Disjunctive Normal Form]
  \[
    \begin{array}{rl}
      \overline{D} & \Coloneqq \mkern8mu \parallel\lbrace \overline{K}_1, \ldots, \overline{K}_n \rbrace \\
      \overline{K} & \Coloneqq \mkern6mu \patAnd{\lbrace x_1,\ldots,x_n \rbrace}{\mathcal{C}(\overline{K}_1,\ldots, \overline{K}_m)} \mid \patAnd{\lbrace x_1,\ldots, x_n \rbrace}{\patNeg{\lbrace \mathcal{C}_1,\ldots,\mathcal{C}_m \rbrace}} \mid \patAnd{\{ x_1,\ldots,x_n \}}{\patBot}\\
    \end{array}
  \]
  \label{def:normalization:ndnf}
\end{definition}

These normalized conjuncts follow from the observation that there are essentially only three different kinds of conjuncts: positive, negative and unsatisfiable.
Each of these three conjuncts can bind a set of variables $\{ x_1,\ldots, x_n \}$.
A \emph{positive conjunct} $\patAnd{\lbrace x_1,\ldots,x_n \rbrace}{\mathcal{C}(\overline{K}_1,\ldots, \overline{K}_m)}$ matches against any value headed by the constructor $\mathcal{C}$, and binds that value against all the variables in the set.
A \emph{negative conjunct} $\patAnd{\lbrace x_1,\ldots, x_n \rbrace}{\patNeg{\lbrace \mathcal{C}_1,\ldots,\mathcal{C}_m \rbrace}}$ matches against any value which is headed by a constructor which is \emph{not} in the list $\lbrace \mathcal{C}_1,\ldots,\mathcal{C}_m \rbrace$.
The negative conjunct $\patAnd{\lbrace x \rbrace}{\patNeg{\lbrace\rbrace}}$ can be used to represent a single pattern variable $x$, and similarly $\patAnd{\lbrace \rbrace}{\patNeg{\lbrace\rbrace}}$ can be used to represent a wildcard pattern.
An \emph{unsatisfiable conjunct} $\patAnd{\{ x_1,\ldots,x_n \}}{\patBot}$ never matches against a value.

Computing these normalized conjuncts consists in collecting all the variables, as well as all negated and non-negated head constructors of an elementary conjunct.
Intuitively, the head constructors are constructors that occur in a pattern without being nested in a subpattern of another constructor.
An elementary conjunct is unsatisfiable if two different non-negated constructors appear in the conjunct, or if a head constructor appears both negated and non-negated.
An elementary conjunct can be simplified to a negative conjunct if only negated head constructors appear in the original conjunct.
Lastly, an elementary conjunct can be simplified to a positive conjunct if neither of the two cases above apply.
We write $\normalize{-}$ for the function which computes these normalized conjuncts; the formal definition of this function is given in \cref{sec:appendix-normalization}.
We also remove negated variables in this step and treat them like absurd patterns; this is, strictly speaking, not semantics preserving, but negated variables cannot be used in the right-hand side of a clause anyway, so this has no effect on the compilation of programs.

\begin{example}
    We compute the normalized disjunctive normal form for our running example.
    \begin{align*}
        \normalize{\parallel \{\patAnd{x}{\tmSa},\patAnd{x}{\tmSu} \}} &=\ \parallel \{\patAnd{\{ x \}}{\tmSa},\patAnd{\{ x \} }{\tmSu} \}\\
        \normalize{\parallel \{\patAnd{x}{\patAnd{(\patNeg{\tmSa})}{(\patNeg{\tmSu})}} \}} &=\ \parallel \{\patAnd{\{ x \}}{\patNeg{ \{ \tmSa, \tmSu \}}} \} \\
    \end{align*}
\end{example}
Like in \cref{subsec:normalization:dnf}, we treat the normal forms $\overline{D}$ and $\overline{K}$ as subsets of patterns $p$.
We use patterns in normalized disjunctive normal form in the next section where we define a pattern matching compilation algorithm for nested patterns.

\section{Compiling Algebraic Patterns}
\label{sec:compilation}

We have shown that algebraic patterns are more expressive because they explicitly allow us to speak about the complement of a pattern.
But programmers won't use them if they have to pay for the gain in expressivity with slower programs.
In order to compile programs using these patterns to efficient code we therefore have to use and modify one of the available pattern matching compilation algorithms.

These algorithms can be divided into two general classes:
Algorithms which use backtracking \cite{Augustsson1985} and algorithms which compile to decision trees \cite{Maranget2008}.
The algorithms which use backtracking guarantee that the size of the generated code is linear in the size of the input, but the generated code might perform redundant tests.
Algorithms which compile to decision trees never examine scrutinees more than once but the generated code might be exponential in the size of the input.
Both classes of algorithms are competitive when they are optimized \cite{LeFessant2001,Maranget2008}, so it is not possible to say which alternative is preferable in general.

In this article we modify and extend the compilation scheme laid out by Maranget \cite{Maranget2008} which is based on decision trees.
His algorithm already includes a treatment of or-patterns, and we will extend it to the normalized patterns that we introduced in \cref{sec:normalization}.
We introduce the general idea of the compilation to decision trees in \cref{subsec:compilation:general-idea}, present the central branching step in \cref{subsec:compilation:branching} in detail and illustrate the algorithm with an example in \cref{subsec:compilation:example}. We conclude with a statement on the correctness of our extended compilation scheme in \cref{subsec:compilation:correctness}.

\subsection{General Idea of the Algorithm}
\label{subsec:compilation:general-idea}

The algorithm operates on a generalization of the case expression $\caseof{v}{p_1 \Rightarrow e_1, \ldots, p_n \Rightarrow e_n ; \defaultClause \Rightarrow e}$ that we introduced in \cref{subsec:terms:expressions}.
Instead of only one scrutinee $v$ and one pattern $p$ in each clause we now allow a vector $v_1,\ldots,v_n$ of scrutinees and a matrix of clauses with patterns in nDNF (see \cref{subsec:normalization:conjuncts}).
\begin{equation}
	\tag{Input}
	\label{eq:case:input}
    \mathbf{case}\ v_1,\ldots, v_n\ \mathbf{of}\
	    \begin{bmatrix}
	        \overline{D}_1^1 	& \ldots 	& \overline{D}_n^1 &\To e^1 \\
			& \vdots	& 				\\
			\overline{D}_1^m	& \ldots	& \overline{D}_n^m &\To e^m \\
			&\defaultClause & &\To e_d
		\end{bmatrix}
\end{equation}
In such multi-column case expressions, the $i$-th pattern column tests the $i$-th scrutinee $v_i$.
The main idea of the algorithm is that we take such a multi-column case expression as input and examine all scrutinees -- and subscrutinees -- one by one, producing a nesting of single-column case expressions with \textit{simple} patterns as output.
These simple patterns consist only of constructors applied to variables.
We start the compilation by embedding an ordinary case expression into this generalized form on which we call the function \textit{compile}.
This function can take one of three steps.
It either terminates with a \textbf{Default} or \textbf{Simple} step, or it examines one of the scrutinees in a \textbf{Branch} step.

\begin{description}
	\item[Default] If $\defaultClause \Rightarrow e_d$ is the only clause then we terminate with the expression $e_d$.
	\item[Simple] If the first row consists only of variable patterns of the form $\patAnd{\lbrace x_1, \ldots x_k\rbrace}{\patNeg{\{\}}}$ then compilation terminates with the right-hand side of that clause together with appropriate substitutions.
	This rule also applies if $n = 0$.
	\item[Branch] If neither of the above steps apply we pick a column $i$ (with $1 \leq i \leq n$) and examine its scrutinee $v_i$ in a single-column case against simple constructor patterns.
\end{description}

All the complexity of the algorithm lies within the \textbf{Branch} step, so we will discuss its details in the next subsection.

\subsection{The Branching Step}
\label{subsec:compilation:branching}

The first step is to choose on which index $i$ we want to perform the case split.
This choice does not affect the correctness of the algorithm; selecting a column is a therefore a matter of optimization which other authors \cite{Scott2000,Maranget2008} have discussed in detail.
We only enforce that the selected pattern column must contain outermost constructors.
Such a column must necessarily exist if neither the \textbf{Default} nor the \textbf{Simple} step apply.

Next, we have to gather the constructors against which we can test the scrutinee $v_i$.
We collect these so-called head constructors into a set $\mathcal{H}(\overline{D}_i)$ defined below.

$$\label{head}
\mathcal{H}(\overline{D}_i) := \hspace{-2mm}\bigcup\limits_{j = 1,\dots,m} \hspace{-2mm}\headconstructors{\overline{D}_i^j} \ \ with \
\begin{array}{lll}
	\headconstructors{\parallel \{ \overline{K}_1,\ldots, \overline{K}_t \}} &\coloneqq& \hspace{-3.25mm}\bigcup\limits_{l=1,\dots,t}
	\hspace{-2mm}\headconstructors{\overline{K}_l}\\
	\headconstructors{\patAnd{\lbrace x_1,\ldots,x_s \rbrace}{\mathcal{C}(\overline{K}_1,\ldots, \overline{K}_t)}} &\coloneqq& \lbrace \mathcal{C} \rbrace \\
	\headconstructors{\patAnd{\lbrace x_1,\ldots, x_s \rbrace}{\patNeg{\lbrace \mathcal{C}_1,\ldots,\mathcal{C}_t \rbrace}}} &\coloneqq& \lbrace \mathcal{C}_1,\ldots,\mathcal{C}_t \rbrace \\
	\headconstructors{\patAnd{\{x_1,\ldots,x_s\}}{\patBot}} &\coloneqq& \emptyset
\end{array}
$$

Once we have found these head constructors $\{ \mathcal{C}^{n_1}, \ldots, \mathcal{C}^{n_z} \}$ we generate fresh pattern variables $x^{k}_1$ to $x^{k}_{n_k}$ for each constructor $\mathcal{C}^{n_k}$ and a simple pattern match.
For each clause we have gained information about the scrutinee $v_i$: Within the right-hand side of a constructor clause, we know that our scrutinee must have conformed to the shape of our constructor.
Likewise, if we arrive at the right-hand side of the default clause, we know that our scrutinee must have been different to all the constructors $\{ \mathcal{C}^{n_1}, \ldots, \mathcal{C}^{n_z} \}$ before.
We use this knowledge by computing clause-specific \textit{subproblems} of our initial \ref{eq:case:input}
and we do so via the function $\mathcal{S}$ for constructor clauses and function $\mathcal{D}$ for the default clause.

\begin{equation*}
	\tag{Output}
	\label{eq:case:output}
    \mathbf{case}\ v_i\ \mathbf{of}\
	    \begin{bmatrix}
	        \mathcal{C}^{n_1}(x_1^{1},\ldots, x_{n_1}^{1}) 	&\To& \mathit{compile}(\mathcal{S}(i,\mathcal{C}^{n_1}(x_1^{1},\ldots, x_{n_1}^{1}),\text{\ref{eq:case:input}})) \\
			\vdots	& 	&  \vdots			\\
			\mathcal{C}^{n_z}(x_1^{z},\ldots, x_{n_z}^{z}) &\To& \mathit{compile}(\mathcal{S}(i,\mathcal{C}^{n_z}(x_1^{z},\ldots, x_{n_z}^{z}),\text{\ref{eq:case:input}})) \\
			\defaultClause &\To& \mathit{compile}(\mathcal{D}(i,\{ \mathcal{C}^{n_1},\ldots, \mathcal{C}^{n_z} \}, \text{\ref{eq:case:input}}))
		\end{bmatrix}
\end{equation*}
All that is left is to compute these subproblems of \textit{specialization} and \textit{default} which are the arguments of recursive invocations of the compile function.

\subsubsection{Computing Specialization}
\label{subsubsec:compilation:branch:specialization}

We now show how to compute $\mathcal{S}(i,\mathcal{C}^{k}(x_1,\ldots, x_{k}),\text{\ref{eq:case:input}})$ which specializes \ref{eq:case:input} into a constructor-specific subproblem that utilizes \textit{the assumption that the scrutinee $v_i$ has already matched the constructor pattern $\mathcal{C}^{k}(x_1,\ldots, x_{k})$}.
We therefore remove the old scrutinee $v_i$ and replace it with the bound variables $x_1,\ldots, x_{k}$ as new scrutinees.
\begin{equation*}
	\mathcal{S}(i,\mathcal{C}^{k}(x_1,\ldots, x_{k}),\text{\ref{eq:case:input}}) \ = \ \mathbf{case}\hspace{2mm}
	x_1, \dots, x_k, v_1,\dots,v_{i-1},v_{i+1},\dots,v_n \hspace{2mm} \mathbf{of} \hspace{1mm} \begin{bmatrix} S \end{bmatrix}
\end{equation*}
We now have to build a new clause matrix $\begin{bmatrix} S \end{bmatrix}$ that matches against the new list of scrutinees and do so by iterating over the previous clauses.
Each clause of the old \ref{eq:case:input} contributes to $\begin{bmatrix} S \end{bmatrix}$ only if its $i$-th pattern could have matched the old scrutinee $v_i$ of shape $\mathcal{C}^{k}$ by \textit{assumption}.
This is, of course, the case for the default clause which we keep unchanged.
Otherwise, these $i$-th patterns are in normalized DNF, and we have to check whether one of its elementary conjuncts could have matched. The following list gives all inclusion rules.

\begin{enumerate}
	\item \label{spec-1} Row $\begin{bmatrix} \mathbf{default} \To e_d \end{bmatrix}$ is added to $\begin{bmatrix} S \end{bmatrix}$ unchanged.
	\item Row
		$\begin{bmatrix}
		\overline{D}_{1}^j \ \ \ldots \ \ \overline{D}_{i-1}^j \ \ \mathbf{\overline{D}_i^j} \ \ \overline{D}_{i+1}^j & \ldots & \overline{D}_{n}^j \ \ \To \ \  e^j \end{bmatrix}$
		only contributes according to $\mathbf{\overline{D}_i^j}$:
		\begin{enumerate}
		\item \label{spec-a} If
				$\mathbf{\overline{D}_i^j} = \patAnd{\{ y_1, \ldots y_s\}}{\mathcal{C}^k(\overline{K}_1,\ldots,\overline{K}_k)}$, \\
				then add the row
				$\begin{bmatrix}
				\overline{K}_1 \dots \overline{K}_k \hspace{4mm}
				\overline{D}_{1}^j \dots \overline{D}_{i-1}^j \overline{D}_{i+1}^j \dots \ \overline{D}_{n}^j \hspace{3mm}
				\To e^j [\{ y_1, \ldots y_s\} \mapsto v_i]
				\end{bmatrix}$. \vspace{2mm}
			\item \label{spec-b} If
				$\mathbf{\overline{D}_i^j} = \patAnd{\{ y_1, \ldots y_s\}}{\patNeg{\{ \mathcal{C}_1, \ldots, \mathcal{C}_t \} }}$
				and
				$\mathcal{C}^k \not\in \{ \mathcal{C}_1, \ldots, \mathcal{C}_t \}$,\\
				then add the row
				$\begin{bmatrix}
					\ \patWildcard \ \ \dots \ \ \patWildcard \hspace{4mm}
					\overline{D}_{1}^j \dots \overline{D}_{i-1}^j \overline{D}_{i+1}^j \dots \ \overline{D}_{n}^j \hspace{3mm}
					\To e^j [\{ y_1, \ldots y_s\} \mapsto v_i]
				\end{bmatrix}$. \\
				(Note that we use $\patWildcard$ simply as a shorthand for the equivalent nDNF $\patAnd{\{\}}{\patNeg{\{\}}}$)\\
			\item \label{spec-c} If
				$\mathbf{\overline{D}_i^j} = \patOr{}{ \{ \overline{K}_1,\ldots,\overline{K}_t \}}$, then \vspace{-0.45cm} \\
				recursively iterate the rows of
				$\begin{bmatrix}
					\overline{D}_{1}^j & \ldots & \overline{D}_{i-1}^j & \mathbf{\overline{K}}_1 & \overline{D}_{i+1}^j & \ldots & \overline{D}_{n}^j & \To & e^j \\
					\multicolumn{7}{c}{\vdots} \\
					\overline{D}_{1}^j & \ldots & \overline{D}_{i-1}^j & \mathbf{\overline{K}}_t & \overline{D}_{i+1}^j & \ldots & \overline{D}_{n}^j & \To & e^j
				\end{bmatrix}$.
	\end{enumerate}
\end{enumerate}

We give an example for case \ref{spec-b} which handles negation.
Suppose that want to specialize the following case expression into a case expression \textit{subproblem} which assumes that $v_1$ was matched against $\mathtt{Cons}(x,y)$.
$$\setlength\arraycolsep{2pt}
compile\left(
\begin{array}{l}
	\ \textbf{case} \ v_1,v_2 \ \textbf{of} \vspace{1mm}\\
	\begin{bmatrix}
	  \patAnd{\{\}}{\patNeg{\{\mathtt{Cons}\}}} & \overline{D}^1_2 &\To e^1 \\
	  \patAnd{\{\}}{\patNeg{\{\mathtt{Nil}\}}} & \overline{D}^2_2 &\To e^2 \\
	  \textbf{default} & &\To e_d
	\end{bmatrix}
\end{array}
\right) \ \ \ = \ \ \begin{array}{l}
	\ \textbf{case}
	\ v_1 \ \textbf{of} \vspace{1mm}\\
	\begin{bmatrix}
		\mathtt{Cons}(x,y) &\To &compile(\textit{subproblem})\\
		&\vdots&\\
	\end{bmatrix}\end{array}
$$
The first clause is discarded, because the pattern $\patAnd{\{\}}{\patNeg{\{\mathtt{Cons}\}}}$ is incompatible with the assumption that $v_i$ matches $\mathtt{Cons}$.
The second clause using $\patAnd{\{\}}{\patNeg{\{\mathtt{Nil}\}}}$ is compatible with the assumption.
So, how can we transform this clause into a new one testing $x,y$ and $v_2$? We see that pattern $\patAnd{\{\}}{\patNeg{\{\mathtt{Nil}\}}}$ does not match on the subvalues $x$ and $y$ of $v_1$, but for $v_2$ we still have to match against the pattern $\overline{D}^2_2$.
We keep the default clause by rule \ref{spec-1} unchanged.
$$
subproblem = \textbf{case}
\ x,y,v_2 \ \textbf{of} \
\begin{bmatrix}
  \patWildcard & \patWildcard & \overline{D}^2_2 &\To e^2 \\
  \textbf{default} 			& &					 &\To e_d
\end{bmatrix}
$$


\subsubsection{Computing Default}
\label{subsubsec:compilation:branch:default}

Next we show how to compute $\mathcal{D}(i,\{ \mathcal{C}^{n_1},\ldots, \mathcal{C}^{n_z} \}, \text{\ref{eq:case:input}})$ whose output reflects \textit{the assumption that the scrutinee $v_i$ did not match any of the head constructors $\mathcal{C}^{n_1}$ to $\mathcal{C}^{n_z}$.}
This time we have no new scrutinees and only match the leftover list against a new pattern matrix $\begin{bmatrix} D
\end{bmatrix}$:
\begin{equation*}
	\mathcal{D}(i,\{ \mathcal{C}^{n_1},\ldots, \mathcal{C}^{n_z} \}, \text{\ref{eq:case:input}}) \ = \ \mathbf{case}\hspace{2mm}
	v_1,\dots,v_{i-1},v_{i+1},\dots,v_n \hspace{2mm} \mathbf{of} \hspace{1mm} \begin{bmatrix} D \end{bmatrix}
\end{equation*}
As before, we build a new matrix $\begin{bmatrix} D
\end{bmatrix}$ by row-wise iteration over the clauses of the \ref{eq:case:input}.
A clause only contributes if its $i$-th pattern entry allows for a match with scrutinee $v_i$, knowing that it differs from all the head constructors $C^{n_1}, \dots, C^{n_z}$.

\begin{enumerate}
	\item \label{def-1}
	Row $\begin{bmatrix} \mathbf{default} \To e_d \end{bmatrix}$ is added to $\begin{bmatrix} D \end{bmatrix}$ unchanged.
	\item
	Row	$\begin{bmatrix}
			\overline{D}_{1}^j \ \ \ldots \ \ \overline{D}_{i-1}^j \ \ \mathbf{\overline{D}_i^j} \ \ \overline{D}_{i+1}^j & \ldots & \overline{D}_{n}^j \ \ \To \ \  e^j \end{bmatrix}$,
			only contributes according to $\mathbf{\overline{D}_i^j}$:
		\begin{enumerate}
			\item \label{def-a}
				If
				$\mathbf{\overline{D}_i^j} = \patAnd{\{ y_1, \ldots, y_s \}}{\patNeg{\{ \mathcal{C}_1, \ldots, \mathcal{C}_t \} }}$, \\
				then add the row
				$\begin{bmatrix}
					\overline{D}_{1}^j \dots \overline{D}_{i-1}^j \overline{D}_{i+1}^j \dots \ \overline{D}_{n}^j \hspace{3mm}
					\To e^j [\{ y_1, \ldots y_s\} \mapsto v_i]
				\end{bmatrix}$. \\
				Note that $v_i$'s constructor differs from all head constructors which include $\{ \mathcal{C}_1, \ldots, \mathcal{C}_t \}$\vspace{3mm}
			\item \label{def-b}
				If
				$\mathbf{\overline{D}_i^j} = \patOr{}{ \{ \overline{K}_1,\ldots,\overline{K}_t \}}$, then \vspace{-0.45cm} \\
				recursively iterate the rows of
				$\begin{bmatrix}
					\overline{D}_{1}^j & \ldots & \overline{D}_{i-1}^j & \mathbf{\overline{K}}_1 & \overline{D}_{i+1}^j & \ldots & \overline{D}_{n}^j & \To & e^j \\
					\multicolumn{7}{c}{\vdots} \\
					\overline{D}_{1}^j & \ldots & \overline{D}_{i-1}^j & \mathbf{\overline{K}}_t & \overline{D}_{i+1}^j & \ldots & \overline{D}_{n}^j & \To & e^j
				\end{bmatrix}$. \vspace{1mm}
		\end{enumerate}
	\end{enumerate}

Let us illustrate rule \ref{def-a} with the case expression below, where we perform a case split on $v_1$.
We want to construct the default \textit{subproblem} which assumes that $v_1$  did not match the constructors $\mathtt{Red}$ or $\mathtt{Green}$.
$$\setlength\arraycolsep{2pt}
compile\left(
\begin{array}{l}
	\ \textbf{case}
	\ v_1,v_2 \ \textbf{of} \vspace{1mm}\\
	\begin{bmatrix}
	\patAnd{\{\}}{\mathtt{Red}} & \overline{D}^1_2 &\To e^1 \\
	\patAnd{\{\}}{\patNeg{\{\mathtt{Green}\}}} & \overline{D}^2_2 &\To e^2 \\
	\textbf{default} & &\To e_d
	\end{bmatrix}
\end{array}
\right) \ \ \ = \ \begin{array}{l}
	\ \textbf{case}
	\ v_1 \ \textbf{of} \vspace{1mm}\\
	\begin{bmatrix}
		\mathtt{\mathtt{Red}}  &\To & \dots\\
		\mathtt{\mathtt{Green}}  &\To & \dots\\
		\textbf{default}	&\To & compile(\textit{subproblem})
	\end{bmatrix}\end{array}
$$
In \textit{subproblem}, we have no new scrutinees and only need to investigate the remaining scrutinee $v_2$
$subproblem = \textbf{case} \ v_2 \ \textbf{of} \
				\begin{bmatrix} D \end{bmatrix}$.
We ask for each clause if it is consistent with $v_1$ not looking like either $\mathtt{Red}$ or $\mathtt{Green}$.
The first clause could not match $\patAnd{\{\}}{\mathtt{Red}}$, so we discard it.
For the second row, however, the test $\patAnd{\{\}}{\patNeg{\{\mathtt{Green}\}}}$ does match any $v_1$ that differs from both $\mathtt{Red}$ and $\mathtt{Green}$, so we keep it.
The default clause is also kept unchanged.
$$
subproblem = \textbf{case}
\ v_2 \ \textbf{of} \
\begin{bmatrix}
   \overline{D}^2_2 &\To e^2 \\
  \textbf{default} 	&\To e_d
\end{bmatrix}
$$


\subsection{Compiling an Example}
\label{subsec:compilation:example}

\newcommand{\rhsweekend}[1]{e_1(#1)}
\newcommand{\rhsnotweekend}[1]{e_2(#1)}
\newcommand{\patweekend}{\patOr{(\patAnd{\{y\}}{\tmSa})}{(\patAnd{\{y\}}{\tmSu})}}
\newcommand{\patnotweekend}{\patAnd{\{y\}}{\patNeg{\{\tmSa,\tmSu\}}}}

Let us illustrate the compilation process with a variant of the example introduced in \cref{sec:pattern-algebra}.
$$
\begin{array}{llll}
\mathbf{case} \ x \ \mathbf{of} \{
	& \patAnd{y}{(\patOr{\tmSa}{\tmSu})} & \To \text{``Today is weekend!''},&\\
	& \patAnd{y}{\patNeg{(\patOr{\tmFr}{\patOr{\tmSa}{\tmSu}})}}  & \To \text{``Today is ''}\mdoubleplus\ \text{show}(y),&\\
	& \mathbf{default} & \To \text{``Tomorrow is weekend...''} &\ \}
\end{array}
$$
First, we rewrite this case expression into the multi-column notation, normalize all patterns (see \cref{sec:normalization}), and abbreviate the right-hand sides by $e_1(y)$, $e_2(y)$, and $d$.
$$
\mathit{Input} := \mathbf{case} \ x \ \mathbf{of}
\begin{bmatrix*}[l]
	\patOr{}{}\{\{y\} \ \& \ \tmSa, \{y\} \ \& \ \tmSu \} & \To & e_1(y) \\
	\patOr{}{}\{\{y\} \ \& \ \patNeg{}\{\tmFr,\tmSa,\tmSu\}\} & \To & e_2(y) \\
	\mathbf{default} & \To & d
\end{bmatrix*}
$$
We branch on the only available scrutinee $x$ by matching it in a single-column case expression against the head constructors $\{ \tmFr, \tmSa, \tmSu \}$.
We then generate a new case expression with four clauses, one for each constructor.
All that is left to do is to compute and compile the clause-specific \textit{subproblems} of specialization and default.
$$
\mathit{compile}\left( \mathit{Input} \right) = \mathbf{case} \ x \ \mathbf{of} \
\begin{bmatrix}
	&\tmFr &\To \mathit{compile}	\left(
								\mathcal{S} \left(1, \tmFr,
								\mathit{Input}
								\right)
								\right) \vspace{2mm}\\
	&\tmSa &\To \mathit{compile}	\left(
								\mathcal{S} \left(1, \tmSa,
								\mathit{Input}
								\right)
								\right) \vspace{2mm}\\
	&\tmSu &\To \mathit{compile}	\left(
								\mathcal{S} \left(1, \tmSu,
								\mathit{Input}
								\right)
								\right) \vspace{2mm}\\
	& \mathbf{default} &\To \mathit{compile} \left(
								\mathcal{D} \left(1, \{\tmFr,\tmSa,\tmSu\},
								\mathit{Input}
								\right)
								\right)
\end{bmatrix}
$$
Let us consider the subproblem $\mathcal{S} \left(1, \tmSa, \mathit{Input}\right)$.
We know that scrutinee $x$ has matched $\tmSa$, and we have no left-over scrutinees nor new subscrutinees to consider ($n = 0$), giving $\mathcal{S} \left(1, \tmSa, \mathit{Input}\right) = \textbf{case} \ \ \ \textbf{of} \begin{bmatrix} S \end{bmatrix}$.
We get the new clause matrix $\begin{bmatrix} S \end{bmatrix}$ by iterating the old one to see how each row contributes.
\begin{itemize}
	\item For row
		$\begin{bmatrix*}[l]
		\patOr{}{}\{\{y\} \ \& \ \tmSa, \{y\} \ \& \ \tmSu \} \ \To \ e_1(y) \\
		\end{bmatrix*}$, by rule (c), iterate
		$\begin{bmatrix*}[l]
			\{y\} \ \& \ \tmSa & \To & e_1(y) \\
			\{y\} \ \& \ \tmSu & \To & e_1(y) \\
		\end{bmatrix*}$. \vspace{-0.3cm}\\
		\begin{itemize}
			\item For the first one, by (a) we add the row $\begin{bmatrix*}[l]
				 & \To \ e_1(x) \\
			\end{bmatrix*}$.
			\item For the second one, no contribution is made. \vspace{1mm}
		\end{itemize}
	\item For row
		$\begin{bmatrix*}[l]
		\patOr{}{}\{\{y\} \ \& \ \patNeg{}\{\tmFr,\tmSa,\tmSu\}\} & \To & e_2(y) \\
		\end{bmatrix*}$ no contribution is made. \vspace{1mm}
	\item For row $ \begin{bmatrix*}[l]
						\mathbf{default} & \To & d
					\end{bmatrix*}$, by (1), we keep it unchanged.
\end{itemize}
This gives the specialization matrix
	$
	\begin{bmatrix}S\end{bmatrix} = \begin{bmatrix}
		& \To & e_1(x) \\
		\mathbf{default} & \To & d
	\end{bmatrix}
	$.

The other specialization subproblems are computed similarly, and we obtain the following intermediate result.
$$
\mathit{compile}\left( \mathit{Input} \right) =
\mathbf{case} \ x \ \mathbf{of}
\begin{bmatrix*}[l]
	&\tmFr &\To \mathit{compile}	\left(
		\mathbf{case} \ \ \ \ \mathbf{of} \begin{bmatrix}
			\mathbf{default} & \To & d
		\end{bmatrix}
								\right), \vspace{2mm}\\
	&\tmSa &\To \mathit{compile}	\left(
		\mathbf{case} \ \ \ \ \mathbf{of} \begin{bmatrix}
			& \To & e_1(x) \\
			\mathbf{default} & \To & d
		\end{bmatrix}
								\right) \vspace{2mm}\\
	&\tmSu &\To \mathit{compile}	\left(
		\mathbf{case} \ \ \ \ \mathbf{of} \begin{bmatrix}
			& \To & e_1(x) \\
			\mathbf{default} & \To & d
		\end{bmatrix}
								\right) \vspace{2mm}\\
	& \mathbf{default} &\To \mathit{compile} \left(
		\mathbf{case} \ \ \ \ \mathbf{of} \begin{bmatrix}
			& \To & e_2(x) \\
			\mathbf{default} & \To & d
		\end{bmatrix}
								\right)
\end{bmatrix*}
$$
For each of the subproblems we can apply either \textbf{Default} or \textbf{Simple} and conclude the computation.
$$
\mathit{compile}\left( \mathit{Input} \right) =
\mathbf{case} \ x \ \mathbf{of} \
\begin{bmatrix*}[l]
	&\tmFr &\To d\\
	&\tmSa &\To e_1(x)\\
	&\tmSu &\To e_1(x)\\
	& \mathbf{default} &\To e_2(x)
\end{bmatrix*}
$$
Note how the right-hand expression of the previous default clause $d$ is utilized in a constructor-clause, whereas the right-hand expression of the new default-clause originates from a previously non-default clause.

\subsection{Correctness of Compilation}
\label{subsec:compilation:correctness}

Correctness of compilation states that the \textit{compile} function introduced in \cref{sec:compilation} preserves semantics.
Because the compilation algorithm operates on multi-column pattern matches, we extend wellformedness of expressions, the operational semantics and denotations to multi-column pattern matches.
We can then state the correctness for the compilation algorithm:
\begin{theorem}
	\label{thm:compilation:correctness}
    On a wellformed multi-column pattern match of the normalized shape \ref{eq:case:input}, the compilation algorithm
	preserves semantics:
    $\semantic{\ref{eq:case:input}} \equiv \semantic{\textit{compile}(\ref{eq:case:input})}$.
\end{theorem}
We provide the proof and the extended definitions in the \cref{subsec:proofs:compilation-correctness}.

\section{Future Work}
\label{sec:future-work}
There are two main directions that we plan to investigate in future work.

\paragraph*{Allowing Overlapping Patterns}
\label{subsec:future-work:confluent-cases}

In this article we enforce that no two patterns in a pattern match can overlap in order to guarantee confluence of the system, but \cite{Cockx2013,Cockx2014} showed that this restriction can be relaxed.
They allow overlapping patterns under the condition that the right-hand sides of the clauses must become judgmentally equal under a substitution computed from the overlapping patterns.
Let us look at an example that is allowed in their system, but not ours:
\begin{align*}
    \text{or}(x) \coloneqq \mathbf{case}\ x\ \mathbf{of}\ \lbrace\ &\tmPair{\tmTrue}{\patWildcard} \Rightarrow \tmTrue, &&\tmPair{\patWildcard}{\tmTrue} \Rightarrow \tmTrue,\\
    &\tmPair{\tmFalse}{x} \Rightarrow x, &&\tmPair{x}{\tmFalse} \Rightarrow x\ \rbrace
\end{align*}
In this example, the first and the last (as well as the second and first, second and third, and the last two) patterns are overlapping.
But the definition of this function is still confluent, since if we unify the overlapping patterns $\tmPair{\text{True}}{\patWildcard}$ and $\tmPair{x}{\text{False}}$, and apply the substitution to the right-hand sides $\text{True}$ and $x$, the terms become judgmentally equal.
Such overlapping patterns can be useful in a proof assistant, since they yield additional judgmental equalities that can be used to simplify a term during normalization.
We think that their approach to allow overlapping patterns and our algebraic patterns can complement each other, and we therefore plan to study their interaction.

\paragraph*{Exhaustiveness Checking}
\label{subsec:future-work:exhaustiveness-checking}

In this article, we check whether a set of patterns is overlapping, but we do not check whether a set of patterns is exhaustive.
This is in spite of the fact that we are compiling these patterns using decision trees, which is a method that can easily accommodate checking the exhaustiveness of sets of patterns \cite{Maranget2007}.
Checking that a given set of patterns is exhaustive is essential for real programming languages, so we plan to develop and verify an algorithm for this purpose, building on the work of \cite{Maranget2007,Karachalias2015,Graf2020}. We sketch a basic outline for such an algorithm in \cref{sec:appendix-exhaustiveness}.

\section{Related Work}
\label{sec:related-work}
\paragraph*{Order-Independent Pattern Matching}
As discussed in \cref{subsec:future-work:confluent-cases}, \cite{Cockx2013} and \cite{Cockx2014} study the issue of order-independent pattern matching in proof assistants.
They want to ensure that every clause in a pattern-matching expression can be used as a judgmental equality for normalizing terms during type checking.
Their system lacks advanced patterns such as or-patterns and negation patterns, but they instead relax the non-overlap requirement.

\paragraph*{Substructural logics}
And-patterns and wildcard patterns are strongly related to the rules of contraction and weakening in the sequent calculus.
This relationship between structural rules and certain forms of patterns has been observed by \cite{Cerrito2004}, who present a pattern matching calculus for the sequent calculus which includes both and-patterns and wildcard-patterns.
This relationship can also be observed in programming languages with substructural type systems, such as Rust, where the expression $\mathtt{match\ v\ \{ x@y => (x,y)\}}$ does not type check if the type of $v$ does not support contraction, i.e. does not implement the \texttt{Copy} trait.

\paragraph*{Logic Programming and Datalog}
The problems with the binding structure of patterns involving negations that we discussed in this paper appears in a similar form for Datalog programs.
In this article we are only interested in the binding occurrences of variables in a pattern, i.e. those variables which occur under an even number of negations, and which can be used in the right-hand side of a clause.
Occurrences of logic variables in a Datalog program can be distinguished into binding occurrences and bound occurrences, where binding occurrences assign a value to a logic variable and bound occurrences can access the value the variable is bound to.
The authors of \cite{Klopp2024datalog} describe a type system which tracks binding and bound occurences of variables and describe the rules for disjunctions and negations in their Section 4.3.
Their rule for negations is:
\begin{prooftree}
    \AxiomC{$\Gamma_1 \vdash^{\# \cdot -} a\ \mathtt{ok} \dashv \Gamma_2$}
    \RightLabel{A-Not}
    \UnaryInfC{$\Gamma_1 \vdash^{\#} \mathtt{not}\ a\ \mathtt{ok} \dashv \Gamma_2$}
\end{prooftree}
Here $\#$ stands for a polarity, either $+$ or $-$, which controls which variables in the input context $\Gamma_1$ are already bound, and the multiplication with $-$ correspondingly switches this interpretation.
This rule therefore operates very similarly to our pattern typing rule \textsc{P-Neg} in \cref{subsec:terms:typing} and guarantees that negation is involutive with respect to the typing relation.

\paragraph*{Dynamic First Class Patterns}
We use \enquote{dynamic first class patterns} for systems in which patterns can be computed and passed at runtime.
Such systems, whose theory was studied by \cite{Jay2009} and \cite{Jay2009Book}, can also express our algebraic patterns.
The technical report \cite{GellerHirschfeldBracha2010} describes an extension of the object-oriented language Newspeak (cf.~\cite{Bracha2022}) by such a system of dynamic pattern matching.
They explicitly discuss various \emph{pattern combinators}, i.e.~combinators which combine patterns to yield new ones.
They discuss disjunction, conjunction and negation pattern combinators which correspond to our or-patterns, and-patterns and negation patterns.
In this article, we focus exclusively on statically-known patterns, since these can be analyzed and compiled more easily.

\paragraph*{Negation Patterns}

Negation patterns also appear in \cite{Krishnaswami2009}, where the author mentions that \enquote{Having false- and or-patterns allows us to define the complement of a pattern -- for any pattern at a type $A$, we can define another pattern that matches exactly the values of type $A$ the original does not}.
\cite{Krishnaswami2009} presents a translation that eliminates negation patterns, but this translation makes a closed-world assumption about all the available constructors for each data type.
For example, he translates the negated pair pattern $\patNeg{\tmPair{p_1}{p_2}}$ by the pattern $\patOr{\tmPair{\patNeg{p_1}}{\patWildcard}}{\tmPair{\patWildcard}{\patNeg{p_2}}}$, whereas we (in \cref{subsec:normalization:nnf}) translate it by the pattern $\patOr{\patNeg{\tmPair{\patWildcard}{\patWildcard}}}{\patOr{\tmPair{\patNeg{p_1}}{\patWildcard}}{\tmPair{\patWildcard}{\patNeg{p_2}}}}$.
He can omit the first disjunct of the or-pattern because he knows that the pair type only has one constructor.
If we extend his translation scheme to the list type, then we would translate the negation of the constructor-pattern $\text{Nil}$ to the pattern $\text{Cons}(\patWildcard,\patWildcard)$, since we know that the list type only has two constructors.
Because we encode negative information directly in patterns, we never have to rely on type information during our compilation process.
Negation patterns have also been introduced under the name of \enquote{anti-patterns} and \enquote{anti-pattern matching} by \cite{Kirchner2007,Kirchner2008, Kirchner2010}.
The formalism they present differs from ours in the following points:
They only allow constructor-patterns, variable patterns and negation patterns, which means that they cannot give a recursive translation for patterns involving negations, since that also requires or-patterns.
Their formalism is also more rooted in the theory of term rewriting systems instead of programming languages, and they do not present a compilation algorithm to simple patterns as we do.

\paragraph*{Pattern matching compilation algorithms}
A simple way to execute pattern matching expressions with nested patterns is to try one clause after another, until the first one matches.
This, however, is extremely inefficient if we have multiple clauses which are headed by the same outer constructor.
In that case, we are performing the same tests multiple times.
As an introduction to this problem, we can still recommend the book by \cite[Chapter. 5]{PeytonJones1987chapter5}.
There are many algorithms that compile nested pattern matches to efficient code, but most of them fall into one of two categories.
Nested pattern matching can be either compiled using backtracking automata \cite{Augustsson1985,LeFessant2001} or by compiling to decision trees \cite{Maranget2008}.
The main difference between these two approaches is a classical time vs. space trade-off.
Compiling to decision trees guarantees that every test is executed at most once, but can potentially increase the size of the generated code; backtracking automata do not increase the size of the code, but may perform some tests multiple times.
In \cref{sec:compilation} we present a version of the algorithm by \cite{Maranget2008} which compiles to decision trees.
Recently, Cheng and Parreaux \cite{Cheng2024ultimate} also proved the correctness of a pattern matching compilation algorithm for their very expressive pattern matching syntax; they, however, do not care about order-independence in the same way that we do.

\section{Conclusion}
\label{sec:conclusion}
Pattern matching is an extremely popular declarative programming construct.
We argued that if we want to make pattern matching even more declarative, then the order of the clauses in a program should not matter.
But if we require that two patterns must not overlap, and if we want to avoid overly verbose pattern matching expressions, then we have to make the language of patterns and pattern matching more expressive.
We introduced two complementary features, a boolean algebra of patterns and default clauses, which together solve the verbosity problem.
We have provided the operational and static semantics of these constructs, and have shown that they can be compiled to efficient code.

\section*{Data-Availability Statement}
\label{sec:data-availability}
The theorems in \cref{sec:pattern-algebra} have been formalized and checked in the proof assistant Rocq.
These proofs are available as related material.

\bibliography{bibliography/bibliography.bib, bibliography/ownpublications.bib}

\begin{thebibliography}{10}

\bibitem{Augustsson1985}
Lennart Augustsson.
\newblock Compiling pattern matching.
\newblock In Jean-Pierre Jouannaud, editor, {\em Functional Programming
  Languages and Computer Architecture}, pages 368--381, Berlin, Heidelberg,
  1985. Springer.
\newblock \href {https://doi.org/10.1007/3-540-15975-4_48}
  {\path{doi:10.1007/3-540-15975-4_48}}.

\bibitem{Bracha2022}
Gilad Bracha.
\newblock Newspeak programming language draft specification version 0.104.
\newblock Technical report, 2022.

\bibitem{Cerrito2004}
Serenella Cerrito and Delia Kesner.
\newblock Pattern matching as cut elimination.
\newblock {\em Theoretical Computer Science}, 323(1):71--127, 2004.
\newblock \href {https://doi.org/10.1016/j.tcs.2004.03.032}
  {\path{doi:10.1016/j.tcs.2004.03.032}}.

\bibitem{Cheng2024ultimate}
Luyu Cheng and Lionel Parreaux.
\newblock The ultimate conditional syntax.
\newblock {\em Proc. ACM Program. Lang.}, 8(OOPSLA2), October 2024.
\newblock \href {https://doi.org/10.1145/3689746} {\path{doi:10.1145/3689746}}.

\bibitem{Kirchner2010}
Horatiu Cirstea, Claude Kirchner, Radu Kopetz, and Pierre-Etienne Moreau.
\newblock Anti-patterns for rule-based languages.
\newblock {\em Journal of Symbolic Computation}, 45(5):523--550, 2010.
\newblock Symbolic Computation in Software Science.
\newblock \href {https://doi.org/10.1016/j.jsc.2010.01.007}
  {\path{doi:10.1016/j.jsc.2010.01.007}}.

\bibitem{Cockx2013}
Jesper Cockx.
\newblock Overlapping and order-independent patterns in type theory.
\newblock Master's thesis, KU Leuven, 2013.
\newblock URL: \url{https://lirias.kuleuven.be/retrieve/262314/}.

\bibitem{Cockx2014}
Jesper Cockx, Frank Piessens, and Dominique Devriese.
\newblock Overlapping and order-independent patterns.
\newblock In Zhong Shao, editor, {\em Programming Languages and Systems}, pages
  87--106, Berlin, Heidelberg, 2014. Springer.
\newblock \href {https://doi.org/10.1007/978-3-642-54833-8_6}
  {\path{doi:10.1007/978-3-642-54833-8_6}}.

\bibitem{Felleisen1992}
Matthias Felleisen and Robert Hieb.
\newblock The revised report on the syntactic theories of sequential control
  and state.
\newblock {\em Theoretical Computer Science}, 103(2):235--271, 1992.
\newblock \href {https://doi.org/10.1016/0304-3975(92)90014-7}
  {\path{doi:10.1016/0304-3975(92)90014-7}}.

\bibitem{Garrigue1998polymorphic}
Jacques Garrigue.
\newblock Programming with polymorphic variants.
\newblock In {\em ML workshop}. Baltimore, 1998.
\newblock URL:
  \url{https://caml.inria.fr/pub/papers/garrigue-polymorphic_variants-ml98.pdf}.

\bibitem{GellerHirschfeldBracha2010}
Felix Geller, Robert Hirschfeld, and Gilad Bracha.
\newblock {\em Pattern Matching for an Object-Oriented and Dynamically Typed
  Programming Language}.
\newblock Technische Berichte des Hasso-Plattner-Instituts für
  Softwaresystemtechnik an der Universität Potsdam. Universitätsverlag
  Potsdam, 2010.

\bibitem{Graf2020}
Sebastian Graf, Simon Peyton~Jones, and Ryan~G. Scott.
\newblock Lower your guards: A compositional pattern-match coverage checker.
\newblock {\em Proc. ACM Program. Lang.}, 4(ICFP), aug 2020.
\newblock \href {https://doi.org/10.1145/3408989} {\path{doi:10.1145/3408989}}.

\bibitem{Jay2009Book}
Barry Jay.
\newblock {\em Pattern Calculus: Computing with Functions and Structures}.
\newblock Springer, 2009.
\newblock \href {https://doi.org/10.1007/978-3-540-89185-7}
  {\path{doi:10.1007/978-3-540-89185-7}}.

\bibitem{Jay2009}
Barry Jay and Delia Kesner.
\newblock First-class patterns.
\newblock {\em Journal of Functional Programming}, 19(2):191--225, 2009.
\newblock \href {https://doi.org/10.1017/S0956796808007144}
  {\path{doi:10.1017/S0956796808007144}}.

\bibitem{Karachalias2015}
Georgios Karachalias, Tom Schrijvers, Dimitrios Vytiniotis, and Simon~Peyton
  Jones.
\newblock Gadts meet their match: Pattern-matching warnings that account for
  gadts, guards, and laziness.
\newblock In {\em Proceedings of the 20th ACM SIGPLAN International Conference
  on Functional Programming}, ICFP 2015, pages 424--436, New York, NY, USA,
  2015. Association for Computing Machinery.
\newblock \href {https://doi.org/10.1145/2784731.2784748}
  {\path{doi:10.1145/2784731.2784748}}.

\bibitem{Kirchner2007}
Claude Kirchner, Radu Kopetz, and Pierre-Etienne Moreau.
\newblock Anti-pattern matching.
\newblock In Rocco De~Nicola, editor, {\em Programming Languages and Systems},
  pages 110--124, Berlin, Heidelberg, 2007. Springer.
\newblock \href {https://doi.org/10.1007/978-3-540-71316-6_9}
  {\path{doi:10.1007/978-3-540-71316-6_9}}.

\bibitem{Kirchner2008}
Claude Kirchner, Radu Kopetz, and Pierre-Etienne Moreau.
\newblock Anti-pattern matching modulo.
\newblock In Carlos Mart{\'i}n-Vide, Friedrich Otto, and Henning Fernau,
  editors, {\em Language and Automata Theory and Applications}, pages 275--286,
  Berlin, Heidelberg, 2008. Springer.
\newblock \href {https://doi.org/10.1007/978-3-540-88282-4_26}
  {\path{doi:10.1007/978-3-540-88282-4_26}}.

\bibitem{Klopp2024datalog}
David Klopp, Sebastian Erdweg, and Andr\'{e} Pacak.
\newblock A typed multi-level datalog ir and its compiler framework.
\newblock {\em Proc. ACM Program. Lang.}, 8(OOPSLA2), October 2024.
\newblock \href {https://doi.org/10.1145/3689767} {\path{doi:10.1145/3689767}}.

\bibitem{Krishnaswami2009}
Neelakantan~R. Krishnaswami.
\newblock Focusing on pattern matching.
\newblock In {\em Proceedings of the 36th Annual ACM SIGPLAN-SIGACT Symposium
  on Principles of Programming Languages}, POPL '09, pages 366--378, New York,
  NY, USA, 2009. Association for Computing Machinery.
\newblock \href {https://doi.org/10.1145/1480881.1480927}
  {\path{doi:10.1145/1480881.1480927}}.

\bibitem{LeFessant2001}
Fabrice Le~Fessant and Luc Maranget.
\newblock Optimizing pattern matching.
\newblock In {\em Proceedings of the Sixth ACM SIGPLAN International Conference
  on Functional Programming}, ICFP '01, pages 26--37, New York, NY, USA, 2001.
  Association for Computing Machinery.
\newblock \href {https://doi.org/10.1145/507635.507641}
  {\path{doi:10.1145/507635.507641}}.

\bibitem{Lubin2021}
Justin Lubin and Sarah~E. Chasins.
\newblock How statically-typed functional programmers write code.
\newblock {\em Proc. ACM Program. Lang.}, 5(OOPSLA), oct 2021.
\newblock \href {https://doi.org/10.1145/3485532} {\path{doi:10.1145/3485532}}.

\bibitem{Maranget2007}
Luc Maranget.
\newblock Warnings for pattern matching.
\newblock {\em Journal of Functional Programming}, 17(3):387--421, 2007.
\newblock \href {https://doi.org/10.1017/S0956796807006223}
  {\path{doi:10.1017/S0956796807006223}}.

\bibitem{Maranget2008}
Luc Maranget.
\newblock Compiling pattern matching to good decision trees.
\newblock In {\em Proceedings of the 2008 ACM SIGPLAN workshop on ML}, pages
  35--46, 2008.
\newblock \href {https://doi.org/10.1145/1411304.1411311}
  {\path{doi:10.1145/1411304.1411311}}.

\bibitem{PeytonJones1987chapter5}
Simon~Loftus Peyton~Jones.
\newblock {\em The Implementation of Functional Programming Languages}.
\newblock Prentice-Hall International Series in Computer Science. Prentice
  Hall, 1987.
\newblock Chapter 5 was contributed by Philip Wadler.

\bibitem{Rioux2023}
Nick Rioux, Xuejing Huang, Bruno C. d.~S. Oliveira, and Steve Zdancewic.
\newblock A bowtie for a beast: Overloading, eta expansion, and extensible data
  types in f.
\newblock {\em Proc. ACM Program. Lang.}, 7(POPL), jan 2023.
\newblock \href {https://doi.org/10.1145/3571211} {\path{doi:10.1145/3571211}}.

\bibitem{Scott2000}
Kevin Scott and Norman Ramsey.
\newblock When do match-compilation heuristics matter.
\newblock {\em University of Virginia, Charlottesville, VA}, 2000.
\newblock URL:
  \url{https://www.cs.tufts.edu/~nr/cs257/archive/norman-ramsey/match.pdf}.

\bibitem{Winters2020google}
Titus Winters, Tom Manshreck, and Hyrum Wright.
\newblock {\em Software Engineering at Google -- Lessons Learned from
  Programming Over Time}.
\newblock O'Reilly Media, Inc., 2020.

\bibitem{Zhang2020}
Weixin Zhang and Bruno C. d.~S. Oliveira.
\newblock Pattern matching in an open world.
\newblock {\em SIGPLAN Not.}, 53(9):134–146, apr 2020.
\newblock \href {https://doi.org/10.1145/3393934.3278124}
  {\path{doi:10.1145/3393934.3278124}}.

\bibitem{Zhang2021}
Weixin Zhang, Yaozhu Sun, and Bruno C. D.~S. Oliveira.
\newblock Compositional programming.
\newblock {\em ACM Trans. Program. Lang. Syst.}, 43(3), sep 2021.
\newblock \href {https://doi.org/10.1145/3460228} {\path{doi:10.1145/3460228}}.

\end{thebibliography}

\appendix
\section{Computing Normalized Elementary Conjuncts}
\label{sec:appendix-normalization}
In this section we show how to compute the normalized disjunctive normal form introduced in \cref{def:normalization:ndnf}.
The function $\normalize{\cdot}$ takes an elementary conjunct (cp.~\cref{def:normalization:dnf}) and computes the corresponding normalized elementary conjunct.
It is defined by the following clauses:

\begin{align*}
    \normalize{x} &\coloneqq \patAnd{\{ x \}}{\patNeg{\lbrace \rbrace}} \\
    \normalize{\patWildcard} &\coloneqq \patAnd{\{ \}}{\patNeg{\lbrace \rbrace}} \\
    \normalize{\patBot} &\coloneqq \patAnd{\{ \}}{\patBot} \\
    \normalize{\patNeg{x}} &\coloneqq \patAnd{\{ \}}{\patBot} \\
    \normalize{\patNeg{\mathcal{C}^n}} &\coloneqq \patAnd{\{ \}}{\patNeg{\lbrace \mathcal{C}^n\rbrace}} \\
    \normalize{\mathcal{C}^n(K_1,\ldots,K_n)} &\coloneqq \patAnd{\lbrace \rbrace}{\mathcal{C}^n(\normalize{K_1},\ldots,\normalize{K_n})} \\
    \normalize{\patAnd{K_1}{K_n}} &\coloneqq \combine{\normalize{K_1}}{\normalize{K_2}}
\end{align*}

The function $\normalize{\cdot}$ uses the auxiliary function $\combine{\cdot\,}{\cdot}$ which merges two normalized elementary conjuncts into one.
Merging any normalized elementary conjunct with an unsatisfiable one results in a new unsatisfiable conjunct:

\begin{align*}
    \combine{\patAnd{S_1}{\patBot}}{\patAnd{S_2}{\patBot}} &\coloneqq \patAnd{S_1 \cup S_2}{\patBot} \\
    \combine{\patAnd{S_1}{\patBot}}{\patAnd{S_2}{\mathcal{C}^n(\overline{K}_1,\ldots,\overline{K}_n)}} &\coloneqq \patAnd{S_1 \cup S_2}{\patBot} \\
    \combine{\patAnd{S_1}{\patBot}}{\patAnd{S_2}{\patNeg{\{\ldots \}}}} &\coloneqq \patAnd{S_1 \cup S_2}{\patBot}
\end{align*}

We combine two negative conjuncts by computing the union of the sets of variables $S_1$ and $S_2$ and the sets of constructors $cs_1$ and $cs_2$ we don't match against:
\begin{align*}
    \combine{\patAnd{S_1}{\patNeg{\lbrace cs_1 \rbrace}},\patAnd{S_2}{\patNeg{\lbrace cs_2\rbrace}}} &\coloneqq \patAnd{S_1 \cup S_2}{\patNeg{\lbrace cs_1 \cup cs_2 \rbrace}}
\end{align*}
When we combine a positive conjunct with a negative one we have to check whether the constructor of the positive conjunct appears in the negated set of constructors:
\begin{align*}
    \combine{\patAnd{S_1}{\mathcal{C}^n(\overline{K}_1,\ldots,\overline{K}_n)}}{\patAnd{S_2}{\patNeg{\lbrace cs \rbrace }}} &\coloneqq \begin{cases} \patAnd{S_1 \cup S_2}{\patBot} &(\text{if $\mathcal{C}^n \in cs$})\\ \patAnd{S_1 \cup S_2}{\mathcal{C}^n(\overline{K}_1,\ldots,\overline{K}_n)} &(\text{otherwise})\end{cases}
\end{align*}

Lastly, when we combine to positive elementary conjuncts we have to check whether they match against the same constructor or not:

\begin{align*}
    \combine{\patAnd{S_1}{\mathcal{C}_1^n(\overline{K}_1,\ldots,\overline{K}_n)}}{\patAnd{S_1}{\mathcal{C}_2^m(\overline{K}'_1,\ldots,\overline{K}'_m)}} \coloneqq\\
    \begin{cases} \patAnd{S_1 \cup S_2}{\mathcal{C}_1^n(\combine{\overline{K}_1}{\overline{K}'_1},\ldots,\combine{\overline{K}_n}{\overline{K}'_n})} & (\text{if $\mathcal{C}_1 = \mathcal{C}_2, n = m$}) \\ \patAnd{S_1 \cup S_2}{\patBot} & (\text{otherwise})\end{cases}
\end{align*}

\section{Checking Exhaustiveness}
\label{sec:appendix-exhaustiveness}
We shortly discuss how we can check a pattern matching matrix in disjunctive normal form (see \cref{subsec:normalization:conjuncts}) for exhaustiveness. 
This algorithm is based on the discussion of the useful clause problem by \cite{Maranget2007}.
We adopt his definition of a \textit{useful} pattern to state a notion for \textit{exhaustiveness}\footnote{Maranget uses an alternative definition for exhaustiveness but shows equivalency between his own initial definition and our adopted formulation via usefulness.}.

\begin{definition}[Usefulness]
    We assume a pattern matrix $P$ and a pattern vector $\vec{p}$, both in normalized DNF as introduced in \cref{subsec:normalization:conjuncts}.
    $$
    \begin{aligned}
        P = \begin{bmatrix}
            \overline{D}_1^1 & \ldots \overline{D}_n^1 \\
            \vdots & \vdots \\
            \overline{D}_1^m & \ldots \overline{D}_n^m
            \end{bmatrix} && \vec{p} = (p_1,\ldots, p_n)
    \end{aligned}
    $$
    We say that the pattern vector $\vec{p}$ is \textit{useful} for matrix $P$, if 
    there exists a value vector $\vec{v}$ and substitution vector $\vec{\sigma}$, such that $\vec{p}$ matches $\vec{v}$ and $P$ does \textit{not} match $\vec{v}$ under the substitutions $\vec{\sigma}$.
    For this, we extend the notion of \enquote{matching} for vectors $\vec{p}$ as we did in \cref{subsec:compilation:correctness} and for a matrix $P$ as is natural.
    $$
    \begin{aligned}
        & \matches{\vec{p}}{\vec{v}}{\vec{\sigma}} := \bigwedge\limits_{i = 1}^n \matches{p_i}{v_i}{\sigma_i} & \ \ &
        \matches{P}{\vec{v}}{\vec{\sigma}} := \bigvee\limits_{j = 1}^m \matches{\vec{P^j}}{\vec{v}}{\vec{\sigma}}\\
        &\matchesNot{\vec{p}}{\vec{v}}{\vec{\sigma}} := \bigvee\limits_{i = 1}^n \matchesNot{p_i}{v_i}{\sigma_i} & \ \ &
        \matchesNot{P}{\vec{v}}{\vec{\sigma}} := \bigwedge\limits_{j = 1}^m \matchesNot{\vec{P^j}}{\vec{v}}{\vec{\sigma}} 
    \end{aligned}
    $$
    We abbreviate the useful clause problem by the formula $\mathcal{U}(P, \vec{p}) := \exists \ \vec{v}, \vec{\sigma}, \ \matches{\vec{p}}{\vec{v}}{\vec{\sigma}} \land \matchesNot{P}{\vec{v}}{\vec{\sigma}}$.
\end{definition}

\begin{definition}[Exhaustiveness]
    Pattern matrix $P$ is exhaustive, if $\mathcal{U}(P, (\_, \ldots \_))$ is false.
\end{definition}

The computation of the useful clause problem has structural similarities to the compilation algorithm of \cref{sec:compilation}, as we will again build specialized (see \cref{subsubsec:compilation:branch:specialization}) and default (see \cref{subsubsec:compilation:branch:default}) versions of pattern matrices.
For simplicity's sake, we will only consider the strategy of selecting the first column ($i = 1$) and restate the construction rules for both the default matrix $\begin{bmatrix} D\end{bmatrix}(P)$ and a specialization
$\begin{bmatrix}S_{C^k} \end{bmatrix}(P)$ by constructor $C^k$ for the pattern matrix $P$ in \cref{table:B1}.
\begin{table}[t]
    \renewcommand{\arraystretch}{2}
    \resizebox{\textwidth}{!}{
    \begin{tabular}{|l||l|l|}
     \hline \hline
        pattern $P^j_1$ of $P$'s first column & add row to $\renewcommand{\arraystretch}{1}\begin{bmatrix}S_{C^k} \end{bmatrix}(P)$ & add row to $\renewcommand{\arraystretch}{1}\begin{bmatrix}D \end{bmatrix}(P)$\\ \hline \hline 
     $\parallel\lbrace \overline{K}_1, \ldots, \overline{K}_s \rbrace$ & 
     $\renewcommand{\arraystretch}{1}\begin{bmatrix}S_{C^k} \end{bmatrix}\left(\begin{bmatrix}
        \mathbf{\overline{K}}_1 & \overline{D}_2^j \ldots \overline{D}_{n}^j \\
        & \vdots \\
        \mathbf{\overline{K}}_s & \overline{D}_2^j \ldots \overline{D}_{n}^j
    \end{bmatrix}\right)$ & 
    $\renewcommand{\arraystretch}{1}
     \begin{bmatrix}D \end{bmatrix}\left(
        \begin{bmatrix}
        \mathbf{\overline{K}}_1 & \overline{D}_2^j \ldots \overline{D}_{n}^j \\
        &\vdots \\
        \mathbf{\overline{K}}_s & \overline{D}_2^j \ldots \overline{D}_{n}^j
    \end{bmatrix}\right)$ \\ \hline
     $\patAnd{\lbrace x_1,\ldots,x_s \rbrace}{\mathcal{C}^k(\overline{K}_1,\ldots, \overline{K}_k)}$ & 
        $\renewcommand{\arraystretch}{1}\begin{bmatrix}
        \overline{K}_1 \ldots \overline{K}_k \hspace{4mm} 
        \overline{D}_{2}^j \ldots \ \overline{D}_{n}^j
        \end{bmatrix}$ 
        & \textit{none}  \\ \hline

     $\patAnd{\lbrace x_1,\ldots, x_s \rbrace}{\patNeg{\lbrace \mathcal{C}'_1,\ldots,\mathcal{C}'_t \rbrace}}$ 
     & $\renewcommand{\arraystretch}{1}\hspace{-1.1mm}\begin{aligned} 
        &\begin{bmatrix}
            \_ \ \ \ldots \ \ \_ \hspace{4.8mm} 
            \overline{D}_{2}^j \ldots \ \overline{D}_{n}^j
             \end{bmatrix} \\ 
        &\textbf{ only if } C^k \not\in \lbrace \mathcal{C}'_1,\ldots,\mathcal{C}'_t \rbrace
        \end{aligned}$
        &  $\renewcommand{\arraystretch}{1}\begin{aligned} 
            \begin{bmatrix}
                \_ \ \ \ldots \ \ \_ \hspace{4.8mm} 
                \overline{D}_{2}^j \ldots \ \overline{D}_{n}^j
                 \end{bmatrix} \\ \ 
            \end{aligned}$ \\ \hline

     $\patAnd{\{ x_1,\ldots,x_s \}}{\patBot}$ & \textit{none}
     & \textit{none} \\ \hline \hline
    \end{tabular}}
    \caption{Construction of $C^k$-constructor spec. and default matrix for the pattern matrix $P$.}
    \label{table:B1}
\end{table}

Let's give the computation of $\mathcal{U}(P, \vec{p})$. If $P$ has no rows or no columns left, we have a base case. For the former we denote $P = \emptyset$ and for the latter we give $P$ as an empty matrix $P = \begin{bmatrix}\ \end{bmatrix}$. 
$$
\begin{aligned}
    \mathcal{U}(\begin{bmatrix}\ \end{bmatrix}, ()) := false & &
    \mathcal{U}(\emptyset, ()) := true 
\end{aligned}.
$$
In any other case, we proceed recursively according to the cases of the first pattern $p_1$. 

\begin{enumerate}
    \item For $p_1 = \patAnd{\lbrace x_1,\ldots,x_u \rbrace}{\mathcal{C}^k(\overline{K}_1,\ldots, \overline{K}_k)}$ we specialize for constructor $\mathcal{C}^k$ and define
    $$
    \mathcal{U}(P, \vec{p}) := \mathcal{U}(\begin{bmatrix}S_{C^k} \end{bmatrix}(P), \begin{bmatrix}S_{C^k} \end{bmatrix}(\vec{p}))
    $$
    \item For $p_1 = \ \parallel\lbrace \overline{K}_1, \ldots, \overline{K}_u \rbrace$ we externalize the disjunction and define 
    $$
        \mathcal{U}(P, \vec{p}) := \bigvee \limits_{i = 1}^u \mathcal{U}(P, (\overline{K}_i, p_2 \ldots, p_n))
    $$
    \item The case of $p_1 = \patAnd{\lbrace x_1,\ldots,x_u \rbrace}{\patNeg{\{ C_1, \ldots, C_v\}}}$ is the most interesting, as it also contains $p_1 = \_ = \patAnd{\{\}}{\patNeg{\{\}}}$, as found at the initial computation for checking exhaustiveness. This is also the case, where our formulation must necessarily diverge from \cite{Maranget2007}'s version, which does not consider negations. \\
          This step is entirely dependent of the head constructors found in the first column $P_1$ of the pattern matrix $P$ (see the definition of head constructors in Section \ref{head}). In contrast to the pattern matching compilation, we need a more granular inspection and collect negated head constructors $\mathcal{H_{-}}(P_1)$ separately from \enquote{positive} head constructors $\mathcal{H_{+}}(P_1)$. 
            \begin{enumerate}
                \item For $\mathcal{H_{-}}(P_1) = \emptyset$ we can follow \cite{Maranget2007}'s standard computation and differentiate cases according to whether $\mathcal{H_{+}}(P_1)$ gives a complete type signature. 
                \begin{enumerate}
                    \item If $\mathcal{H_{+}}(P_1) = \{C'_1,\ldots, C'_z\} $ is complete, we form a disjunction over all specializations. 
                          $$
                          \mathcal{U}(P, \vec{p}) := \bigvee \limits_{i = 1}^{z} \mathcal{U}\left(\begin{bmatrix}S_{C'_i}
                          \end{bmatrix}(P), \begin{bmatrix}S_{C'_i}
                          \end{bmatrix}(\vec{p})\right)
                          $$
                    \item If $\mathcal{H_{+}}(P_1) = \{C'_1,\ldots, C'_z\} $ is incomplete, we consider the constructors \textit{not} contained in the collection.
                    $$
                        \mathcal{U}(P, \vec{p}) :=  \mathcal{U}\left(\begin{bmatrix}D\end{bmatrix}(P), \begin{bmatrix}D\end{bmatrix}(\vec{p})\right)
                    $$
                \end{enumerate}
                \item For $\mathcal{H_{-}}(P_1) = \{C'_1, \ldots, C'_z\} \neq \emptyset$ we must consider
                exactly these negative head constructors for an exhaustiveness check. 
                    $$
                        \mathcal{U}(P, \vec{p}) := \bigvee \limits_{i = 1}^{z} \mathcal{U}\left(\begin{bmatrix}S_{C'_i}
                        \end{bmatrix}(P), \begin{bmatrix}S_{C'_i}
                        \end{bmatrix}(\vec{p})\right)
                    $$
            \end{enumerate}
\end{enumerate}

\subsection{Example}

We conclude this section by checking the example from \cref{subsec:compilation:example} for exhaustiveness:
$$
\begin{array}{llll}
\mathbf{case} \ x \ \mathbf{of} \{ 
	& \patAnd{y}{(\patOr{\tmSa}{\tmSu})} & \To \text{``Today is weekend!''},&\\
	& \patAnd{y}{\patNeg{(\patOr{\tmFr}{\patOr{\tmSa}{\tmSu}})}}  & \To \text{``Today is ''}\mdoubleplus\ \text{show}(y),&\\
	& \mathbf{default} & \To \text{``Tomorrow is weekend...''} &\ \}
\end{array}
$$
This time, we want to check whether the non-default clauses are already exhaustive. So, we check whether $p_1 := \_$ is a useful pattern for the pattern matrix $P$, which we have normalized into the following form:
$$
P := 
\begin{bmatrix*}[l]
	\patOr{}{}\{\{y\} \ \& \ \tmSa, \{y\} \ \& \ \tmSu \} \\
	\patOr{}{}\{\{y\} \ \& \ \patNeg{}\{\tmFr,\tmSa,\tmSu\}\} \\
\end{bmatrix*}
$$
To solve $ \mathcal{U}(P, (\_))$, we see that the pattern $p_1 = \_ = \patAnd{\{\}}{\patNeg{\{\}}}$ satisfies case (3) and we collect the negated head constructors $\mathcal{H_{-}}(P_1) = \{Fr, Sa, Su \}$. According to (3.b), we form a disjunction over specializations for each negated head constructor. Here, it suffices to look at the constructor $\tmFr$:
$$
\mathcal{U}(P, (\_)) = \hspace{-5mm}
\bigvee \limits_{C \in \{\tmFr, \tmSa, \tmSu \}} \hspace{-5mm}\mathcal{U}\left(\begin{bmatrix}S_{C}
\end{bmatrix}(P), \begin{bmatrix}S_{C}
\end{bmatrix}(\_)\right) = 
\mathcal{U}\left(\begin{bmatrix}S_{\tmFr}
\end{bmatrix}(P), \begin{bmatrix}S_{\tmFr}
\end{bmatrix}(\_)\right) = \mathcal{U}(\emptyset, ()) = \textit{true}
$$
This determines $P$ as non-exhaustive.

\section{Deciding Overlap}
\label{sec:appendix-overlap}
In this section, we give a decision procedure on whether two patterns in nDNF (see \cref{subsec:normalization:conjuncts}) $overlap$.  

\begin{enumerate}
    \item $overlap(\overline{D}_1, \overline{D}_2)$
    \begin{itemize}
        \item $overlap(\parallel \{ \overline{K}_1^1,\dots,\overline{K}_s^1 \}, \parallel \{ \overline{K}_1^2,\dots,\overline{K}_t^2 \}) := \bigvee\limits_{i=1}^s \bigvee\limits_{i'=1}^t overlap(\overline{K}^1_i, \overline{K}^2_{i'})$
    \end{itemize}
    \item $overlap(\overline{K}_1, \overline{K}_2)$
    \begin{itemize}
        \item $overlap(\mathbf{\patAnd{X}{C^s(\overline{K}^1_1,\dots,\overline{K}^1_s)}}, \patAnd{Y}{C^t(\overline{K}^2_1,\dots,\overline{K}^2_t)}) := $\\
        if $C^s \neq C^t$ then \textit{false} else $\bigwedge\limits_{i=1}^{s} \ overlap(\overline{K}^1_i,\overline{K}^2_i)$
        \item $overlap(\mathbf{\patAnd{X}{C^s(\overline{K}^1_1,\dots,\overline{K}^1_s)}}, \patAnd{Y}{\patNeg{\{C_1,\dots,C_t\}}}) := $\\
        if $C^s \in \{C_1,\dots,C_t\}$ then \textit{false} else $\bigwedge\limits_{i=1}^{s} \ overlap(\overline{K}^1_i, \patAnd{\{\}}{\patNeg{\{\}}})$
        \item $overlap(\mathbf{\patAnd{X}{C^s(\overline{K}^1_1,\dots,\overline{K}^1_s)}}, \patAnd{Y}{\patBot}) := \textit{false}$ \vspace{2mm}
        \item $overlap(\mathbf{\patAnd{X}{\patBot}}, \overline{K}_2) := \textit{false}$ \vspace{2mm}
        \item $overlap(\mathbf{\patAnd{X}{\patNeg{\{C_1,\dots,C_s\}}}}, \patAnd{Y}{\patNeg{\{C_1',\dots,C_t'\}}}) := \textit{true}$ \vspace{2mm}\\
        (Note that this is a conservative definition of overlap that does not require further type information. The less conservative alternative for static, known types would be \vspace{1mm}\\
        if $\{C_1,\dots,C_s\} \cup \{C_1',\dots,C_t'\}$ gives the whole type signature then \textit{false} else \textit{true}) \vspace{1mm}
        \item $overlap(\mathbf{\patAnd{X}{\patNeg{\{C_1,\dots,C_s\}}}},\patAnd{Y}{C^t(\overline{K}^2_1,\dots,\overline{K}^2_t)}) := $\vspace{-0.5mm}\\ $overlap(\patAnd{Y}{C^t(\overline{K}^2_1,\dots,\overline{K}^2_t)},\patAnd{X}{\patNeg{\{C_1,\dots,C_s\}}})$ \vspace{2mm}
        \item $overlap(\mathbf{\patAnd{X}{\patNeg{\{C_1,\dots,C_s\}}}},\patAnd{Y}{\patBot})$ := $overlap(\patAnd{Y}{\patBot},\patAnd{X}{\patNeg{\{C_1,\dots,C_s\}}})$
    \end{itemize}
\end{enumerate}

\section{Proofs}
\label{sec:appendix-proofs}
In this section we provide proofs which were omitted in the main part of the paper.

\subsection{Correctness of Negation Normal Forms}

This subsection contains the proofs of \cref{lem:nnf:correctness}.
All the proofs in this subsection rely on the fact that in order to compute the negation normal form of a pattern we only use the algebraic laws of \cref{thm:algebraic-equivalences:one} and \cref{thm:algebraic-equivalences:four}.
But first, we need some simple lemmas:
\begin{lemma}[Negation Normal Form Preserves Free Variables]
    For all patterns $p$, $\freevareven{p} = \freevareven{\nnf{p}}$ and $\freevarodd{p} = \freevarodd{\nnf{p}}$.
    \label{lem:appendix-proofs:free-vars}
\end{lemma}
\begin{proof}
  By simple induction on $p$.
\end{proof}

\begin{lemma}[De Morgan Preserves Linearity]
    For all patterns $p_1$ and $p_2$, we have
    \begin{align*}
        \linearpos{\patNeg{(\patOr{p_1}{p_2})}}\quad &\Rightarrow\quad \linearpos{\patAnd{\patNeg{p_1}}{\patNeg{p_2}}} \\
        \linearneg{\patNeg{(\patOr{p_1}{p_2})}}\quad &\Rightarrow\quad \linearneg{\patAnd{\patNeg{p_1}}{\patNeg{p_2}}} \\
        \linearpos{\patNeg{(\patAnd{p_1}{p_2})}}\quad &\Rightarrow\quad \linearpos{\patOr{\patNeg{p_1}}{\patNeg{p_2}}} \\
        \linearneg{\patNeg{(\patAnd{p_1}{p_2})}}\quad &\Rightarrow\quad \linearneg{\patOr{\patNeg{p_1}}{\patNeg{p_2}}} \\
        \linearpos{\patNeg{\patNeg{p}}}\quad &\Rightarrow\quad \linearpos{p} \\
        \linearneg{\patNeg{\patNeg{p}}}\quad &\Rightarrow\quad \linearneg{p}
    \end{align*}
    \label{lem:appendix-proofs:de-morgan}
\end{lemma}
\begin{proof}
    We only prove the first two lines, the proofs for the remaining cases are analogous or trivial in the case of double negations.
    \begin{enumerate}
        \item If $\linearpos{\patNeg{(\patOr{p_1}{p_2})}}$ then $\linearpos{\patAnd{\patNeg{p_1}}{\patNeg{p_2}}}$:
        \begin{itemize}
            \item $\linearneg{\patOr{p_1}{p_2}}$ via \textsc{L-Neg}$^+$
            \item $\linearneg{p_1}$ and $\linearneg{p_2}$ and $\freevarodd{p_1} \cap \freevarodd{p_2} = \emptyset$ via \textsc{L-Or}$^-$
            \item $\linearpos{\patNeg{p_1}}$ and $\linearpos{\patNeg{p_2}}$ and $\freevareven{\patNeg{p_1}} \cap \freevareven{\patNeg{p_2}} = \emptyset$
            \item $\linearpos{\patAnd{\patNeg{p_1}}{\patNeg{p_2}}}$ via \textsc{L-And}$^+$
        \end{itemize}
        \item If $\linearneg{\patNeg{(\patOr{p_1}{p_2})}}$ then $\linearneg{\patAnd{\patNeg{p_1}}{\patNeg{p_2}}}$:
        \begin{itemize}
            \item $\linearpos{\patOr{p_1}{p_2}}$ via \textsc{L-Neg}$^-$
            \item $\linearpos{p_1}$ and $\linearpos{p_2}$ and $\freevareven{p_1} = \freevareven{p_2}$ via \textsc{L-Or}$^+$
            \item $\linearneg{\patNeg{p_1}}$ and $\linearneg{\patNeg{p_2}}$ and $\freevarodd{\patNeg{p_1}} = \freevarodd{\patNeg{p_2}}$
            \item $\linearneg{\patAnd{\patNeg{p_1}}{\patNeg{p_2}}}$ via \textsc{L-And}$^-$
        \end{itemize}
    \end{enumerate}
\end{proof}

\begin{lemma}[Negated Constructor Preserves Linearity]
    For all patterns $p_1,\ldots,p_n$, the following implication holds:
    \begin{equation*}
        \linearpos{\patNeg{\mathcal{C}^n(p_1,\ldots,p_n)}}
        \quad \Rightarrow \quad
        \linearpos{\patOr{\patNeg{\mathcal{C}^n}}{\patOr{\mathcal{C}^n(\patNeg{p_1},\ldots,\patWildcard)}{\patOr{\ldots}{\mathcal{C}^n(\patWildcard,\ldots,\patNeg{p_n})}}}}
    \end{equation*}
    \label{lem:appendix-proofs:negated-constructor}
\end{lemma}
\begin{proof}
    We reason stepwise:
    \begin{itemize}
        \item $\linearneg{\mathcal{C}^n(p_1,\ldots,p_n)}$ via \textsc{L-Neg}$^+$
        \item $\linearneg{p_i}$ and $\freevarodd{p_i} = \emptyset$ via \textsc{L-Ctor}$^-$
        \item $\linearpos{\patNeg{\mathcal{C}^n}}$ and $\linearpos{\patNeg{p_i}}$ are immediate from the rules.
        \item Each of $\linearpos{\mathcal{C}^n(\patWildcard,\ldots, \patNeg{p_i},\ldots,\patWildcard)}$ follows because there is only one argument which might contain variables, so the precondition of \textsc{L-Ctor}$^+$ is trivially fulfilled.
        \item $\linearpos{\patOr{\patNeg{\mathcal{C}^n}}{\patOr{\mathcal{C}^n(\patNeg{p_1},\ldots,\patWildcard)}{\patOr{\ldots}{\mathcal{C}^n(\patWildcard,\ldots,\patNeg{p_n})}}}}$ follows via \textsc{L-Or}$^+$, because the set of free variables under an even number of negations is the empty set in each of the disjuncts.
    \end{itemize}
\end{proof}

\begin{theorem}[Negation Normal Form Preserves Linearity]
    For all patterns $p$, if $\linearpos{p}$ then $\linearpos{\nnf{p}}$.
    \label{lem:appendix-proofs:linearity-preservation}
\end{theorem}
\begin{proof}
    We can first apply \cref{lem:appendix-proofs:de-morgan} whenever the algorithm pushes negations inside until they disappear or remain before a constructor headed by a single negation.
    In that case we can use \cref{lem:appendix-proofs:negated-constructor}.
\end{proof}

\begin{lemma}[De Morgan Preserves Typing]
    For all types $\tau$, typing contexts $\Gamma$ and $\Delta$, and patterns $p, p_1$ and $p_2$:
    \begin{align*}
        \patTyping{\Gamma}{\Delta}{\patNeg{(\patOr{p_1}{p_2})}}{\tau}\quad &\Leftrightarrow\quad \patTyping{\Gamma}{\Delta}{\patAnd{\patNeg{p_1}}{\patNeg{p_2}}}{\tau} \\
        \patTyping{\Gamma}{\Delta}{\patNeg{(\patAnd{p_1}{p_2})}}{\tau}\quad &\Leftrightarrow\quad \patTyping{\Gamma}{\Delta}{\patOr{\patNeg{p_1}}{\patNeg{p_2}}}{\tau} \\
        \patTyping{\Gamma}{\Delta}{\patNeg{\patNeg{p}}}{\tau}\quad &\Leftrightarrow\quad \patTyping{\Gamma}{\Delta}{p}{\tau}
    \end{align*}
    \label{lem:appendix-proofs:demorgan-typing}
\end{lemma}
\begin{proof}
    By simple inspection of the typing rules for algebraic patterns.
\end{proof}

\begin{lemma}[Negated Constructors Preserve Typing]
    For all patterns $p_1,\ldots,p_n$, types $\tau$ and typing contexts $\Gamma$ and $\Delta$:
    \begin{equation*}
        \patTyping{\Gamma}{\Delta}{\patNeg{\mathcal{C}(p_1,\ldots,p_n)}}{\tau}
        \quad \Leftrightarrow\quad
        \patTyping{\Gamma}{\Delta}{\patOr{\patNeg{\mathcal{C}^n}}{\patOr{\mathcal{C}^n(\patNeg{p_1},\ldots,\patWildcard)}{\patOr{\ldots}{\mathcal{C}^n(\patWildcard,\ldots,\patNeg{p_n})}}}}{\tau}
    \end{equation*}
    \label{lem:appendix-proofs:negated-constructor-typing}
\end{lemma}
\begin{proof}
    This lemma depends, of course, on the concrete typing rules for constructors that are available in the system.
    One can easily verify that it holds for the constructors specified in \cref{subsec:terms:typing}.
\end{proof}

\begin{theorem}[Negation Normal Form Preserves Typing]
    If $\patTyping{\Gamma}{\Delta}{p}{\tau}$, then $\patTyping{\Gamma}{\Delta}{\nnf{p}}{\tau}$.
    \label{lem:appendix-proofs:typing-preservation}
\end{theorem}
\begin{proof}
    This follows from \cref{lem:appendix-proofs:demorgan-typing,lem:appendix-proofs:negated-constructor-typing}.
\end{proof}

\begin{theorem}[Negation Normal Form Preserves Semantics]
    For all patterns $p$, if $\linearposneg{p}$ then $\llbracket p \rrbracket \equiv \llbracket \nnf{p} \rrbracket$.
\end{theorem}
\begin{proof}
    This follows from the semantic equivalences proven in \cref{thm:algebraic-equivalences:one,thm:algebraic-equivalences:four}
\end{proof}

\subsection{Correctness of Compilation}
\label{subsec:proofs:compilation-correctness}

This subsection contains the correctness proof for compilation as given by \cref{thm:compilation:correctness}. Because the compilation algorithm operates on multi-column pattern matches, we must extend the relevant definitions accordingly.

First, we extend the notion of \enquote{matching} for a any pattern vector $\vec{p}$ as is natural.
$$
\begin{aligned}
    & \matches{\vec{p}}{\vec{v}}{\vec{\sigma}} := \bigwedge\limits_{i = 1}^n \matches{p_i}{v_i}{\sigma_i} & \ \
    &\matchesNot{\vec{p}}{\vec{v}}{\vec{\sigma}} := \bigvee\limits_{i = 1}^n \matchesNot{p_i}{v_i}{\sigma_i}
\end{aligned}
$$

Second, we extend single-step evaluation for a multi-column pattern match by stating that if any non-default row matches the values, we evaluate to its substituted right-hand side and if no row matches all values, we evaluate to the default expression.\\
\begin{minipage}{0.5\textwidth}
	\vspace{3mm}
	$$\begin{array}{l}
		P \To E = \begin{bmatrix}
			P_1^1 	& \ldots 	& P_n^1 &\To e^1 \\
			& \vdots	& 				\\
			P_1^m	& \ldots	& P_n^m &\To e^m \\
			&\defaultClause & &\To e_d
		\end{bmatrix} \vspace{2mm}\\
		MultiCase(\vec{e}) =
		\mathbf{case}\ \vec{e} \ \mathbf{of} \ \left( P \To E \right)
	\end{array}$$
	\vspace{2mm}
\end{minipage}
\begin{minipage}{0.5\textwidth}
	\begin{prooftree}
		\AxiomC{$\matches{P^j}{\vec{v}}{\vec{\sigma}}$}
		\RightLabel{\textsc{E-MC}}
		\UnaryInfC{$MultiCase(\vec{v}) \singlestep^r e^j\vec{\sigma}$}
	\end{prooftree}
	\begin{prooftree}
		\AxiomC{$\matchesNot{P^1}{\vec{v}}{\vec{\sigma_1}}\quad \cdots\quad \matchesNot{P^m}{\vec{v}}{\vec{\sigma_m}}$}
		\RightLabel{\textsc{E-MD}}
		\UnaryInfC{$MultiCase(\vec{v}) \singlestep^r e_d$}
	\end{prooftree}
\end{minipage}
These evaluation steps need not be deterministic. However, if we want to reason about the correctness of compilation as semantic equivalence $\semantic{\ref{eq:case:input}} \equiv \semantic{\textit{compile}(\ref{eq:case:input})}$, which boils down to $\semantic{\ref{eq:case:input}} \to^* v \textit{ iff } \textit{compile}(\ref{eq:case:input}) \to^* v$, this is a somewhat nonsensical (and false) statement for non-determinism, because the compilation algorithm is deterministic in nature.

We will pose the theorem of correctness for deterministic order-independent semantics as guaranteed by wellformed expressions, a notion that we extend to multi-column pattern matches.
As before, all patterns need to be deterministic and positively linear. Additionally, we want non-overlapping rows, meaning that they must not overlap for at least one pattern. Finally, we also demand that a row needs disjoint free variables under even negation. The intuition here is that a $n$-column row is the \enquote{curried} equivalent to a $n$-ary anonymous constructor, which also needs to be positively linear (and deterministic).

\begin{prooftree}
\AxiomC{\renewcommand{\arraystretch}{1.5}$\begin{array}{l}
			\wf{\{e,e_d,e^1,\ldots,e^m\}} \\
			\forall j \neq j', \ \exists \ i, \text{ st. } \disjoint{P_i^j}{P_i^{j'}} \\
			\forall j, i \neq i', \ \text{FV}^e(P_i^j) \cap  \text{FV}^e(P_{i'}^j) = \emptyset
		\end{array}$}
\AxiomC{\hspace{-8mm}$
		\begin{array}{c}
			\deterministic{P_1^1} \cdots \deterministic{P_n^1}\\
			\vdots \\
			\deterministic{P_1^m} \cdots \deterministic{P_n^m}
		 \end{array} \hspace{0.25cm}
		 \hspace{-4mm}\begin{array}{c}
			\linearpos{P_1^1} \cdots \linearpos{P_n^1}\\
			\vdots \\
			\linearpos{P_1^m} \cdots \linearpos{P_n^m}
		 \end{array}
		$}
\hspace{-4mm}\RightLabel{\textsc{Wf-MC}}
\BinaryInfC{$\wf{MultiCase(e)}$}
\end{prooftree}

Below, we restate compilation correctness for \textit{wellformed} multi-column pattern matches in \cref{thm:correctness-proof}. This necessitates that we can reason about wellformedness under compilation. We state by \cref{lem:compilation:subproblem-wf}, that specialization and default cases are also wellformed and ultimately conclude in \cref{lem:compilation:preservation-wf} that compilation preserves wellformedness.
On a side note, we see by \cref{lem:wf:determinism:multicol} that the determinism provided by wellformed expressions does indeed extend to wellformed multi-column pattern matches.

For correctness, we will need reasoning about reduction, especially for specialization and default. For the nail that is the correctness theorem, \cref{lem:compilation:subproblem-reduction} will be our hammer, stating that the subproblems for specialization and default reduce as expected. 

\begin{lemma}[Specialization/Default preserve/reflect Reduction]
	\label{lem:compilation:subproblem-reduction}
	Consider a wellformed multi-column pattern match of the normalized shape \ref{eq:case:input} and denote the $i$-th scrutinee of \ref{eq:case:input} as $v_i = C^k(w_1,\dots,w_k)$. It holds that specialization and default for the $i$-th column preserve and reflect the single-step reduction of order-independent semantics.
	\begin{enumerate}
		\item $\ref{eq:case:input} \to e$ only if $\mathcal{S}(i,\mathcal{C}^{k}(x_1,\ldots, x_{k}),\text{\ref{eq:case:input}})[x_1 \mapsto w_1, \dots, x_k \mapsto w_k] \to e$
		\item If $C^k \not \in \mathcal{H}(\overline{D}_i)$, then $\ref{eq:case:input} \to e$ only if $\mathcal{D}(i,\mathcal{H}(\overline{D}_i), \text{\ref{eq:case:input}}) \to e$
	\end{enumerate}
\end{lemma}

This lemma poses an extension of one that \cite{Maranget2008} stated for multi-column pattern matches with the usual patterns, including or-patterns. As such, we will follow the author's proof structure. Each case for 1 and 2 inevitably boils down to an induction on the pattern $\overline{D}_i^j$ for some row $j$ in \ref{eq:case:input}, either directly or by knowing it has a contributing shape from the inclusion rules of either \cref{subsubsec:compilation:branch:specialization} or \cref{subsubsec:compilation:branch:default}. 

We will dedicate great care to showing a detailed proof for 1 in the highly non-trivial case of a negation-pattern, because this case embodies a true extension to Maranget's algorithm. We will omit the other cases and refer to his own statement of the lemma, because they are similarly mechanically involved -- but equally procedural -- as the one we will show. For the same reason, we will also spare the proof of 2 as it is easier and entirely analogous to the proof of 1, so much so, that we advise the doubtful reader to alter the proof themself by crossing out and scribbling in some of the necessary changes.

\begin{proof} \
\begin{enumerate}
        \item Assume $\ref{eq:case:input} \to e$. \\
                There are two reduction rules, which we differentiate: 
                \begin{enumerate}
                \item \vspace{2mm} $\ref{eq:case:input} \to e^j\sigma_1\dots\sigma_n$ for $\matches{\overline{D}^j}{(v_1,\dots,v_n)}{(\sigma_1,\dots,\sigma_n)}$, i.e. the $j$-th non-default row matches all values, and we reduce to its right-hand side with all the point wise substitutions applied successively. \vspace{2mm}\\
                To show: $\mathcal{S}(i,\mathcal{C}^{k}(x_1,\ldots, x_{k}),\text{\ref{eq:case:input}})[x_1 \mapsto w_1, \dots, x_k \mapsto w_k] \to e^j\sigma_1\dots\sigma_n $ \vspace{2mm}\\
                Let us first spell out, what specialization gives: \\
                    $\mathcal{S}(i,\mathcal{C}^{k}(x_1,\ldots, x_{k}),\text{\ref{eq:case:input}}) \ = \ \mathbf{case}\hspace{2mm} 
                    x_1, \dots, x_k, v_1,\dots,v_{i-1},v_{i+1},\dots,v_n \hspace{2mm} \mathbf{of} \hspace{1mm} \begin{bmatrix} S \end{bmatrix}$, \\
                where the rows of matrix $\begin{bmatrix} S \end{bmatrix}$ depend on the pattern $\overline{D}_i^j$, on which we induct. We focus on the case of a negation-pattern, which poses the only non-trivial extension of \cite{Maranget2008}'s compilation algorithm. \vspace{2mm}\\
                Assume $\mathbf{\overline{D}_i^j = \patAnd{\{ y_1, \ldots y_s\}}{\patNeg{\{ \mathcal{C}_1, \ldots, \mathcal{C}_t \} }}}$. \vspace{2mm}\\ 
                Because we assumed $\matches{\overline{D}^j}{(v_1,\dots,v_n)}{(\sigma_1,\dots,\sigma_n)}$, notably $\matches{\overline{D}_i^j}{v_i}{\sigma_i}$, it must already be 
                $\mathcal{C}^k \not\in \{ \mathcal{C}_1, \ldots, \mathcal{C}_t \}$. This satisfies the inclusion rule \cref{subsubsec:compilation:branch:specialization} \ref{spec-b}, meaning that the matrix  $\begin{bmatrix} S \end{bmatrix}$ contains the row \\
                $\begin{bmatrix}
                    \patWildcard \dots \ \patWildcard \
                    \overline{D}_{1}^j \dots \overline{D}_{i-1}^j \overline{D}_{i+1}^j \dots \ \overline{D}_{n}^j 
                    \To e^j [\{ y_1, \ldots y_s\} \mapsto v_i]
                \end{bmatrix}$. \vspace{2mm} \\
                For this row in $\begin{bmatrix} S \end{bmatrix}$, we can derive the matching judgment \\
                $\matches{(\patWildcard, \dots,\patWildcard,
                \overline{D}_{1}^j, \dots, \overline{D}_{i-1}^j, \overline{D}_{i+1}^j, \dots,\overline{D}_{n}^j)}{(w_1,\dots,w_k,v_1,\dots,v_{i-1},v_{i+1},\dots,v_n)}{\\
                ( [],\dots, [],\sigma_1,\dots,\sigma_{i-1},\sigma_{i+1},\dots,\sigma_n )}$. \vspace{2mm}\\
                Because the entire row does not depend on the fresh variables $x_1,\dots,x_k$ in our specialization -- particularly its patterns --
                we can use this matching judgment to reduce our goal by the rule \textsc{E-MC}\\
                $\mathcal{S}(i,\mathcal{C}^{k}(x_1,\ldots, x_{k}),\text{\ref{eq:case:input}})[x_1 \mapsto w_1, \dots, x_k \mapsto w_k] \to \\
                (e^j [\{ y_1, \ldots y_s\} \mapsto v_i]) \sigma_1\dots\sigma_{i-1}\sigma_{i+1}\dots\sigma_n$. \vspace{2mm}\\
                To show: $(e^j [\{ y_1, \ldots y_s\} \mapsto v_i]) \sigma_1\dots\sigma_{i-1}\sigma_{i+1}\dots\sigma_n = e^j\sigma_1\dots\sigma_n$ \vspace{2mm}\\
                Because the following matching judgments 
                \begin{itemize}
                    \item $\matches{\overline{D}^j_i}{v_i}{\sigma_i}$ gives for $\overline{D}_i^j = \patAnd{\{ y_1, \ldots y_s\}}{\patNeg{\{ \mathcal{C}_1, \ldots, \mathcal{C}_t \} }}$ the judgment \\
                    $\matches{\patAnd{\{ y_1, \ldots y_s\}}{\patNeg{\{ \mathcal{C}_1, \ldots, \mathcal{C}_t \} }}}{v_i}{\sigma_i}$
                    \item Natively, we can also derive for $v_i = C^k(w_1,\dots,w_k)$ the judgment \\ 
                    $\matches{\patAnd{\{ y_1, \ldots y_s\}}{\patNeg{\{ \mathcal{C}_1, \ldots, \mathcal{C}_t \} }}}{v_i}{[\{ y_1, \ldots y_s\} \mapsto v_i]}$
                \end{itemize}
                yield two separate substitutions, we use our assumption of positively linear and deterministic patterns to conclude by \cref{thm:wf-matching-det}, the two substitutions must coincide: $\sigma_i \equiv [\{ y_1, \ldots y_s\} \mapsto v_i]$. \vspace{2mm}\\
                We are allowed to permute the application of these substitutions $\sigma_1\dots\sigma_n$. This follows from our assumption that the positively linear patterns of each \ref{eq:case:input} row do not share free variables under even negation, because then we know by \cref{thm:pattern-algebra:linear-patterns-covering} that the domains of these substitutions are pairwise disjoint. We conclude: \\$ (e^j[\{ y_1, \ldots y_s\} \mapsto v_i])\sigma_1\dots\sigma_{i-1}\sigma_{i+1}\dots\sigma_n = (e^j\sigma_i)\sigma_1\dots\sigma_{i-1}\sigma_{i+1}\dots\sigma_n  = e^j\sigma_1\dots\sigma_n$ \vspace{2mm}
                \item 
                $\ref{eq:case:input} \to e_d$ for $\matchesNot{\overline{D}^j}{\vec{v}}{\vec{\sigma}^j}, \ \forall j \in \{1,\dots,m\}$, i.e. no non-default row matches, and we reduce to the default expression. \vspace{2mm}\\
                $\begin{array}{ll}
                        & \mathcal{S}(i,\mathcal{C}^{k}(x_1,\ldots, x_{k}),\text{\ref{eq:case:input}})[x_1 \mapsto w_1, \dots, x_k \mapsto w_k] \\
                    =   & \mathbf{case}\hspace{2mm} w_1, \dots, w_k, v_1,\dots,v_{i-1},v_{i+1},\dots,v_n \hspace{2mm} \mathbf{of} \hspace{1mm} \left(\begin{bmatrix} S \end{bmatrix}[x_1 \mapsto w_1, \dots, x_k \mapsto w_k] \right)
                \end{array}$ \vspace{2mm}\\
                We need to show that no row of $\left(\begin{bmatrix} S \end{bmatrix}[x_1 \mapsto w_1, \dots, x_k \mapsto w_k] \right)$ but the carried-over default clause matches the new arguments $w_1, \dots, w_k, v_1,\dots,v_{i-1},v_{i+1},\dots,v_n$. For this, we pick (if existing) any non-default row in $\left(\begin{bmatrix} S \end{bmatrix}[x_1 \mapsto w_1, \dots, x_k \mapsto w_k] \right)$ and induct over the inclusion rules of \cref{subsubsec:compilation:branch:specialization} \ref{spec-a}-\ref{spec-c} that could have led to its inclusion. Once again, as we are adapting a lemma proven by \cite{Maranget2008}, we are focusing on the new case involving negation. \vspace{2mm}\\
                Assume this row is the contribution of some row $j$ in \ref{eq:case:input} by the rule \ref{spec-b}, namely \\
                $\mathbf{\overline{D}_i^j = \patAnd{\{ y_1, \ldots y_s\}}{\patNeg{\{ \mathcal{C}_1, \ldots, \mathcal{C}_t \} }}}$ with $\mathcal{C}^k \not\in \{ \mathcal{C}_1, \ldots, \mathcal{C}_t \}$. \vspace{2mm}\\
                Then, our row looks like 
                $\begin{bmatrix}
                    \patWildcard \dots \ \patWildcard \
                    \overline{D}_{1}^j \dots \overline{D}_{i-1}^j \overline{D}_{i+1}^j \dots \ \overline{D}_{n}^j 
                    \To e^j [\{ y_1, \ldots y_s\} \mapsto v_i]
                \end{bmatrix}$. \vspace{2mm}\\ 
                To show: There exists some substitution vector $\vec{\tau}$, such that \\
                $\matchesNot{(\patWildcard, \dots,\patWildcard,
                \overline{D}_{1}^j, \dots, \overline{D}_{i-1}^j, \overline{D}_{i+1}^j, \dots,\overline{D}_{n}^j)}{
                (w_1,\dots,w_k,v_1,\dots,v_{i-1},v_{i+1},\dots,v_n)}{\vec{\tau}}$ \vspace{2mm}\\
                We proceed with a proof of negation showing that there cannot be any $\vec{\tau}$, s.t. \\
                $\matches{(\patWildcard, \dots,\patWildcard,
                \overline{D}_{1}^j, \dots, \overline{D}_{i-1}^j, \overline{D}_{i+1}^j, \dots,\overline{D}_{n}^j)}{
                (w_1,\dots,w_k,v_1,\dots,v_{i-1},v_{i+1},\dots,v_n)}{\vec{\tau}}$ \\ is derivable and then conclude our goal using completeness (\cref{thm:matching-complete}). \vspace{2mm}\\
                Now, assume we could have this row match \\
                $\matches{(\patWildcard, \dots,\patWildcard,
                \overline{D}_{1}^j, \dots, \overline{D}_{i-1}^j, \overline{D}_{i+1}^j, \dots,\overline{D}_{n}^j)}{
                (w_1,\dots,w_k,v_1,\dots,v_{i-1},v_{i+1},\dots,v_n)}{\vec{\tau}}$, and we denote $\vec{\tau} = (\tau_1,\dots,\tau_k, \pi_1, \dots, \pi_{i-1},\pi_{i+1},\dots,\pi_n)$. \vspace{2mm}\\
                This implies that row $j$ in \ref{eq:case:input} must have been a matching row, because 
                we can derive the judgment $\matches{\overline{D}^j}{\vec{v}}{(\pi_1,\dots,\pi_{i-1},[\{ y_1, \ldots y_s\} \mapsto v_i], \pi_{i+1},\dots,\pi_n)}$, because
                \begin{itemize}
                    \item $\forall \ x \in (\{1,\dots,n\}\setminus\{i\}), \ \matches{\overline{D}_x^j}{v_x}{\pi_x}$
                    \item For $\overline{D}_{i}^j = \patAnd{\{ y_1, \ldots y_s\}}{\patNeg{\{ \mathcal{C}_1, \ldots, \mathcal{C}_t \} }}$ and $v_i = C^k(w_1,\dots,w_k)$ we can derive \\
                    $\matches{\overline{D}_i^j}{v_i}{[\{ y_1, \ldots y_s\} \mapsto v_i]}$, because  $\mathcal{C}^k \not\in \{ \mathcal{C}_1, \ldots, \mathcal{C}_t \}$.
                \end{itemize} \vspace{2mm}
                But such a matching judgment contradicts our assumption of $\matchesNot{\overline{D}^j}{\vec{v}}{\vec{\sigma}^j}$ \\
                by soundness of pattern matching (\cref{thm:matching-sound}).
            \end{enumerate}
        \item 
        \vspace{2mm} Assume for the specialization \\
        $\begin{array}{ll}
            & \mathcal{S}(i,\mathcal{C}^{k}(x_1,\ldots, x_{k}),\text{\ref{eq:case:input}})[x_1 \mapsto w_1, \dots, x_k \mapsto w_k] \\
        =   & \mathbf{case}\hspace{2mm} w_1, \dots, w_k, v_1,\dots,v_{i-1},v_{i+1},\dots,v_n \hspace{2mm} \mathbf{of} \hspace{1mm} \left(\begin{bmatrix} S \end{bmatrix}[x_1 \mapsto w_1, \dots, x_k \mapsto w_k] \right) \end{array}$\\
        the reduction $\mathcal{S}(i,\mathcal{C}^{k}(x_1,\ldots, x_{k}),\text{\ref{eq:case:input}})[x_1 \mapsto w_1, \dots, x_k \mapsto w_k] \to e$, for some $e$. \vspace{2mm}\\
        First, we can always give the monotone function $f: \{1,\dots,\tilde{m}\} \to \{1,\dots,m\}$, that maps the number $j$ of a non-default row in $\begin{bmatrix} S \end{bmatrix}$, to the number $f(j)$ of the row in \ref{eq:case:input}, which gave rise to the addition of row $j$ in $\begin{bmatrix} S \end{bmatrix}$. Because the default clause in \ref{eq:case:input} contributes exactly itself -- a default clause -- to $\begin{bmatrix} S \end{bmatrix}$, the range of $f$ excludes $m+1$.\\
        Second, we denote the matrix $\begin{bmatrix} S \end{bmatrix}$ as a $(\tilde{m} + 1) \times (k + n - 1)$ matrix that we know by definition to be of the following shape.
        $$
        \begin{bmatrix} S \end{bmatrix} =
        \begin{bmatrix} 
            \overline{K}^1_1 \dots  \overline{K}^1_k & \overline{D}^{f(1)}_1 \dots \overline{D}^{f(1)}_{i-1},\overline{D}^{f(1)}_{i+1} \dots \overline{D}^{f(1)}_n & \To & e^{f(1)}\pi^1\\
            \vdots & \vdots & & \vdots \\
            \overline{K}^{\tilde{m}}_1 \dots \overline{K}^{\tilde{m}}_k & \overline{D}^{f(\tilde{m})}_1 \dots \overline{D}^{f(\tilde{m})}_{i-1},\overline{D}^{f(\tilde{m})}_{i+1} \dots \overline{D}^{f(\tilde{m})}_n & \To & e^{f(\tilde{m})}\pi^{\tilde{m}} \\
            \textbf{default} & & \To & e_d
        \end{bmatrix}
        $$
        There are two possible reduction rules, that could give rise to an evaluation to $e$.
        \begin{enumerate}
            \item 
            \vspace{2mm} $e = \left(e^{f(j)}\pi^j\right)\tau_1\dots\tau_k \sigma_1\dots\sigma_{i-1}\sigma_{i+1}\dots\sigma_n$ \\
            for a match with the $j$-th non-default row of $\left(\begin{bmatrix} S \end{bmatrix}[x_1 \mapsto w_1, \dots, x_k \mapsto w_k] \right)$:\\
            $\matches{\overline{K}^j_1[x_1 \mapsto w_1, \dots, x_k \mapsto w_k],\dots,\overline{K}^j_k[x_1 \mapsto w_1, \dots, x_k \mapsto w_k],\\
            \overline{D}^{f(j)}_1, \dots, \overline{D}^{f(j)}_{i-1},\overline{D}^{f(j)}_{i+1}, \dots, \overline{D}^{f(j)}_n}{\\
            w_1, \dots, w_k, v_1,\dots,v_{i-1},v_{i+1},\dots,v_n }{
            (\tau_1,\dots,\tau_k, \sigma_1, \dots, \sigma_{i-1},\sigma_{i+1},\dots,\sigma_n)}$ \vspace{2mm}\\
            To show: $\ref{eq:case:input} \to \left(e^{f(j)}\pi^j\right)\tau_1\dots\tau_k \sigma_1\dots\sigma_{i-1}\sigma_{i+1}\dots\sigma_n$ \vspace{2mm} \\
            We conduct an induction over the inclusion rules of \cref{subsubsec:compilation:branch:specialization} \ref{spec-a}-\ref{spec-c} that could have led to the inclusion of row $j$ in $\begin{bmatrix} S \end{bmatrix}$. As before, we focus on the case of rule \ref{spec-b} which involves negation and refer to \cite{Maranget2008} for the usual cases. \vspace{2mm}\\
            Assume row $j$ in $\begin{bmatrix} S \end{bmatrix}$ was added by contribution of row $f(j)$ in \ref{eq:case:input} such that $\mathbf{\overline{D}_i^{f(j)} = \patAnd{\{ y_1, \ldots y_s\}}{\patNeg{\{ \mathcal{C}_1, \ldots, \mathcal{C}_t \} }}}$ with $\mathcal{C}^k \not\in \{ \mathcal{C}_1, \ldots, \mathcal{C}_t \}$. \vspace{2mm}\\
            First, we can derive the following judgment of row $f(j)$ in \ref{eq:case:input} as matching\\
            $\matches{\overline{D}^{f(j)}}{\vec{v}}{(\sigma_1,\dots,\sigma_{i-1},[\{ y_1, \ldots y_s\} \mapsto v_i],\sigma_{i+1},\dots,\sigma_n)}$, because 
            \begin{itemize}
                \item By assumption, $\forall \ x \in (\{1,\dots,n\}\setminus\{i\}), \ \matches{\overline{D}_x^{f(j)}}{v_x}{\sigma_x}$
                \item For $\overline{D}_i^{f(j)} = \patAnd{\{ y_1, \ldots y_s\}}{\patNeg{\{ \mathcal{C}_1, \ldots, \mathcal{C}_t \} }}$ and $v_i = C^k(w_1,\dots,w_k)$ we can derive $\matches{\overline{D}_i^{f(j)}}{v_i}{[\{ y_1, \ldots y_s\} \mapsto v_i]}$, because $\mathcal{C}^k \not\in \{ \mathcal{C}_1, \ldots, \mathcal{C}_t \}$.
            \end{itemize} \vspace{2mm}

            Second, by \textsc{E-MC} we can now reduce our goal\\ 
            $\ref{eq:case:input} \to e^{f(j)}\sigma_1\dots\sigma_{i-1}[\{ y_1, \ldots y_s\} \mapsto v_i]\sigma_{i+1}\dots\sigma_n$. \vspace{2mm}\\
            $\begin{array}{lll} \hspace{-1.5mm}\text{To show:} && e^{f(j)}\sigma_1\dots\sigma_{i-1}[\{ y_1, \ldots y_s\} \mapsto v_i]\sigma_{i+1}\dots\sigma_n \\ & = & \left(e^{f(j)}\pi^j\right)\tau_1\dots\tau_k \sigma_1\dots\sigma_{i-1}\sigma_{i+1}\dots\sigma_n\end{array}$\vspace{2mm}\\
            Note, that our matching row $j$ in $\left(\begin{bmatrix} S \end{bmatrix}[x_1 \mapsto w_1, \dots, x_k \mapsto w_k] \right)$ equates to \\
            $\begin{bmatrix}
                \ \patWildcard \ \ \dots \ \ \patWildcard \hspace{2mm} 
                \overline{D}_{1}^{f(j)} \dots \overline{D}_{i-1}^{f(j)} \overline{D}_{i+1}^{f(j)} \dots \ \overline{D}_{n}^{f(j)}
                \To e^{f(j)} [\{ y_1, \ldots y_s\} \mapsto v_i]
            \end{bmatrix}$. \vspace{2mm}\\
            First of all, we see that $\pi^j = [\{ y_1, \ldots y_s\} \mapsto v_i]$. Second, we know that the substitutions $\tau_1,\dots,\tau_k$ must be empty, because having $\forall \ x \in \{1,\dots,k\}, \ \matches{\patWildcard}{w_x}{\tau_x}$ necessitates $\tau_x = []$. 
            Third, we can -- as before -- permute the application of the substitutions. \vspace{2mm}\\
            $\begin{array}{lll}
                \hspace{-1.5mm}\text{In sum, this gives: }
                &  &   e^{f(j)}\sigma_1\dots\sigma_{i-1}[\{ y_1, \ldots y_s\} \mapsto v_i]\sigma_{i+1}\dots\sigma_n \\
                &= & e^{f(j)}[\{ y_1, \ldots y_s\} \mapsto v_i]\sigma_1\dots\sigma_{i-1}\sigma_{i+1}\dots\sigma_n \\
                &= & (e^{f(j)}\pi^j)\sigma_1\dots\sigma_{i-1}\sigma_{i+1}\dots\sigma_n \\
                &= & (e^{f(j)}\pi^j)[]\dots[]\sigma_1\dots\sigma_{i-1}\sigma_{i+1}\dots\sigma_n \\
                &= & (e^{f(j)}\pi^j)\tau_1\dots \tau_k\sigma_1\dots\sigma_{i-1}\sigma_{i+1}\dots\sigma_n
                \end{array}$ \vspace{2mm}
            \item $e = e_d$ and no non-default row of $\left(\begin{bmatrix} S \end{bmatrix}[x_1 \mapsto w_1, \dots, x_k \mapsto w_k] \right)$ matches, i.e.: \\
            $\forall \ j \in \{1,\dots,\tilde{m}\},$ there is a $\vec{\sigma}^j$ such that \\
            $\matchesNot{\overline{K}^j_1[x_1 \mapsto w_1, \dots, x_k \mapsto w_k],\dots,\overline{K}^j_k[x_1 \mapsto w_1, \dots, x_k \mapsto w_k],\\
            \overline{D}^{f(j)}_1, \dots, \overline{D}^{f(j)}_{i-1},\overline{D}^{f(j)}_{i+1}, \dots, \overline{D}^{f(j)}_n}{w_1, \dots, w_k, v_1,\dots,v_{i-1},v_{i+1},\dots,v_n }{\vec{\sigma}^j}$ \vspace{2mm}\\
            We need to show that no non-default row of \ref{eq:case:input} is a matching row. For this, we pick (if such exists) any non-default row $u$ in \ref{eq:case:input} and induct over $\overline{D}_i^u$. Again, we only look at the negation-case and refer to \cite{Maranget2008} for the other non-trivial cases. \vspace{2mm} \\
            Assume $\mathbf{\overline{D}_i^u = \patAnd{\{ y_1, \ldots y_s\}}{\patNeg{\{ \mathcal{C}_1, \ldots, \mathcal{C}_t \} }}}$. \vspace{2mm}
            \begin{itemize}
                \item If $C^k \in \{C_1,\dots,C_t\}$, then we know that we can trivially derive $\matchesNot{\overline{D}_i^u}{v_i}{[]}$, hence using $\vec{\sigma}^j$ we derive $\matchesNot{\overline{D}^u}{\vec{v}}{(\sigma_{k+1}^j,\dots,\sigma_{k + i-1}^j,[],\sigma_{k + i+1}^j,\dots,\sigma_{k+n})}$. \vspace{2mm}
                \item More interestingly, assume $C^k \not\in \{C_1,\dots,C_t\}$. \vspace{2mm}\\
                Then, by \ref{subsubsec:compilation:branch:specialization} \ref{spec-b} row $u$ of \ref{eq:case:input} contributed a single row $j$ with $u = f(j)$ of the shape
                $\begin{bmatrix}
                \ \patWildcard \ \ \dots \ \ \patWildcard \hspace{2mm} 
                \overline{D}_{1}^{u} \dots \overline{D}_{i-1}^{u} \overline{D}_{i+1}^{u} \dots \ \overline{D}_{n}^{u}
                \To e^{u} [\{ y_1, \ldots y_s\} \mapsto v_i]
                \end{bmatrix}$. \vspace{2mm}\\
                To show: There exists a substitution vector $\vec{\tau}$, such that $\matchesNot{\overline{D}^u}{\vec{v}}{\vec{\tau}}$\vspace{2mm}\\
                We proceed with a proof of negation showing that there cannot be any $\vec{\tau}$, s.t. \\
                $\matches{\overline{D}^u}{\vec{v}}{\vec{\tau}}$ and then conclude our goal using completeness (\cref{thm:matching-complete}). \vspace{2mm}\\
                Now, assume row $u$ in \ref{eq:case:input} as a matching row $\matches{\overline{D}^u}{\vec{v}}{\vec{\tau}}$. \vspace{2mm} \\
                Then, row j of $\left(\begin{bmatrix} S \end{bmatrix}[x_1 \mapsto w_1, \dots, x_k \mapsto w_k] \right)$ must be a matching row by \\
                $\matches{(\patWildcard,\dots,\patWildcard,\overline{D}^u_{1},\dots,\overline{D}^u_{i-1},\overline{D}^u_{i+1},\dots,\overline{D}^u_{n})}{(w_1,\dots,w_k,v_1,\dots,v_{i-1},v_{i+1},\dots,v_n)}{\\
                ([],\dots,[],\tau_1,\dots,\tau_{i-1},\tau_{i+1},\dots,\tau_n)}$, because $\forall \ x \in \{1,\dots,k\}, \ \matches{\patWildcard}{w_x}{[]}$. \vspace{2mm}\\
                But using soundness (\cref{thm:matching-sound}) this contradicts the initial assumption.
            \end{itemize}
        \end{enumerate}
    \end{enumerate}
\end{proof}

\begin{lemma}[Specialization/Default preserves wellformedness]
    \label{lem:compilation:subproblem-wf}
    Consider a multi-column pattern match of the normalized shape \ref{eq:case:input} and denote the $i$-th scrutinee of \ref{eq:case:input} as $v_i = C^k(w_1,\dots,w_k)$. 
    \begin{enumerate}
        \item If $\wf{\ref{eq:case:input}}$, then $\wf{\mathcal{S}(i,\mathcal{C}^{k}(x_1,\ldots, x_{k}),\text{\ref{eq:case:input}})}$
        \item If $\wf{\ref{eq:case:input}}$, then $\wf{\mathcal{D}(i,\mathcal{H}(\overline{D}_i), \text{\ref{eq:case:input}})}$
    \end{enumerate}
\end{lemma}
\begin{proof} \ 
    \begin{enumerate}
    \item
        $\mathcal{S}(i,\mathcal{C}^{k}(x_1,\ldots, x_{k}),\text{\ref{eq:case:input}}) \ = \ \mathbf{case}\hspace{2mm} 
        x_1, \dots, x_k, v_1,\dots,v_{i-1},v_{i+1},\dots,v_n \hspace{2mm} \mathbf{of} \hspace{1mm} \begin{bmatrix} S \end{bmatrix}$ \vspace{2mm}\\
        For any row of $\begin{bmatrix} S \end{bmatrix}$ we can show that its patterns are deterministic, positively linear, and do not share free variables under even negation by inducting on its inclusion rules \cref{subsubsec:compilation:branch:specialization}. \vspace{2mm}\\
        To show: Two rows $j \neq j'$ of $\begin{bmatrix} S \end{bmatrix}$ share a column $u$ for which their patterns are disjoint. \vspace{2mm}\\
        This is shown by a nested induction on the inclusion rules \cref{subsubsec:compilation:branch:specialization} for $j$ and $j'$. \vspace{2mm}
    \item Entirely analogous to 1.
    \end{enumerate}
\end{proof}

\begin{lemma}[Compilation preserves wellformedness]
    \label{lem:compilation:preservation-wf}
    For any multi-column pattern match of the normalized shape \ref{eq:case:input}, it holds that if \wf{\ref{eq:case:input}} then \wf{compile(\ref{eq:case:input})}.
\end{lemma}
\begin{proof}
    By induction over the recursive definition of $compile(\ref{eq:case:input})$ and use of \cref{lem:compilation:subproblem-wf}.
\end{proof}

\begin{lemma}[Determinism for wellformed Multi-Column Cases]
    \label{lem:wf:determinism:multicol}
    \cref{thm:wf:deterministic} extends to wellformed multi-column pattern matches, i.e. for $\wf{e,e_1,e_2}$, if $e \to e_1$ and $e \to e_2$, then $e_1 = e_2$.
\end{lemma}
\begin{proof} \ \\
    Let $e = \mathbf{case}\ \vec{v} \ \mathbf{of} \ 
        \begin{bmatrix}
			P_1^1 	& \ldots 	& P_n^1 &\To e^1 \\
			& \vdots	& 				\\
			P_1^m	& \ldots	& P_n^m &\To e^m \\
			&\defaultClause & &\To e_d
		\end{bmatrix} \vspace{2mm}$ be wellformed 
        \\ and assume $e \to e_1$ and $e \to e_2$. \vspace{2mm}\\
        If either one was derived by \textsc{E-MC} (a matching non-default row), so must have been the other. So assume $e_1$ got derived from a matching row $j$ and $e_2$ from a matching row $j'$. Trivially, these rows must overlap as witnessed by $\vec{v}$. If $j \neq j'$, this contradicts our assumption of wellformedness. Hence, $j = j'$ and $e_1 = e_2$. \vspace{2mm}\\
        If both of them are derived by \textsc{E-DC}, we have $e_1 = e_d = e_2$.
\end{proof}

\begin{theorem}\label{thm:correctness-proof}
    On a wellformed multi-column pattern match of the normalized shape \ref{eq:case:input}, the compilation algorithm 
	preserves semantics. 
    $$\semantic{\ref{eq:case:input}} \equiv \semantic{\textit{compile}(\ref{eq:case:input})}$$
\end{theorem}
\begin{proof}
    We perform an induction over the recursive definition of $compile(\ref{eq:case:input})$ as was already suggested by \cite{Maranget2008}.
    We want to show that
    $
    \ref{eq:case:input} \to^* v \textit{ only if } \textit{compile}(\ref{eq:case:input}) \to^* v
    $.
    \begin{enumerate}
        \item \textbf{Default} \vspace{2mm}\\
        $\begin{array}{l}
            \ref{eq:case:input} = 
            \mathbf{case}\ v_1,\ldots, v_n\ \mathbf{of}\
                \begin{bmatrix}
                    &\defaultClause & &\To e_d
                \end{bmatrix} \vspace{2mm}\\
            compile(\ref{eq:case:input}) = e_d
        \end{array}$\vspace{2mm}\\
        As we have no choice but to reduce by rule \textsc{E-MD} $\ref{eq:case:input} \to e_d$, \\
        it is $\ref{eq:case:input} \to^* v$ only if $compile(\ref{eq:case:input}) \to^* v$ \vspace{2mm}
        \item \textbf{Simple} \\
        $\begin{array}{l}
            \ref{eq:case:input} = 
            \mathbf{case}\ v_1,\ldots, v_n\ \mathbf{of}\
                \begin{bmatrix}
                    \patAnd{Y^1}{\patNeg{\{\}}} 	& \ldots 	& \patAnd{Y^n}{\patNeg{\{\}}}   &\To e^1 \\
                    & \vdots	& 				\\
                    \overline{D}_1^m	& \ldots	& \overline{D}_n^m &\To e^m \\
                    &\defaultClause & &\To e_d
                \end{bmatrix} \vspace{2mm}\\
                \text{where } Y^i = \{y_1,\dots,y_{t_i}\}, \text{ a set of variables for some }t_i. \\
                compile(\ref{eq:case:input}) = e^1[Y^1 \mapsto v_1]\dots[Y^n \mapsto v_n]
            \end{array}$ \vspace{2mm}\\
        The first row of \ref{eq:case:input} matches the scrutinees $\matches{\overline{D}^1}{\vec{v}}{([Y^1 \mapsto v_1],\dots,[Y^n \mapsto v_n])}$, which gives a single-step reduction by rule \textsc{E-MC} of $\ref{eq:case:input} \to e^1[Y^1 \mapsto v_1]\dots[Y^n \mapsto v_n]$.
        By \cref{lem:wf:determinism:multicol}, we can conclude that this reduction is the only possible reduction, hence $\ref{eq:case:input} \to^* v$ only if $e^1[Y^1 \mapsto v_1]\dots[Y^n \mapsto v_n] \to^* v$. \vspace{2mm}
        \item \textbf{Branch}\\
        $\begin{array}{l}
            \ref{eq:case:input} = 
            \mathbf{case}\ v_1,\ldots, v_n\ \mathbf{of}\
                \begin{bmatrix}
                    \overline{D}_1^1 	& \ldots 	& \overline{D}_n^1 &\To e^1 \\
                    & \vdots	& 				\\
                    \overline{D}_1^m	& \ldots	& \overline{D}_n^m &\To e^m \\
                    &\defaultClause & &\To e_d
                \end{bmatrix} \vspace{2mm}\\
                \text{with the set of head constructors } \mathcal{H}(\overline{D}_i) = \{ \mathcal{C}^{n_1}, \ldots, \mathcal{C}^{n_z} \},\\
                compile(\ref{eq:case:input}) = \\
                \mathbf{case}\ v_i\ \mathbf{of}\
                \begin{bmatrix}
                    \mathcal{C}^{n_1}(x_1^{1},\ldots, x_{n_1}^{1}) 	&\To& \mathit{compile}(\mathcal{S}(i,\mathcal{C}^{n_1}(x_1^{1},\ldots, x_{n_1}^{1}),\text{\ref{eq:case:input}})) \\
                    \vdots	& 	&  \vdots			\\
                    \mathcal{C}^{n_z}(x_1^{z},\ldots, x_{n_z}^{z}) &\To& \mathit{compile}(\mathcal{S}(i,\mathcal{C}^{n_z}(x_1^{z},\ldots, x_{n_z}^{z}),\text{\ref{eq:case:input}})) \\
                    \defaultClause &\To& \mathit{compile}(\mathcal{D}(i,\{ \mathcal{C}^{n_1},\ldots, \mathcal{C}^{n_z} \}, \text{\ref{eq:case:input}}))
                \end{bmatrix}
        \end{array}$ \vspace{2mm}\\
        We can always reduce $compile(\ref{eq:case:input})$ by either \textsc{E-MC} or \textsc{E-DC}. \\
        \begin{itemize}
            \item $compile(\ref{eq:case:input}) \to \mathit{compile}(\mathcal{S}(i,\mathcal{C}^{n_j}(x_1^{j},\ldots, x_{n_j}^{j}),\text{\ref{eq:case:input}}))\pi$ \\
            for a matching row 
            $\matches{\mathcal{C}^{n_j}(x_1^{j},\ldots, x_{n_j}^{j})}{v_i}{\pi}$.\\
            It must necessarily be $v_i = {C}^{n_j}(w_1,\dots, w_{n_j})$ and $\pi = [x_1^{j} \mapsto w_1,\ldots, x_{n_j}^{j} \mapsto w_{n_j}]$. \vspace{2mm}\\
            By \cref{lem:compilation:subproblem-reduction}, we have for single-step reductions\\
            $\ref{eq:case:input} \to e \text{ iff }  \mathcal{S}(i,\mathcal{C}^{n_j}(x_1^{j},\ldots, x_{n_j}^{j}),\text{\ref{eq:case:input}})\pi \to e$ \vspace{2mm}\\
            which, of course, gives \\
            $\ref{eq:case:input} \to^* v \text{ iff }   \mathcal{S}(i,\mathcal{C}^{n_j}(x_1^{j},\ldots, x_{n_j}^{j}),\text{\ref{eq:case:input}})\pi \to^* v$ \vspace{2mm}\\
            Because $\mathcal{S}(i,\mathcal{C}^{n_j}(x_1^{j},\ldots, x_{n_j}^{j}),\text{\ref{eq:case:input}})$ is wellformed (see \cref{lem:compilation:subproblem-wf}), and given that variable substitution does not impact wellformedness, we can instantiate our induction hypothesis with $\mathcal{S}(i,\mathcal{C}^{n_j}(x_1^{j},\ldots, x_{n_j}^{j}),\text{\ref{eq:case:input}})\pi$\\
            $\mathcal{S}(i,\mathcal{C}^{n_j}(x_1^{j},\ldots, x_{n_j}^{j}),\text{\ref{eq:case:input}})\pi \to^* v
            \text{ iff } 
            compile\left(\mathcal{S}(i,\mathcal{C}^{n_j}(x_1^{j},\ldots, x_{n_j}^{j}),\text{\ref{eq:case:input}})\pi\right) \to^* v$.\vspace{2mm}\\
            It does not matter, whether we apply the substitution $\pi$ before or after compilation \\
            $compile\left(\mathcal{S}(i,\mathcal{C}^{n_j}(x_1^{j},\ldots, x_{n_j}^{j}),\text{\ref{eq:case:input}})\pi\right) = compile\left(\mathcal{S}(i,\mathcal{C}^{n_j}(x_1^{j},\ldots, x_{n_j}^{j}),\text{\ref{eq:case:input}})\right)\pi$.\vspace{2mm}\\
            As $compile(\ref{eq:case:input}) \to \mathit{compile}(\mathcal{S}(i,\mathcal{C}^{n_j}(x_1^{j},\ldots, x_{n_j}^{j}),\text{\ref{eq:case:input}}))\pi$, we can infer via \ref{lem:wf:determinism:multicol}\\
            $compile\left(\mathcal{S}(i,\mathcal{C}^{n_j}(x_1^{j},\ldots, x_{n_j}^{j}),\text{\ref{eq:case:input}})\right)\pi \to^* v \text{ iff } compile(\ref{eq:case:input}) \to^* v$. \vspace{2mm}
            \item Entirely analogous.
        \end{itemize}
    \end{enumerate}
\end{proof}

Note that although this proof assumes the order-independent semantics of this paper, we gave the compilation algorithm in the usual order-sensitive style, which is still \enquote{backwards-compatible} with correctness under first-match semantics.

\end{document}